\documentclass[11pt]{article}%
\pdfoutput=1 %See https://tex.stackexchange.com/questions/182523/publishing-to-arxiv-with-hyperlinks for more information

\usepackage{amsmath,amsfonts,amssymb}
\usepackage{amsthm}
\usepackage{mathtools}
\usepackage{graphicx}
\usepackage[margin=1in]{geometry} % makes smaller bottom margins than fullpage package, but otherwise the same
\usepackage[pagebackref,linktocpage]{hyperref}
\hypersetup{
    colorlinks=true, % false: boxed links; true: colored links
    linkcolor=blue, % color of internal links
    citecolor=blue, % color of links to bibliography
    urlcolor=blue % color of external links
}
\usepackage{microtype}
\usepackage{thm-restate}
\usepackage[capitalize,nameinlink]{cleveref} 
\crefname{algocf}{Algorithm}{Algorithms}

\usepackage{tikz}
\usepackage{pgfplots}
\pgfplotsset{compat=1.13}
\usepackage{wrapfig}
\usepackage{enumerate}
\usepackage{braket}
\usepackage{qcircuit}
\usepackage{circledsteps} %Text in circles
\usepackage[linesnumbered,ruled,vlined]{algorithm2e}
\SetKwInput{KwInput}{Input}                %
\SetKwInput{KwOutput}{Output}
\SetKwFor{RepTimes}{repeat}{times}{end}

\renewcommand{\backref}[1]{}

\renewcommand{\backrefalt}[4]{%
\ifcase #1 %
\or
[p.\ #2]%
\else
[pp.\ #2]%
\fi}

\newtheorem{theorem}{Theorem}

\newtheorem{claim}[theorem]{Claim}

\newtheorem{conjecture}[theorem]{Conjecture}
\newtheorem{corollary}[theorem]{Corollary}

\newtheorem{definition}[theorem]{Definition}

\newtheorem{fact}[theorem]{Fact}
\newtheorem{lemma}[theorem]{Lemma}

\newtheorem{problem}[theorem]{Problem}
\newtheorem{proposition}[theorem]{Proposition}

\newcommand{\eps}{\varepsilon}
\newcommand{\E}{\mathop{\mathbb{E}}}
\newcommand{\Naturals}{\mathbb{N}}

\newcommand{\poly}{\mathrm{poly}}
\newcommand{\polylog}{\mathrm{polylog}}
\newcommand{\quasipoly}{\mathrm{quasipoly}}
\newcommand{\Sipser}{\normalfont\textsc{Sipser}}
\newcommand{\AND}{\normalfont\textsc{AND}} %\textsc{And}
\newcommand{\OR}{\normalfont\textsc{OR}}
\newcommand{\NOT}{\normalfont\textsc{NOT}}
\newcommand{\PARITY}{\normalfont\textsc{Parity}}
\newcommand{\XOR}{\normalfont\textsc{XOR}}
\newcommand{\Forrelation}{\normalfont\textsc{Forrelation}}
\newcommand{\Dom}{\mathrm{Dom}}

\DeclareMathOperator{\disagr}{\mathrm{disagr}}
\DeclareMathOperator{\D}{\mathsf{D}}
\DeclareMathOperator{\C}{\mathsf{C}}
\DeclareMathOperator{\Q}{\mathsf{Q}}
\DeclareMathOperator{\s}{\mathsf{s}}
\DeclareMathOperator{\bs}{\mathsf{bs}}
\DeclareMathOperator{\QMA}{\mathsf{QMA}}
\DeclareMathOperator{\QC}{\mathsf{QC}}

\begin{document}

\title{The Acrobatics of $\mathsf{BQP}$}
\author{Scott Aaronson\thanks{University of Texas at Austin. \ Email:
\texttt{scott@scottaaronson.com}. \ Supported by a Vannevar Bush Fellowship from the US
Department of Defense, a Simons Investigator Award, and the Simons
\textquotedblleft It from Qubit\textquotedblright\ collaboration.}
\and DeVon Ingram\thanks{University of Chicago. Email: \texttt{dingram@uchicago.edu}. \ Supported by an NSF Graduate Research Fellowship.}
\and William Kretschmer\thanks{University of Texas at Austin. \ Email:
\texttt{kretsch@cs.utexas.edu}. \ Supported by an NDSEG Fellowship.}}
\date{}
\maketitle

\begin{abstract}
One can fix the randomness used by a randomized algorithm, but there is no analogous notion of fixing the quantumness used by a quantum algorithm.  Underscoring this fundamental difference, we show that, in the black-box setting, the behavior of quantum
polynomial-time ($\mathsf{BQP}$) can be remarkably decoupled from that of
classical complexity classes like $\mathsf{NP}$. \ Specifically:

\begin{itemize}
\item There exists an oracle relative to
which $\mathsf{NP}^{\mathsf{BQP}}\not \subset \mathsf{BQP}^{\mathsf{PH}}$, resolving a 2005 problem of Fortnow. 
%Interpreted another way, we show that $\mathsf{AC^0}$ circuits cannot perform useful homomorphic encryption on instances of the \textsc{Forrelation} problem.
As a corollary, there exists an oracle relative to which $\mathsf{P} = \mathsf{NP}$ but $\mathsf{BQP} \neq \mathsf{QCMA}$.
\item Conversely, there exists an oracle relative to
which $\mathsf{BQP}^{\mathsf{NP}}\not \subset \mathsf{PH}^{\mathsf{BQP}}$.
\item Relative to a random oracle, $\mathsf{PP} = \mathsf{PostBQP}$ is not contained in the ``$\mathsf{QMA}$ hierarchy'' $\mathsf{QMA}^{\mathsf{QMA}^{\mathsf{QMA}^{\cdots}}}$. %This result shows that there is no black-box quantum analogue of Stockmeyer's approximate counting algorithm.
\item Relative to a random oracle, $\mathsf{\Sigma}_{k+1}^\mathsf{P} \not\subset \mathsf{BQP}^{\mathsf{\Sigma}_{k}^\mathsf{P}}$ for every $k$.
\item There exists an oracle relative to which $\mathsf{BQP} = \mathsf{P^{\# P}}$\ and yet $\mathsf{PH}$\ is infinite. \ (By contrast, relative to all oracles, if
$\mathsf{NP}\subseteq\mathsf{BPP}$, then $\mathsf{PH}$\ collapses.)
\item There exists an oracle relative to which $\mathsf{P}=\mathsf{NP} \neq \mathsf{BQP}=\mathsf{P}^{\mathsf{\#P}}$.
\end{itemize}

To achieve these results, we build on the 2018 achievement by Raz and Tal of an oracle relative to which $\mathsf{BQP}\not \subset \mathsf{PH}$, and associated results about the \textsc{Forrelation} problem.  We also introduce new tools that might be of independent interest. These include a ``quantum-aware'' version of the random restriction method, a concentration theorem for the block sensitivity of $\mathsf{AC^0}$ circuits, and a (provable) analogue of the Aaronson-Ambainis Conjecture for sparse oracles.
\end{abstract}
\newpage
\tableofcontents
\newpage

\section{Introduction\label{INTRO}}

The complexity-theoretic study of quantum computation is often dated from 1993, when Bernstein and Vazirani \cite{BV97} defined $\mathsf{BQP}$, or Bounded-Error Quantum Polynomial-Time: the class of languages that admit efficient quantum algorithms.  Then as now, a central concern was how $\mathsf{BQP}$ relates to classical complexity classes, such as $\mathsf{P}$, $\mathsf{NP}$, and $\mathsf{PH}$.  Among the countless questions that one could raise here, let us single out three as especially fundamental:
\begin{enumerate}[(1)]
\item Can quantum computers efficiently solve any problems that classical computers cannot?  In other words, does $\mathsf{BPP}= \mathsf{BQP}$?
\item Can quantum computers solve $\mathsf{NP}$-complete problems in polynomial time?  In other words, is $\mathsf{NP} \subseteq \mathsf{BQP}$?
\item What is the best classical upper bound on the power of quantum computation?  Is $\mathsf{BQP} \subseteq \mathsf{NP}$?  Is $\mathsf{BQP} \subseteq \mathsf{PH}$?
\end{enumerate}

Three decades later, all three of these still stand as defining
questions of the field.  Nevertheless, from the early 2000s onwards, it became rare for work in quantum computing theory to address any of these questions directly, perhaps simply because it became too hard to say anything new about them.  A major recent exception was the seminal work of Raz and Tal \cite{RT19}, who gave an oracle relative to which $\mathsf{BQP}\not\subset\mathsf{PH}$, by completing a program proposed by one of us \cite{Aar10}.  In this paper, we take the Raz-Tal breakthrough as a starting point.  Using it, together with new tools that we develop, we manage to prove many new theorems about the power of $\mathsf{BQP}$---at least in the black-box setting where much of our knowledge of quantum
algorithms resides.

Before discussing the black-box setting or Raz-Tal, though, let's
start by reviewing what is known in general about $\mathsf{BQP}$. Bernstein and Vazirani \cite{BV97} showed that $\mathsf{BPP} \subseteq \mathsf{BQP} \subseteq \mathsf{P^{\# P}}$, and Adleman, DeMarrais, and Huang \cite{ADH97} improved the upper bound to $\mathsf{BQP} \subseteq \mathsf{PP}$, giving us the following chain of inclusions:
\[
\mathsf{P} \subseteq \mathsf{BPP} \subseteq \mathsf{BQP} \subseteq \mathsf{PP} \subseteq \mathsf{P^{\# P}} \subseteq \mathsf{PSPACE} \subseteq \mathsf{EXP}.
\]

Fortnow and Rogers \cite{FR98} slightly strengthened the inclusion $\mathsf{BQP} \subseteq \mathsf{PP}$, to show for example that $\mathsf{PP^{BQP}} = \mathsf{PP}$.  This complemented the result of Bennett, Bernstein, Brassard, and Vazirani \cite{BBBV97} that $\mathsf{BQP^{BQP}} = \mathsf{BQP}$: that is, $\mathsf{BQP}$ is ``self-low,'' or ``the $\mathsf{BQP}$ hierarchy collapses to $\mathsf{BQP}$.''

\subsection{The Contrast with \texorpdfstring{$\mathsf{BPP}$}{BPP}}
\label{sec:contrast_with_bpp}

Meanwhile, though, the relationships between $\mathsf{BQP}$ and complexity classes like $\mathsf{NP}$, $\mathsf{PH}$, and $\mathsf{P/poly}$ have remained mysterious.  Besides the
fundamental questions mentioned above---is $\mathsf{NP} \subseteq \mathsf{BQP}$? is $\mathsf{BQP} \subseteq \mathsf{NP}$? is $\mathsf{BQP} \subseteq \mathsf{PH}$?---one could ask other questions:

\begin{enumerate}[(i)]
\item In a 2005 blog post, Fortnow \cite{For05} raised the question of whether $\mathsf{NP^{BQP}} \subseteq \mathsf{BQP^{NP}}$. Do we even have $\mathsf{NP^{BQP}} \subseteq \mathsf{BQP^{PH}}$? I.e., when quantum computation is combined with classical nondeterminism, how does the order of combination matter?
\item What about the converse: is $\mathsf{BQP^{NP}} \subseteq \mathsf{PH^{BQP}}$? 

\item Suppose $\mathsf{NP} \subseteq \mathsf{BQP}$.  Does it follow that $\mathsf{PH} \subseteq \mathsf{BQP}$ as well?

\item Suppose $\mathsf{NP} \subseteq \mathsf{BQP}$.  Does it follow that $\mathsf{PH}$ collapses?

\item Is $\mathsf{BQP} \subset \mathsf{P/poly}$?

\item Suppose $\mathsf{P} = \mathsf{NP}$.  Does it follow that $\mathsf{BQP}$ is ``small'' (say, not equal to $\mathsf{EXP}$)?

\item Suppose $\mathsf{P} = \mathsf{NP}$.  Does it follow that $\mathsf{BQP} = \mathsf{QCMA}$, where $\mathsf{QCMA}$ (Quantum Classical Merlin Arthur) is the analogue of $\mathsf{NP}$ with a $\mathsf{BQP}$ verifier?
\end{enumerate}

What is particularly noteworthy about the questions above is that, if we replace $\mathsf{BQP}$ by $\mathsf{BPP}$, then positive answers are known to all of them:

\begin{enumerate}[(i)]
\item $\mathsf{NP}^\mathsf{BPP} \subseteq \mathsf{AM} \subseteq \mathsf{BPP}^\mathsf{NP}$.
\item $\mathsf{BPP}^\mathsf{NP} \subseteq \mathsf{PH} = \mathsf{PH}^\mathsf{BPP}$.
\item If $\mathsf{NP} \subseteq \mathsf{BPP}$, then $\mathsf{PH} = \mathsf{BPP}$---this is sometimes given as a homework exercise in complexity theory courses, and also follows from (i).

\item If $\mathsf{NP} \subseteq \mathsf{BPP}$, then $\mathsf{PH} = \mathsf{\Sigma}_2^\mathsf{P}$---this follows from (iii) and the Sipser-Lautemann Theorem \cite{Sip83,Lau83}.

\item $\mathsf{BPP} \subset \mathsf{P/poly}$ is Adleman's Theorem \cite{Adl78}.

\item If $\mathsf{P}=\mathsf{NP}$, then $\mathsf{P}=\mathsf{BPP}$ and hence $\mathsf{BPP}\neq\mathsf{EXP}$, by the time hierarchy theorem.

\item If $\mathsf{P}=\mathsf{NP}$, then of course $\mathsf{BPP} = \mathsf{MA}$.
\end{enumerate}

So what is it that distinguishes $\mathsf{BPP}$ from $\mathsf{BQP}$ in these cases?  In all of the above examples, the answer turns out to be one of the fundamental properties of classical randomized algorithms: namely, that one can always ``pull the randomness out'' from such algorithms, viewing them as simply deterministic algorithms that take a uniform random string $r$ as an auxiliary input, in addition to their ``main'' input $x$.  This, in turn, enables one to play all sorts of tricks with such an algorithm $M(x,r)$---from using approximate counting to estimate the fraction of $r$'s that cause $M(x,r)$ to accept, to moving $r$ from inside to outside a quantifier, to hardwiring $r$ as advice.  By contrast, there is no analogous notion of ``pulling the randomness (or quantumness) out of a quantum algorithm.''  In quantum computation, randomness is just an intrinsic part of the model that rears its head at the \textit{end} (rather than the beginning) of a computation, when we take the squared absolute values of amplitudes to get probabilities.

This difference between randomized and quantum algorithms is crucial to the analysis of the so-called ``sampling-based quantum supremacy experiments''---for example, those recently carried out by Google \cite{AAB+19} and USTC \cite{ZWD+20}.  The theoretical foundations of these experiments were laid a decade ago, in the work of Aaronson and Arkhipov \cite{AA13} on $\textsc{BosonSampling}$, and (independently) Bremner, Jozsa, and Shepherd \cite{BJS10} on the commuting Hamiltonians or IQP model. Roughly speaking, the idea is that, by using a quantum computer, one can efficiently sample a probability distribution $\mathcal{D}$ over $n$-bit strings such that even \textit{estimating} the probabilities of the outcomes is a $\mathsf{\# P}$-hard problem.  Meanwhile, though, if there were a polynomial-time classical randomized algorithm $M(x,r)$ to sample from the same distribution $\mathcal{D}$, then one could use the ``pulling out $r$'' trick to estimate the probabilities of $M$'s outcomes in $\mathsf{PH}$.  But this would put $\mathsf{P}^\mathsf{\# P}$ into $\mathsf{PH}$, thereby collapsing $\mathsf{PH}$ by Toda's Theorem \cite{Tod91}.

More generally, with any of the apparent differences between quantum algorithms and classical randomized algorithms, the question is: how can we prove that the difference is genuine, that no trick will ever be discovered that makes $\mathsf{BQP}$ behave more like $\mathsf{BPP}$? For questions like whether $\mathsf{NP} \subseteq \mathsf{BQP}$ or whether $\mathsf{BQP} \subseteq \mathsf{NP}$, the hard truth here is that not only have we been unable to resolve these questions in the unrelativized world, we've been able to say little more about them than certain ``obvious'' implications.  For example, suppose $\mathsf{NP} \subseteq \mathsf{BQP}$ \textit{and} $\mathsf{BQP} \subseteq \mathsf{AM}$.  Then since $\mathsf{BQP}$ is closed under complement, we would also have $\mathsf{coNP} \subseteq \mathsf{BQP}$, and hence $\mathsf{coNP} \subseteq \mathsf{AM}$, which is known to imply a collapse of $\mathsf{PH}$ \cite{BHZ87}.  And thus, if $\mathsf{PH}$ is infinite, then either $\mathsf{NP} \not\subset \mathsf{BQP}$ \textit{or} $\mathsf{BQP} \not\subset \mathsf{AM}$.  How can we say anything more interesting and nontrivial?

\subsection{Relativization}
Since the work of Baker, Gill, and Solovay \cite{BGS75}, whenever complexity theorists were faced with an impasse like the one above, a central tool has been \textit{relativized} or \textit{black-box} complexity: in other words, studying what happens when all the complexity classes one cares about are fed some specially-constructed oracle.  Much like perturbation theory in physics, relativization lets us make well-defined progress even when the original questions we wanted to answer are out of reach.  It is well-known that relativization is an imperfect tool---the $\mathsf{IP}=\mathsf{PSPACE}$ \cite{Sha92}, $\mathsf{MIP}=\mathsf{NEXP}$ \cite{BFL91}, and more recently, $\mathsf{MIP}^*=\mathsf{RE}$ \cite{JNVWY20} theorems provide famous examples where complexity classes turned out to be equal, even in the teeth of oracles relative to which they were unequal.  On the other hand, so far, almost all such examples have originated from a single source: namely, the use of algebraic techniques in interactive proof systems.  And if, for example, we want to understand the consequences of $\mathsf{NP} \subseteq \mathsf{BQP}$, then arguably it makes little sense to search for nonrelativizing consequences if we don't even understand yet what the relativizing consequences (that is, the consequences that hold relative to all oracles) are or are not.

In quantum complexity theory, even more than in classical complexity theory, relativization has been an inextricable part of progress from the very beginning.  The likely explanation is that, even when we just count queries to an oracle, in the quantum setting we need to consider algorithms that query all oracle bits in superposition---so that even in the most basic scenarios, it is already unintuitive what can and cannot be done, and so oracle results must do much more than formalize the obvious.

More concretely, Bernstein and Vazirani \cite{BV97} introduced some of the basic techniques of quantum algorithms in order to prove, for the first time, that there exists an oracle $A$ such that $\mathsf{BPP}^A \neq \mathsf{BQP}^A$. Shortly afterward, Simon \cite{Sim97} gave a quantitatively stronger oracle
separation between $\mathsf{BPP}$ and $\mathsf{BQP}$, and then Shor \cite{Sho99} gave a still stronger separation, along the way to his famous discovery that $\textsc{Factoring}$ is in $\mathsf{BQP}$.

On the negative side, Bennett, Bernstein, Brassard, and Vazirani \cite{BBBV97} showed that there exists an oracle relative to which $\mathsf{NP} \not\subset \mathsf{BQP}$: indeed, relative to which there are problems that take $n$ time for an $\mathsf{NP}$ machine but $\Omega\left(2^{n/2}\right)$ time for a $\mathsf{BQP}$ machine.  Following the discovery of Grover's algorithm \cite{Gro96}, which quantumly searches any list of $N$ items in $O\left(\sqrt{N}\right)$ queries, the result of Bennett, Bernstein, Brassard, and Vazirani gained the interpretation that \textit{Grover's algorithm is optimal}.  In other words, any quantum algorithm for $\mathsf{NP}$-complete problems that gets more than the square-root speedup of Grover's algorithm must be ``non-black-box.''  It must exploit the structure of a particular $\mathsf{NP}$-complete problem much like a classical algorithm would have to, rather than treating the problem as just an abstract space of $2^n$ possible solutions.

Meanwhile, clearly there are oracles relative to which $\mathsf{P}=\mathsf{BQP}$---for example, a $\mathsf{PSPACE}$-complete oracle.  But we can ask: would such oracles necessarily collapse the hierarchy of classical complexity classes as well?  In a prescient result that provided an early example of the sort of thing we do in this paper, Fortnow and Rogers \cite{FR98} showed that there exists an oracle relative to which $\mathsf{P}=\mathsf{BQP}$ and yet $\mathsf{PH}$ is infinite.  In other words, if $\mathsf{P}=\mathsf{BQP}$ would imply a collapse of the polynomial hierarchy, then it cannot be for a relativizing reason.  Aaronson and Chen \cite{AC17} extended this to show that there exists an oracle relative to which \textit{sampling-based quantum supremacy is impossible}---i.e., any probability distribution approximately samplable in quantum polynomial time is also approximately samplable in classical polynomial time---and yet $\mathsf{PH}$ is infinite.  In other words, if it is possible to prove the central theoretical conjecture of
quantum supremacy---namely, that there are noisy quantum sampling
experiments that cannot be simulated in classical polynomial time
unless $\mathsf{PH}$ collapses---then nonrelativizing techniques will be needed there as well.

What about showing the power of $\mathsf{BQP}$, by giving oracle obstructions to containments like $\mathsf{BQP}\subseteq\mathsf{NP}$, or $\mathsf{BQP}\subseteq\mathsf{PH}$?  There, until recently, the progress was much more limited.  Watrous \cite{Wat00} showed that there exists an oracle relative to which $\mathsf{BQP}\not\subset\mathsf{NP}$ and even $\mathsf{BQP}\not\subset\mathsf{MA}$ (these separations could also have been shown using the $\textsc{Recursive Fourier Sampling}$ problem, introduced by Bernstein and Vazirani \cite{BV97}).  But extending this further, to get an oracle relative to which $\mathsf{BQP}\not\subset\mathsf{PH}$ or even $\mathsf{BQP}\not\subset\mathsf{AM}$, remained an open problem for two decades.  Aaronson \cite{Aar10} proposed a program for proving an oracle separation between $\mathsf{BQP}$ and $\mathsf{PH}$, involving a new problem he introduced called $\Forrelation$:

\begin{problem}[$\Forrelation$]
Given black-box access to two Boolean functions $f,g:\{0,1\}^n \to \{1,-1\}$, and promised that either
\begin{enumerate}[(i)]
\item $f$ and $g$ are uniformly random and independent, or
\item $f$ and $g$ are uniformly random individually, but $g$ has $\Omega(1)$ correlation with $\hat{f}$, the Boolean Fourier transform of $f$ (i.e., $f$ and $g$ are ``Forrelated''),
\end{enumerate}
decide which.
\end{problem}

Aaronson \cite{Aar10} showed that $\Forrelation$ is solvable, with constant bias, using only a single quantum query to $f$ and $g$ (and $O(n)$ time).  By contrast, he showed that any classical randomized algorithm for the problem needs $\Omega\left(2^{n/4}\right)$ queries---improved by Aaronson and Ambainis \cite{AA18} to $\Omega\left(\frac{2^{n/2}}{n}\right)$ queries, which is essentially tight.  The central conjecture, which Aaronson left open, said that $\Forrelation \not\in \mathsf{PH}$---or equivalently, by the connection between $\mathsf{PH}$ machines and $\mathsf{AC^0}$ circuits \cite{FSS84}, that there are no $\mathsf{AC^0}$ circuits for $\Forrelation$ of constant depth and $2^{\poly(n)}$ size.

Finally, Raz and Tal \cite{RT19} managed to prove Aaronson's conjecture, and thereby obtain the long-sought oracle separation between $\mathsf{BQP}$ and $\mathsf{PH}$.\footnote{Strictly speaking, they did this for a variant of $\Forrelation$ where the correlation between $g$ and $\hat{f}$ is only $\sim \frac{1}{n}$, and thus a quantum algorithm needs $\sim n$ queries to solve the problem, but this will not affect anything that follows.}  Raz and Tal achieved this by introducing new techniques for constant-depth circuit lower bounds, involving Brownian motion and the $L_1$-weight of the low-order
Fourier coefficients of $\mathsf{AC^0}$ functions.  Relevantly for us, Raz and Tal actually proved the following stronger result:

\begin{theorem}[\cite{RT19}]
\label{thm:raz-tal-informal-intro}
A $\mathsf{PH}$ machine can guess whether $f$ and $g$ are uniform or Forrelated with bias at most $2^{-\Omega(n)}$.
\end{theorem}

Recall that before Raz and Tal, we did not even have an oracle
relative to which $\mathsf{BQP} \not\subset \mathsf{AM}$.  Notice that, if $\mathsf{BQP} \subseteq \mathsf{AM}$, then many other conclusions would follow in a relativizing way.  For example, we would have:

\begin{itemize}
\item $\mathsf{P} = \mathsf{NP}$ implies $\mathsf{P} = \mathsf{BQP}$,
\item $\mathsf{NP}^\mathsf{BQP} \subseteq \mathsf{NP}^{\mathsf{AM} \cap \mathsf{coAM}} \subseteq \mathsf{BPP}^\mathsf{NP} \subseteq \mathsf{BQP}^\mathsf{NP}$,
\item If $\mathsf{NP} \subseteq \mathsf{BQP}$, then $\mathsf{NP}^\mathsf{NP} \subseteq \mathsf{NP}^\mathsf{BQP} \subseteq \mathsf{BQP}^\mathsf{NP} = \mathsf{BQP}^\mathsf{BQP} = \mathsf{BQP}$, and
\item If $\mathsf{NP} \subseteq \mathsf{BQP}$, then $\mathsf{NP} \subseteq \mathsf{coAM}$, which implies that $\mathsf{PH}$ collapses.
\end{itemize}

Looking at it a different way, our inability even to separate $\mathsf{BQP}$ from $\mathsf{AM}$ by an oracle served as an obstruction to numerous other oracle separations.

The starting point of this paper was the following question: in a ``post-Raz-Tal world,'' can we at last completely ``unshackle'' $\mathsf{BQP}$ from $\mathsf{P}$, $\mathsf{NP}$, and $\mathsf{PH}$, by showing that there are no relativizing obstructions to any possible answers to questions like the ones we asked in \Cref{sec:contrast_with_bpp}?

\subsection{Our Results}
\label{sec:results}
We achieve new oracle separations that show an astonishing range of possible behaviors for $\mathsf{BQP}$ and related complexity classes---in at least one case, resolving a longstanding open problem in this topic.  Our title, ``The Acrobatics of $\mathsf{BQP}$,'' comes from a unifying theme of the new results being ``freedom.'' We will show that, as far as relativizing techniques can detect, collapses and separations of classical complexity classes place surprisingly few constraints on the power of quantum computation. In most cases, this can be understood as ultimately stemming from the fact that one cannot ``fix the randomness'' (or quantumness) used by a quantum algorithm, similarly to how one fixes the randomness used by a randomized algorithm in many complexity-theoretic arguments.

As we alluded to earlier, many of our new results would not have been possible without Raz and Tal's analysis of $\Forrelation$ \cite{RT19}, which we rely on extensively.  We will treat $\Forrelation$ no longer as just an isolated problem, but as a sort of cryptographic code, by which an oracle can systematically make certain information available to $\mathsf{BQP}$ machines while keeping the information hidden from classical machines.

Having said that, very few of our results will follow from Raz-Tal in any straightforward way.  Most often we need to develop other lower bound tools, in addition to or instead of Raz-Tal.  Our new tools, which seem likely to be of independent interest, include a random restriction lemma for quantum query algorithms, a concentration theorem for the block sensitivity of $\mathsf{AC^0}$ functions, and a provable analogue of the Aaronson-Ambainis conjecture \cite{AA14} for certain sparse oracles.

Perhaps our single most interesting result is the following.

\begin{theorem}[\Cref{cor:np^bqp_not_in_bqp^ph}, restated]
\label{thm:np^bqp_not_in_bqp^np_intro}
There exists an oracle relative to which $\mathsf{NP}^\mathsf{BQP} \not\subset \mathsf{BQP}^\mathsf{NP}$, and indeed $\mathsf{NP}^\mathsf{BQP} \not\subset \mathsf{BQP}^\mathsf{PH}$.
\end{theorem}

As mentioned earlier, \Cref{thm:np^bqp_not_in_bqp^np_intro} resolves an open problem of Fortnow \cite{For05}, and demonstrates a clear difference between $\mathsf{BPP}$ and $\mathsf{BQP}$ that exemplifies the impossibility of pulling the randomness out of a quantum algorithm. Indeed, \Cref{thm:np^bqp_not_in_bqp^np_intro} shows that there is no general, black-box way to move quantumness past an $\mathsf{NP}$ quantifier, like we can do for classical randomness.

As a straightforward byproduct of \Cref{thm:np^bqp_not_in_bqp^np_intro}, we are also able to prove the following:

\begin{theorem}[\Cref{cor:p=np_bqp!=qcma}, restated]
\label{thm:p=np_bqp!=qcma_intro}
There exists an oracle relative to which $\mathsf{P} = \mathsf{NP}$ but $\mathsf{BQP} \neq \mathsf{QCMA}$.
\end{theorem}

Conversely, it will follow from one of our later results, \Cref{thm:np_in_bqp_ph_infinite_intro}, that there exists an oracle relative to which $\mathsf{P} \neq \mathsf{NP}$ and yet $\mathsf{BQP} = \mathsf{QCMA} = \mathsf{QMA}$.  In other words, as far as relativizing techniques are concerned, the classical and quantum versions of the $\mathsf{P}$ vs. $\mathsf{NP}$ question are completely uncoupled from one another.

\Cref{thm:np^bqp_not_in_bqp^np_intro} also represents progress toward a proof of the following conjecture, which might be the most alluring open problem that we leave.

\begin{conjecture}
\label{conj:sigma_k_in_bqp_sigma_k+1_not}
There exists an oracle relative to which $\mathsf{NP} \subseteq \mathsf{BQP}$ but $\mathsf{PH} \not\subset \mathsf{BQP}$.\footnote{This first part of the conjecture was previously raised by Aaronson \cite{Aar10}.}  Indeed, for every $k \in \Naturals$, there exists an oracle relative to which $\mathsf{\Sigma}_k^\mathsf{P} \subseteq \mathsf{BQP}$ but $\mathsf{\Sigma}_{k+1}^\mathsf{P} \not\subset \mathsf{BQP}$.
\end{conjecture}

\Cref{conj:sigma_k_in_bqp_sigma_k+1_not} would provide spectacularly fine control over the relationship between $\mathsf{BQP}$ and $\mathsf{PH}$, going far beyond Raz-Tal to show how $\mathsf{BQP}$ could, e.g., swallow the first $18$ levels of $\mathsf{PH}$ without swallowing the $19$th.  To see the connection between \Cref{thm:np^bqp_not_in_bqp^np_intro} and \Cref{conj:sigma_k_in_bqp_sigma_k+1_not}, suppose $\mathsf{NP}^\mathsf{BQP} \subseteq \mathsf{BQP}^\mathsf{NP}$, and suppose also that $\mathsf{NP} \subseteq \mathsf{BQP}$.  Then, as observed by Fortnow \cite{For05}, this would imply
\[
\mathsf{NP}^\mathsf{NP} \subseteq \mathsf{NP}^\mathsf{BQP} \subseteq \mathsf{BQP}^\mathsf{NP} \subseteq \mathsf{BQP}^\mathsf{BQP} = \mathsf{BQP},
\]
(and so on, for all higher levels of $\mathsf{PH}$), so that $\mathsf{PH} \subseteq \mathsf{BQP}$ as well. Hence, any oracle that witnesses \Cref{conj:sigma_k_in_bqp_sigma_k+1_not} also witnesses \Cref{thm:np^bqp_not_in_bqp^np_intro}, so our proof of \Cref{thm:np^bqp_not_in_bqp^np_intro} is indeed a prerequisite to \Cref{conj:sigma_k_in_bqp_sigma_k+1_not}.

At a high level, we prove \Cref{thm:np^bqp_not_in_bqp^np_intro} by showing that no $\mathsf{BQP}^\mathsf{PH}$ machine can solve the $\OR \circ \Forrelation$ problem, in which one is given a long list of $\Forrelation$ instances, and is tasked with distinguishing whether (1) all of the instances are uniformly random, or (2) at least one of the instances is Forrelated. A first intuition is that $\mathsf{PH}$ machines should gain no useful information from the input, just because $\Forrelation$ ``looks random'' (by Raz-Tal), and hence a $\mathsf{BQP}^\mathsf{PH}$ machine should have roughly the same power as a $\mathsf{BQP}$ machine at deciding $\OR \circ \Forrelation$. If one could show this, then completing the theorem would amount to showing that $\OR \circ \Forrelation$ is hard for $\mathsf{BQP}$ machines, which easily follows from the BBBV Theorem \cite{BBBV97}.

Alas, initial attempts to formalize this intuition fail for a single, crucial reason: the possibility of homomorphic encryption! The Raz-Tal Theorem merely proves that $\Forrelation$ is a strong form of encryption against $\mathsf{PH}$ algorithms. But to rule out a $\mathsf{BQP}^\mathsf{PH}$ algorithm for $\OR \circ \Forrelation$, we \textit{also} have to show that one cannot take a collection of $\Forrelation$ instances and transform them, by means computable in $\mathsf{PH}$, into a single $\Forrelation$ instance whose solution is the $\OR$ of the solutions to the input instances. Put another way, we must show that $\mathsf{AC^0}$ circuits of constant depth and $2^{\poly(n)}$ size cannot homomorphically evaluate the $\OR$ function, when the encryption is done via the $\Forrelation$ problem.

More generally, we even have to show that $\mathsf{AC^0}$ circuits cannot transform the ``ciphertext'' into \textit{any} string that could later be decoded by an efficient quantum algorithm. \Cref{thm:np^bqp_not_in_bqp^np_intro} accomplishes this with the help of an additional structural property of $\mathsf{AC^0}$ circuits: our concentration theorem for block sensitivity. Loosely speaking, the concentration theorem implies that, with
overwhelming probability, any small $\mathsf{AC^0}$ circuit is insensitive to toggling between a yes-instance and a neighboring no-instance of the $\OR \circ \Forrelation$ problem.  This, together with the BBBV Theorem \cite{BBBV97}, then implies that such ``homomorphic encryption'' is impossible.

\bigskip

We also achieve the following converse to \Cref{thm:np^bqp_not_in_bqp^np_intro}:

\begin{theorem}[\Cref{cor:bqp^np_not_in_ph^bqp}, restated]
\label{thm:bqp^np_not_in_ph^bqp_intro}
There exists an oracle relative to which $\mathsf{BQP}^\mathsf{NP} \not\subset \mathsf{PH}^\mathsf{BQP}$, and even $\mathsf{BQP}^\mathsf{NP} \not\subset \mathsf{PH}^\mathsf{PromiseBQP}$.
\end{theorem}

Note that an oracle relative to which $\mathsf{BQP}^\mathsf{NP} \not\subset \mathsf{NP}^\mathsf{BQP}$ is almost trivial to achieve, for example by considering a problem in $\mathsf{coNP}$.  However, $\mathsf{BQP}^\mathsf{NP} \not\subset \mathsf{PH}^\mathsf{BQP}$ is much harder.  At a high level, rather than considering the composed problem $\OR \circ \Forrelation$, we now need to consider the reverse composition: $\Forrelation \circ \OR$, a problem that's clearly in $\mathsf{BQP}^\mathsf{NP}$, but plausibly not in $\mathsf{PH}^\mathsf{BQP}$. The key step is to show that, when solving $\Forrelation \circ \OR$, \textit{any $\mathsf{PH}^\mathsf{BQP}$ machine can be simulated by a $\mathsf{PH}$ machine}: the $\mathsf{BQP}$ oracle is completely superfluous!  Once we've shown that, $\Forrelation \circ \OR \not \in \mathsf{PH}$ then follows immediately from Raz-Tal.

\bigskip

For our next result, recall that $\mathsf{QMA}$, or \textit{Quantum Merlin-Arthur}, is the class of problems for which a yes-answer can be witnessed by a polynomial-size quantum state. Perhaps our second most interesting result is this:

\begin{theorem}[\Cref{cor:pp_not_in_qmah}, restated]
\label{thm:pp_not_in_qmah_intro}
$\mathsf{PP}$ is not contained in the \emph{``$\mathsf{QMA}$ hierarchy''}, consisting of constant-depth towers of the form $\mathsf{QMA}^{\mathsf{QMA}^{\mathsf{QMA}^{\cdots}}}$, with probability $1$ relative to a random oracle.\footnote{Actually, our formal definition of the $\mathsf{QMA}$ hierarchy is more general than the version given here, in order to accommodate recursive queries to $\mathsf{QMA}$ promise problems. This only makes our separation stronger. See \Cref{sec:complexity_classes} for details.}
\end{theorem}

Note that $\mathsf{PP} = \mathsf{PostBQP}$, where $\mathsf{PostBQP}$ denotes $\mathsf{BQP}$ augmented with the power of postselection \cite{Aar05}, and so \Cref{thm:pp_not_in_qmah_intro} contrasts with the classical containment $\mathsf{PostBPP} \subseteq \mathsf{BPP}^\mathsf{NP} \subseteq \mathsf{PH}$ \cite{HHT97,Kup15}. Nevertheless, before this paper, to our knowledge, it was not even known how to
construct an oracle relative to which $\mathsf{PP} \not\subset \mathsf{BQP}^\mathsf{NP}$, let alone classes like $\mathsf{BQP}^{\mathsf{NP}^{\mathsf{BQP}^{\mathsf{NP}^{\cdots}}}}$ or $\mathsf{QCMA}^{\mathsf{QCMA}^{\mathsf{QCMA}^{\cdots}}}$, which are contained in the $\mathsf{QMA}$ hierarchy. The closest result we are aware of is due to Kretschmer \cite{Kre21c}, who gave a \textit{quantum} oracle relative to which $\mathsf{BQP} = \mathsf{QMA} \neq \mathsf{PostBQP}$.

Perhaps shockingly, our proof of \Cref{thm:pp_not_in_qmah_intro} can be extended even to show that $\mathsf{PP}$ is not in, say, $\mathsf{QMIP}^{\mathsf{QMIP}^{\mathsf{QMIP}^{\cdots}}}$ relative to a random oracle, where $\mathsf{QMIP}$ means Quantum Multi-prover Interactive Proofs with entangled provers.  This is despite the breakthrough results of Reichardt, Unger, and Vazirani \cite{RUV13}, and more recently Ji, Natarajan, Vidick, Wright, and Yuen \cite{JNVWY20}, which showed that in the \textit{unrelativized} world, $\mathsf{QMIP} = \mathsf{MIP}^* = \mathsf{RE}$ (where $\mathsf{MIP}^*$ means $\mathsf{QMIP}$ with classical communication only, and $\mathsf{RE}$ means Recursively Enumerable), so in particular, $\mathsf{QMIP}$ contains the halting problem.  This underscores the dramatic extent to which results like $\mathsf{QMIP} = \mathsf{RE}$ are nonrelativizing!

\Cref{thm:pp_not_in_qmah_intro} can also be understood as showing that in the black-box setting, there is no quantum analogue of Stockmeyer's approximate counting algorithm \cite{Sto83}. For a probabilistic algorithm $M$ that runs in $\poly(n)$ time and an error bound $\eps \ge \frac{1}{\poly(n)}$, the approximate counting problem is to estimate the acceptance probability of $M$ up to a multiplicative factor of $1 + \eps$. Stockmeyer's algorithm \cite{Sto83} gives a relativizing $\poly(n)$-time reduction from the approximate counting problem to a problem in the third level of the polynomial hierarchy, and crucially relies on pulling the randomness out of $M$. In structural complexity terms, Stockmeyer's algorithm can be reinterpreted as showing that $\mathsf{SBP} \subseteq \mathsf{PH}$ relative to all oracles, where $\mathsf{SBP}$ is the complexity class defined in \cite{BGM05} that captures approximate counting.

One might wonder: is there a version of Stockmeyer's algorithm for the \textit{quantum} approximate counting problem, where we instead wish to approximate the acceptance probability of a quantum algorithm? In particular, is $\mathsf{SBQP}$, the complexity class that captures quantum approximate counting \cite{Kup15}, contained in the $\mathsf{QMA}$ hierarchy?\footnote{We thank Patrick Rall (personal communication) for bringing this question to our attention.} Kuperberg \cite{Kup15} showed that $\mathsf{PP} \subseteq \mathsf{P}^\mathsf{SBQP}$, so it follows that $\mathsf{PP} \subseteq \mathsf{QMAH}$ if and only if $\mathsf{SBQP} \subseteq \mathsf{QMAH}$, where $\mathsf{QMAH}$ denotes the $\mathsf{QMA}$ hierarchy. Thus, \Cref{thm:pp_not_in_qmah_intro} implies that $\mathsf{SBQP} \not\subset \mathsf{QMAH}$ relative to a random oracle, implying that such a quantum analogue of Stockmeyer's algorithm does not exist in the black-box setting.\footnote{Note that this is just one of many possible ways that we could ask whether there exists a quantum analogue of Stockmeyer's algorithm. For example, one might consider alternative definitions of the quantum approximate counting task, such as the problem defined in \cite{BCGW21} of approximating the number of witness states accepted by a $\mathsf{QMA}$ verifier. One might also consider other definitions of the ``quantum polynomial hierarchy,'' some of which are explored in \cite{GSS+18}.} This demonstrates yet another case where a classical complexity result that relies on fixing randomness cannot be generalized to the quantum setting.

Notably, our proof of \Cref{thm:pp_not_in_qmah_intro} does not appeal to Raz-Tal at all, but instead relies on a new random restriction lemma for the acceptance probabilities of quantum query algorithms. Our random restriction lemma shows that if one randomly fixes most of the inputs to a quantum query algorithm, then the algorithm's behavior on the unrestricted inputs can be approximated by a ``simple'' function (say, a small decision tree or small DNF formula). We then use this random restriction lemma to generalize the usual random restriction proof that, for example, $\PARITY \not\in \mathsf{AC^0}$ \cite{Has87}.

\bigskip

Here is another noteworthy result that we are able to obtain, by
combining random restriction arguments with lower bounds on quantum query complexity:

\begin{theorem}[\Cref{cor:bqp_sigma_k}, restated]
\label{thm:bqp_sigma_k_intro}
For every $k \in \Naturals$, $\mathsf{\Sigma}_{k+1}^\mathsf{P} \not\subset \mathsf{BQP}^{\mathsf{\Sigma}_{k}^\mathsf{P}}$ with probability $1$ relative to a random oracle.
\end{theorem}

\Cref{thm:bqp_sigma_k_intro} extends the breakthrough of H\r{a}stad, Rossman, Servedio, and Tan \cite{HRST17}, who (solving an open problem from the 1980s) showed that $\mathsf{PH}$ is infinite relative to a random oracle with probability $1$.  Our result shows, not only that a random oracle creates a gap between every two successive levels of $\mathsf{PH}$, but that quantum computing fails to bridge that gap.

Again, \Cref{thm:bqp_sigma_k_intro} represents a necessary step toward a proof of \Cref{conj:sigma_k_in_bqp_sigma_k+1_not}, because if we had $\mathsf{\Sigma}_{k+1}^\mathsf{P} \subseteq \mathsf{BQP}^{\mathsf{\Sigma}_{k}^\mathsf{P}}$, then clearly $\mathsf{\Sigma}_k^\mathsf{P} \subseteq \mathsf{BQP}$ would imply $\mathsf{\Sigma}_{k+1}^\mathsf{P} \subseteq \mathsf{BQP}^\mathsf{BQP} = \mathsf{BQP}$.

\bigskip

Our last two theorems return to the theme of the autonomy of $\mathsf{BQP}$.

\begin{theorem}[\Cref{thm:bqp=pp_ph_infinite}, restated]
\label{thm:np_in_bqp_ph_infinite_intro}
There exists an oracle relative to which $\mathsf{NP} \subseteq \mathsf{BQP}$, and indeed $\mathsf{BQP} = \mathsf{P^{\# P}}$, and yet $\mathsf{PH}$ is infinite.
\end{theorem}

\Cref{thm:np_in_bqp_ph_infinite_intro} resolves a question of Aaronson \cite{Aar10}.
As a simple corollary (\Cref{cor:bqp_not_in_np/poly}), we also obtain an oracle relative to which $\mathsf{BQP} \not\subset \mathsf{NP/poly}$, resolving a question of Aaronson, Cojocaru, Gheorghiu, and Kashefi \cite{ACGK19}.

For three decades, one of the great questions of quantum computation has been whether it can solve $\mathsf{NP}$-complete problems in polynomial time. Many experts guess that the answer is no, for similar reasons as they guess that $\mathsf{P} \neq \mathsf{NP}$---say, the BBBV Theorem \cite{BBBV97}, combined with our failure to find any promising leads for evading that theorem's assumptions in the worst case.  But the fact remains that we have no structural evidence connecting the $\mathsf{NP} \not\subset \mathsf{BQP}$ conjecture to any ``pre-quantum'' beliefs about complexity classes.  No one has any idea how to show, for example, that if $\mathsf{NP} \subseteq \mathsf{BQP}$ then $\mathsf{P} = \mathsf{NP}$ as well, or anything even remotely in that direction.

Given the experience of classical complexity theory, it would be
reasonable to hope for a theorem showing that, if $\mathsf{NP} \subseteq \mathsf{BQP}$, then $\mathsf{PH}$ collapses---analogous to the Karp-Lipton Theorem \cite{KL80}, that if $\mathsf{NP} \subset \mathsf{P/poly}$ then $\mathsf{PH}$ collapses, or the Boppana-H\r{a}stad-Zachos Theorem \cite{BHZ87}, that if $\mathsf{NP} \subseteq \mathsf{coAM}$ then $\mathsf{PH}$ collapses.  No such result is known for $\mathsf{NP} \subseteq \mathsf{BQP}$, once again because of the difficulty that there is no known way to pull the randomness out of a $\mathsf{BQP}$ algorithm.  \Cref{thm:np_in_bqp_ph_infinite_intro} helps to explain this situation, by showing that any proof of such a conditional collapse would have to be nonrelativizing.  The proof of \Cref{thm:np_in_bqp_ph_infinite_intro} builds, again, on the Raz-Tal Theorem.  And this is easily seen to be necessary, since as we pointed out earlier, if $\mathsf{BQP} \subseteq \mathsf{AM}$, then $\mathsf{NP} \subseteq \mathsf{BQP}$ really \textit{would} imply a collapse of $\mathsf{PH}$.

\begin{theorem}[\Cref{thm:p=np_bqp=pp}, restated]
\label{thm:p=np_bqp=pp_intro}
There exists an oracle relative to which $\mathsf{P} = \mathsf{NP} \neq \mathsf{BQP} = \mathsf{P^{\# P}}$.
\end{theorem}

\Cref{thm:p=np_bqp=pp_intro} says, in effect, that there is no relativizing obstruction to $\mathsf{BQP}$ being inordinately powerful even while $\mathsf{NP}$ is inordinately weak. It substantially extends the Raz-Tal Theorem, that there is an oracle relative to which $\mathsf{BQP} \not\subset \mathsf{PH}$, to show that in some oracle worlds, $\mathsf{BQP}$ doesn't go just \textit{slightly} beyond the power of $\mathsf{PH}$ (which,
if $\mathsf{P} = \mathsf{NP}$, is simply the power of $\mathsf{P}$), but \textit{vastly} beyond it. Once again, this illustrates the difference between randomness and quantumness, because if $\mathsf{P} = \mathsf{NP}$, then $\mathsf{P} = \mathsf{BPP}$ for relativizing reasons.

We conjecture that \Cref{thm:p=np_bqp=pp_intro} could be extended yet further, to give an oracle relative to which $\mathsf{P} = \mathsf{NP}$ and yet $\mathsf{BQP} = \mathsf{EXP}$, but we leave that problem to future work.

\subsection{Proof Techniques}

We now give rough sketches of the important ideas needed to prove our results. Here, in contrast to \Cref{sec:results}, we present the results in the order that they appear in the main text, which is roughly in order of increasing technical difficulty.

Our proofs of \Cref{thm:np_in_bqp_ph_infinite_intro} and \Cref{thm:p=np_bqp=pp_intro} serve as useful warm-ups, giving a flavor for how we use the Raz-Tal Theorem and oracle construction techniques in later proofs. In \Cref{thm:np_in_bqp_ph_infinite_intro}, to construct an oracle where $\mathsf{BQP} = \mathsf{P^{\# P}}$ but $\mathsf{PH}$ is infinite, we start by taking a random oracle, which by the work of H\r{a}stad, Rossman, Servedio, and Tan \cite{HRST17, RST15} is known to make $\mathsf{PH}$ infinite. Then, for each $\mathsf{P^{\# P}}$ machine $M$, we add to the oracle an instance of the $\Forrelation$ problem that encodes the behavior of $M$: if $M$ accepts, we choose a Forrelated instance, while if $M$ rejects, we choose a uniformly random instance. This gives a $\mathsf{BQP}$ machine the power to decide any $\mathsf{P^{\# P}}$ language.\footnote{The careful reader might wonder: if we can encode the answers to $\mathsf{P^{\# P}}$ machines, then what is to stop us from encoding the answers to some arbitrarily powerful class, such as $\mathsf{EXP}$ or $\mathsf{RE}$, into the $\Forrelation$ instances? For a $\mathsf{P^{\# P}}$ machine $M$, we exploit the fact that we can always choose $\Forrelation$ instances on oracle strings that cannot be queried by $M$. For example, if $M$ runs in time $t$, then we can encode $M$'s output into strings of length $t^c$ for some $c > 1$, which remain accessible to a $\mathsf{BQP}$ machine with a larger polynomial running time. By contrast, if we tried to do the same for an $\mathsf{EXP}$ machine (say), we run into the problem that the machine whose behavior we are trying to encode could query the very encoding we are making of its output, and thus our oracle would be circularly defined.}

It remains to argue that adding these $\Forrelation$ instances does not collapse $\mathsf{PH}$. We want to show that relative to our oracle, for every $k$, there exists a language in $\mathsf{\Sigma}_{k+1}^\mathsf{P}$ that is not in $\mathsf{\Sigma}_{k}^\mathsf{P}$. This is where we leverage the Raz-Tal Theorem: because the $\Forrelation$ instances look random to $\mathsf{PH}$, we can show, by a hybrid argument, that a $\mathsf{\Sigma}_{k}^\mathsf{P}$ algorithm's probability of correctly deciding a target function in $\mathsf{\Sigma}_{k+1}^\mathsf{P}$ is roughly unchanged if we replace the $\Forrelation$ instances with uncorrelated, uniformly random bits. But auxiliary random bits cannot possibly improve the success probability, and so a simple appeal to \cite{HRST17} implies that the $\mathsf{\Sigma}_{k+1}^\mathsf{P}$ language remains hard for $\mathsf{\Sigma}_{k}^\mathsf{P}$.

\bigskip

The proof of \Cref{thm:p=np_bqp=pp_intro}, giving an oracle where $\mathsf{P} = \mathsf{NP} \neq \mathsf{BQP} = \mathsf{P^{\# P}}$, follows a similar recipe to the proof of \Cref{thm:np_in_bqp_ph_infinite_intro}. We start with a random oracle, which separates $\mathsf{PH}$ from $\mathsf{P^{\# P}}$, and then we add a second region of the oracle that puts $\mathsf{P^{\# P}}$ into $\mathsf{BQP}$ by encoding all $\mathsf{P^{\# P}}$ queries in instances of the $\Forrelation$ problem. Next, we add a third region of the oracle that answers all $\mathsf{NP}$ queries, which has the effect of collapsing $\mathsf{PH}$ to $\mathsf{P}$. Finally, we again leverage the Raz-Tal Theorem to argue that the $\Forrelation$ instances have no effect on the separation between $\mathsf{PH}$ and $\mathsf{P^{\# P}}$, because the $\Forrelation$ instances look random to $\mathsf{PH}$ algorithms.

\bigskip

We next prove \Cref{thm:bqp_sigma_k_intro}, that $\mathsf{\Sigma}_{k+1}^\mathsf{P} \not\subset \mathsf{BQP}^{\mathsf{\Sigma}_{k}^\mathsf{P}}$ relative to a random oracle. Our proof builds heavily on the proof by \cite{HRST17} that $\mathsf{\Sigma}_{k+1}^\mathsf{P} \not\subset \mathsf{\Sigma}_{k}^\mathsf{P}$ relative to a random oracle. Indeed, our proof is virtually identical, except for a single additional step.

\cite{HRST17}'s proof involves showing that there exists a function $\Sipser_{d}$ that is computable by a small $\mathsf{AC^0}$ circuit of depth $d$ (which corresponds to a $\mathsf{\Sigma}_{d-1}^\mathsf{P}$ algorithm), but such that any small $\mathsf{AC^0}$ circuit of depth $d-1$ (which corresponds to a $\mathsf{\Sigma}_{d-2}^\mathsf{P}$ algorithm) computes $\Sipser_{d}$ on at most a $\frac{1}{2} + o(1)$ fraction of random inputs. This proof uses random restrictions, or more accurately, a generalization of random restrictions called \textit{random projections} by \cite{HRST17}. Roughly speaking, the proof constructs a distribution $\mathcal{R}$ over random projections with the following properties:

\begin{enumerate}[(i)]
\item Any small $\mathsf{AC^0}$ circuit $C$ of depth $d-1$ ``simplifies'' with high probability under a random projection drawn from $\mathcal{R}$, say, by collapsing to a low-depth decision tree.
\item The target $\Sipser_d$ function ``retains structure''  with high probability under a random projection drawn from $\mathcal{R}$.
\item The structure retained in (ii) implies that the original unrestricted circuit $C$ fails to compute the $\Sipser_d$ function on a large fraction of inputs.
\end{enumerate}

To prove \Cref{thm:bqp_sigma_k_intro}, we generalize step (i) above from $\mathsf{\Sigma}_{d-2}^\mathsf{P}$ algorithms to $\mathsf{BQP}^{\mathsf{\Sigma}_{d-2}^\mathsf{P}}$ algorithms. That is, if we have a quantum algorithm that queries arbitrary depth-$(d-1)$ $\mathsf{AC^0}$ functions of the input, then we show that this algorithm's acceptance probability also ``simplifies'' under a random projection from $\mathcal{R}$.  We prove this by combining the BBBV Theorem \cite{BBBV97} with \cite{HRST17}'s proof of step (i).

\bigskip

We next move on to the proof of \Cref{thm:np^bqp_not_in_bqp^np_intro}, where we construct an oracle relative to which $\mathsf{NP^{BQP}}\not\subset\mathsf{BQP^{PH}}$. Recall that we prove \Cref{thm:np^bqp_not_in_bqp^np_intro} by showing that no $\mathsf{BQP}^\mathsf{PH}$ machine can solve the $\OR \circ \Forrelation$ problem. To establish this, imagine that we fix a ``no'' instance $x$ of the $\OR \circ \Forrelation$ problem, meaning that $x$ consists of a list of $\sim 2^n$ $\Forrelation$ instances that are all uniformly random (i.e. non-Forrelated). We can turn $x$ into an adjacent ``yes'' instance $y$ by randomly choosing one of the $\Forrelation$ instances of $x$ and changing it to be Forrelated.

Our proof amounts to showing that with high probability over $x$, an $\mathsf{AC^0}$ circuit of size $2^{\poly(n)}$ is unlikely (over $y$) to distinguish $x$ from $y$. Then, applying the BBBV Theorem \cite{BBBV97}, we can show that for most choices of $x$, a $\mathsf{BQP}^\mathsf{PH}$ algorithm is unlikely to distinguish $x$ from $y$, implying that it could not have solved the $\OR \circ \Forrelation$ problem.

Next, we notice that it suffices to consider what happens when, instead of choosing $y$ by randomly flipping one of the $\Forrelation$ instances of $x$ from uniformly random to Forrelated, we instead choose a string $z$ by randomly resampling one of the instances of $x$ from the uniform distribution. This is because, as a straightforward consequence of the Raz-Tal Theorem (\Cref{thm:raz-tal-informal-intro}), if $f$ is an $\mathsf{AC^0}$ circuit of size $2^{\poly(n)}$, then $\left|\Pr_y[f(x) \neq f(y)] - \Pr_{z}[f(x) \neq f(z)] \right| \le 2^{-\Omega(n)}$.

Our key observation is that the quantity $\Pr_{z}[f(x) \neq f(z)]$ is proportional to a sort of ``block sensitivity'' of $f$ on $x$. More precisely, it is proportional to an appropriate averaged notion of block sensitivity, where the average is taken over collections of blocks that respect the partition into separate $\Forrelation$ instances. This is where our block sensitivity concentration theorem comes into play:

\begin{theorem}[\Cref{cor:ac0_block_sensitivity_tail_bound}, informal]
\label{thm:ac0_block_sensitivity_intro}
Let $f: \{0,1\}^N \to \{0,1\}$ be an $\mathsf{AC^0}$ circuit of size $\quasipoly(N)$ and depth $O(1)$, and let $B = \{B_1,B_2,\ldots,B_k\}$ be a collection of disjoint subsets of $[N]$. Then for any $t$,
\[
\Pr_{x \sim \{0,1\}^N} \left[ \bs_B^x(f) \ge t \right] \le 4N\cdot 2^{-\Omega\left(\frac{t}{\polylog(N)}\right)},
\]
where $\bs_B^x(f)$ denotes the block sensitivity of $f$ on $x$ with respect to $B$.
\end{theorem}

Informally, \Cref{thm:ac0_block_sensitivity_intro} says that the probability that an $\mathsf{AC^0}$ circuit has $B$-block sensitivity $t \gg \polylog(N)$ on a random input $x$ decays exponentially in $t$. This generalizes the result of Linial, Mansour, and Nisan \cite{LMN93} that the \textit{average} sensitivity of $\mathsf{AC^0}$ circuits is at most $\polylog(N)$. It also generalizes a concentration theorem for the sensitivity of $\mathsf{AC^0}$ circuits that appeared implicitly in the work of Gopalan, Servedio, Tal, and Wigderson \cite{GSTW16}, by taking $B$ to be the partition into singletons.\footnote{Interestingly, \cite{GSTW16}'s goal, in proving their concentration theorem for the sensitivity of $\mathsf{AC^0}$, was to make progress toward a proof of the famous \textit{Sensitivity Conjecture}---a goal that Huang \cite{Hua19} achieved shortly afterward using completely different methods.  One happy corollary of this work is that, nevertheless, \cite{GSTW16}'s attempt on the problem was not entirely in vain.} In fact, we derive \Cref{thm:ac0_block_sensitivity_intro} as a simple corollary of such a sensitivity tail bound for $\mathsf{AC^0}$. For completeness, we will also prove our own sensitivity tail bound, rather than appealing to \cite{GSTW16}. Our sensitivity tail bound follows from an $\mathsf{AC^0}$ random restriction lemma due to Rossman \cite{Ros17}.

\bigskip

To prove \Cref{thm:p=np_bqp!=qcma_intro}, which gives an oracle relative to which $\mathsf{P} = \mathsf{NP}$ but $\mathsf{BQP} \neq \mathsf{QCMA}$, we use a similar technique to the proof of \Cref{thm:p=np_bqp=pp_intro}. We first take the oracle constructed in \Cref{thm:np^bqp_not_in_bqp^np_intro} that contains instances of the $\OR \circ \Forrelation$ problem. Next, we add a second region of the oracle that answers all $\mathsf{NP}$ queries. This collapses $\mathsf{PH}$ to $\mathsf{P}$. Finally, we use \Cref{thm:np^bqp_not_in_bqp^np_intro} to argue that these $\mathsf{NP}$ queries do not enable a $\mathsf{BQP}$ machine to solve the $\OR \circ \Forrelation$ problem, which is in $\mathsf{QCMA}$.

\bigskip

We now move on to the proof of \Cref{thm:bqp^np_not_in_ph^bqp_intro}, that there exists an oracle relative to which $\mathsf{BQP}^\mathsf{NP} \not\subset \mathsf{PH}^\mathsf{BQP}$. Recall that our strategy is to show that no $\mathsf{PH}^\mathsf{BQP}$ machine can solve the $\Forrelation \circ \OR$ problem. We prove this by showing that with high probability, a $\mathsf{PH}^\mathsf{BQP}$ machine on a random instance of the $\Forrelation \circ \OR$ problem can be simulated by a $\mathsf{PH}$ machine, from which a lower bound easily follows from the Raz-Tal Theorem. This simulation hinges on the following theorem, which seems very likely to be of independent interest:

\begin{theorem}[\Cref{thm:sparse_aa_parameterized}, informal]
\label{thm:sparse_aa_intro}
Consider a quantum algorithm $Q$ that makes $T$ queries to an $M \times N$ array of bits $x$, where each length-$N$ row of $x$ contains a single uniformly random $1$ and $0$s everywhere else. Then for any $\eps \gg \frac{T}{\sqrt{N}}$ and $\delta > 0$, there exists a deterministic classical algorithm that makes $O\left(\frac{T^5}{\eps^4} \log \frac{T}{\delta}\right)$ queries to $x$, and approximates $Q$'s acceptance probability to within additive error $\eps$ on a $1 - \delta$ fraction of such randomly chosen $x$'s.
\end{theorem}

Informally, \Cref{thm:sparse_aa_intro} says that any fast enough quantum algorithm can be simulated by a deterministic classical algorithm, with at most a polynomial blowup in query complexity, on almost all sufficiently sparse oracles. The crucial point here is that the classical simulation still needs to work, even in most cases where the quantum algorithm is lucky enough to find many `$1$' bits.  We prove \Cref{thm:sparse_aa_intro} via a combination of tail bounds and the BBBV hybrid argument \cite{BBBV97}.

In the statement of \Cref{thm:sparse_aa_intro}, we do not know whether the exponent of $5$ on $T$ is tight, and suspect that it isn't.  We only know that the exponent needs to be at least $2$, because of Grover's algorithm \cite{Gro96}.

We remark that \Cref{thm:sparse_aa_intro} bears similarity to a well-known conjecture that involves simulation of quantum query algorithms by classical algorithms. A decade ago, motivated by the question of whether $\mathsf{P}=\mathsf{BQP}$ relative to a random oracle with probability $1$, Aaronson and Ambainis \cite{AA14} proposed the following conjecture:

\begin{conjecture}[{\cite[Conjecture 1.5]{AA14}; attributed to folklore}]
\label{conj:aaronson_ambainis}
Consider a quantum algorithm $Q$ that makes $T$ queries to $x \in \{0,1\}^N$. Then for any $\eps, \delta > 0$, there exists a deterministic classical algorithm that makes $\poly\left(T, \frac{1}{\eps}, \frac{1}{\delta}\right)$ queries to $x$, and approximates $Q$'s acceptance probability to within additive error $\eps$ on a $1 - \delta$ fraction of uniformly randomly inputs $x$.
\end{conjecture}

While \Cref{conj:aaronson_ambainis} has become influential in Fourier analysis of Boolean functions,\footnote{In the context of Fourier analysis, the Aaronson-Ambainis Conjecture usually refers to a closely-related conjecture about influences of bounded low-degree polynomials; see e.g. \cite{Mon12,OZ16}. Aaronson and Ambainis \cite{AA14} showed that this related conjecture implies \Cref{conj:aaronson_ambainis}.} it remains open to this day. \Cref{thm:sparse_aa_intro} could be seen as \textit{the analogue of \Cref{conj:aaronson_ambainis} for sparse oracles}---an analogue that, because of the sparseness, turns out to be much easier to prove.

\bigskip

We conclude with the proof of \Cref{thm:pp_not_in_qmah_intro}, showing that $\mathsf{PP}$ is not contained in the $\mathsf{QMA}$ hierarchy relative to a random oracle. This is arguably the most technically involved part of this work. Recall that our key contribution, and the most important step of our proof, is a random restriction lemma for quantum query algorithms. In fact, we even prove a random restriction lemma for functions with low \emph{quantum Merlin-Arthur $(\mathsf{QMA})$ query complexity}: that is, functions $f$ where a verifier, given an arbitrarily long ``witness state,'' can become convinced that $f(x) = 1$ by making few queries to $x$. Notably, our definition of $\mathsf{QMA}$ query complexity does not care about the length of the witness, but only the number of queries made by the verifier. This property allows us to extend our results to complexity classes beyond $\mathsf{QMA}$, such as $\mathsf{QMIP}$.

An informal statement of our random restriction lemma is given below:

\begin{theorem}[\Cref{thm:qma_random_restriction}, informal]
\label{thm:qma_random_restriction_intro}
Consider a partial function $f: \{0,1\}^N \to \{0,1,\bot\}$ with $\mathsf{QMA}$ query complexity at most $\polylog(N)$. For some $p = \frac{1}{\sqrt{N}\polylog(N)}$, let $\rho$ be a random restriction that leaves each variable unrestricted with probability $p$. Then $f_\rho$ is $\frac{1}{\quasipoly(N)}$-close, in expectation over $\rho$, to a $\polylog(N)$-width DNF formula.\footnote{By saying that $f_\rho$ is ``close'' to a DNF formula, we mean that there exists a DNF $g$ depending on $\rho$ such that the fraction of inputs on which $f_\rho$ and $g$ agree is $1 - \frac{1}{\quasipoly(N)}$, in expectation over $\rho$. In \Cref{sec:closeness}, we introduce some additional notation and terminology that makes it easier to manipulate such expressions, but we will not use them in this exposition.}
\end{theorem}

An unusual feature of \Cref{thm:qma_random_restriction_intro} is that we can only show that $f_\rho$ is \textit{close} to a simple function \textit{in expectation}. By contrast, H\r{a}stad's switching lemma for DNF formulas \cite{Has87} shows that the restricted function reduces to a simple function \textit{with high probability}, so in some sense our result is weaker. Additionally, unlike the switching lemma, our result has a quantitative dependence on the number of inputs $N$. Whether this dependence can be removed (so that the bound depends only on the number of queries) remains an interesting problem for future work.

With \Cref{thm:qma_random_restriction_intro} in hand, proving that $\mathsf{PP} \not\subset \mathsf{QMA}^{\mathsf{QMA}^{\mathsf{QMA}^{\cdots}}}$ relative to a random oracle is conceptually analogous to the proof that $\mathsf{PP} \not\subset \mathsf{PH}$ relative to a random oracle \cite{Has87}. We first view a $\mathsf{QMA}^{\mathsf{QMA}^{\mathsf{QMA}^{\cdots}}}$ machine as a small constant-depth circuit in which the gates are functions of low $\mathsf{QMA}$ query complexity. Then we want to argue that the probability that such a circuit agrees with the $\PARITY$ function on a random input is small. We accomplish this via repeated application of \Cref{thm:qma_random_restriction_intro}, interleaved with H\r{a}stad's switching lemma for DNF formulas \cite{Has87}.

To elaborate further, we first take a random restriction that, by \Cref{thm:qma_random_restriction_intro}, turns all of the bottom-layer $\mathsf{QMA}$ gates into DNF formulas. Next, we apply another random restriction and appeal to the switching lemma to argue that these DNFs reduce to functions of low decision tree complexity, which can be absorbed into the next layer of $\mathsf{QMA}$ gates. Finally, we repeat as many times as needed until the entire circuit collapses to a low-depth decision tree. Since the $\PARITY$ function reduces to another $\PARITY$ function under any random restriction, we conclude that this decision tree will disagree with the reduced $\PARITY$ function on a large fraction of inputs, and hence the original circuit must have disagreed with the $\PARITY$ function on a large fraction of inputs as well.

Of course, the actual proof of \Cref{thm:pp_not_in_qmah_intro} is more complicated because of the accounting needed to bound the error introduced from \Cref{thm:qma_random_restriction_intro}, but all of the important concepts are captured above.

We end with a few remarks on the proof ideas needed for \Cref{thm:qma_random_restriction_intro}. Essentially, the first step involves proving that if we take a function $f$ computed by a quantum query algorithm $Q$, a random restriction $\rho$, and a uniformly random input $x$ to $f_\rho$, then $x$ likely contains a small set $K$ of ``influential'' variables. These influential variables have the property that for any string $y$ that agrees with $x$ on $K$, $\left|\Pr[Q(x) = 1] - \Pr[Q(y) = 1]\right|$ is bounded by a small constant. Hence, $K$ serves as a certificate for $f_\rho$'s behavior on $x$.

Proving that such a $K$ usually exists amounts to a careful application of the BBBV Theorem \cite{BBBV97}; the reader may find the details in \Cref{thm:bqp_random_restriction}. Finally, we generalize from quantum query algorithms to arbitrary $\mathsf{QMA}$ query algorithms by observing that we only need to keep track of the certificates for inputs $x$ such that $f_\rho(x) = 1$. The DNF we obtain in \Cref{thm:qma_random_restriction_intro} is then simply the $\OR$ of all of these small $1$-certificates.

\section{Preliminaries}
\subsection{Notation and Basic Tools}
We denote by $[N]$ the set $\{1,2,\ldots,N\}$. For a finite set $S$, $|S|$ denotes the size of $S$. If $\mathcal{D}$ is a probability distribution, then $x \sim \mathcal{D}$ means that $x$ is a random variable sampled from $\mathcal{D}$. If $v$ is a real or complex vector, then $||v||$ denotes the Euclidean norm of $v$.

We use $\poly(n)$ to denote an arbitrary polynomially-bounded function of $n$, i.e. a function $f$ for which there is a constant $c$ such that $f(n) \le n^c$ for all sufficiently large $n$. Likewise, we use $\polylog(n)$ for an arbitrary $f$ satisfying $f(n) \le \log(n)^c$ for all sufficiently large $n$, and $\quasipoly(n)$ for an arbitrary $f$ satisfying $f(n) \le 2^{\log(n)^c}$ for all sufficiently large $n$.

For a string $x \in \{0,1\}^N$, $|x|$ denotes the length of $x$. Additionally, if $i \in [N]$, then $x^{\oplus i}$ denotes the string obtained from $x$ by flipping the $i$th bit. Similarly, if $S \subseteq [N]$, then $x^{\oplus S}$ denotes the string obtained from $x$ by flipping the bits corresponding to all indices in $S$. For sets $S \subseteq [N]$, we sometimes use $\{0,1\}^S$ to denote mappings from $S$ to $\{0,1\}$; these may equivalently be identified with strings in $\{0,1\}^{|S|}$ obtained by concatenating the bits of the mapping in order. We denote by $x|_S$ the string in $\{0,1\}^S$ obtained by restricting $x$ to the bits indexed by $S$.

We view partial Boolean functions as functions of the form $f: S \to \{0,1,\bot\}$, where the domain of $f$ is $\Dom(f) \coloneqq \left\{ x \in S : f(x) \in \{0,1\} \right\}$. We use $\bot$ (instead of $*$) to refer to the evaluation of $f$ on inputs outside the domain so as to avoid conflicting with our notation for random restrictions; see \Cref{sec:random_restrictions} below.

We use the following forms of the Chernoff bound:

\begin{fact}[Chernoff bound]
\label{fact:chernoff}
Suppose $X_1,\ldots,X_n$ are independent identically distributed random variables where $X_i = 1$ with probability $p$ and $X_i = 0$ with probability $1-p$. Let $X = \sum_{i=1}^n X_i$ and let $\mu = \E[X] = pn$. Then for all $\delta \ge 0$ it holds that:
\begin{align*}
\Pr\left[X \ge (1 + \delta)\mu\right] &\le e^{-\frac{\delta^2\mu}{2+\delta}},\\
\Pr\left[X \le (1 - \delta)\mu\right] &\le e^{-\frac{\delta^2\mu}{2}},
\end{align*}
Additionally, if $\delta \le 1$, then we may use the weaker bound:
\[
\Pr\left[X \ge (1 + \delta)\mu\right] \le e^{-\frac{\delta^2\mu}{3}}.
\]
\end{fact}

We also require Hoeffding's inequality, which generalizes \Cref{fact:chernoff} to sums of arbitrary independent bounded random variables:
\begin{fact}[Hoeffding's inequality]
\label{fact:hoeffding}
Suppose $X_1,\ldots,X_n$ are independent random variables subject to $a_i \le X_i \le b_i$ for all $i$. Let $X = \sum_{i=1}^n X_i$ and let $\mu = \E[X]$. Then for all $\delta \ge 0$ it holds that:
\[
\Pr\left[X \ge (1 + \delta)\mu\right] \le \exp\left(-\frac{2\delta^2\mu^2}{ \sum_{i=1}^n (b_i - a_i)^2}\right).
\]

\end{fact}

\subsection{Complexity Classes and Oracles}
\label{sec:complexity_classes}
We assume familiarity with basic complexity classes, including: $\mathsf{P}$, $\mathsf{NP}$, $\mathsf{PH} = \bigcup_{k=0}^\infty \mathsf{\Sigma}_k^\mathsf{P}$, $\mathsf{PP}$, $\mathsf{P}^{\mathsf{\# P}}$, $\mathsf{BQP}$, $\mathsf{QCMA}$, and $\mathsf{QMA}$; see e.g. the Complexity Zoo\footnote{\url{https://complexityzoo.net/Complexity_Zoo}} for definitions. For one of these complexity classes $\mathcal{C}$, $\mathsf{Promise}\mathcal{C}$ denotes the corresponding class of promise problems. Recall that a promise problem can be viewed as a partial function $\Pi: \{0,1\}^* \to \{0,1,\bot\}$. We say that a language $A: \{0,1\}^* \to \{0,1\}$ \textit{extends} $\Pi$ if for all $x \in \Dom(\Pi)$, $\Pi(x) = A(x)$.

In this work, we follow the convention that an \textit{algorithm} refers to a (possibly probabilistic) abstract procedure, while a \textit{machine} refers to a computational problem (either a language or promise problem) that is decided by an algorithm. For example, if $\mathcal{A}$ is a polynomial-time quantum algorithm, and $M$ is the $\mathsf{PromiseBQP}$ machine corresponding to $\mathcal{A}$, then this means that $M: \{0,1\}^* \to \{0,1,\bot\}$ is defined by:
\[
M(x) \coloneqq \begin{cases}
0 & \Pr[\mathcal{A}(x) = 1] \le \frac{1}{3},\\
1 & \Pr[\mathcal{A}(x) = 1] \ge \frac{2}{3},\\
\bot & \text{otherwise}.
\end{cases}
\]
Note that, while $\mathcal{A}$ defines a probabilistic procedure, $M$ has no further randomness in its definition after rounding the acceptance probabilities of $\mathcal{A}$.

We frequently make use of complexity classes augmented with oracles, where we use the standard notation that $\mathcal{C}^\mathcal{O}$ denotes a complexity class $\mathcal{C}$ augmented with oracle $\mathcal{O}$. We also consider oracles for promise problems in the standard way: if $\mathcal{C}$ is a complexity class and $\Pi$ a promise problem, then a language (or promise problem) $L$ is in $\mathcal{C}^\Pi$ if there exists a $\mathcal{C}$ oracle machine $M$ such that, for every language $A$ that extends $\Pi$, $M^A$ decides $L$. We also take this as a definition of $M^\Pi$:
\[
M^\Pi(x) \coloneqq \begin{cases}
0 & M^L(x) = 0 \text{ for every language } L \text{ that extends } \Pi,\\
1 & M^L(x) = 1 \text{ for every language } L \text{ that extends } \Pi,\\
\bot & \text{otherwise}.
\end{cases}
\]

We remark that it is not clear to us if this is the ``right'' way to define queries to a promise problem for quantum complexity classes, such as $\mathsf{BQP}^\Pi$ or $\mathsf{QMA}^\Pi$. If we have query access to some quantum algorithm that ``solves'' a promise problem, that algorithm could conceivably behave arbitrarily (even non-unitarily) on the inputs outside of the promise: there is no guarantee that it decides some language, as we assume above. However, since we are chiefly interested in proving \textit{lower bounds} on $\mathsf{QMA}$ complexity, this distinction makes little difference to us: our choice of definition could only possibly make the class $\mathsf{QMA}^\Pi$ \textit{more} powerful than if the oracle could have worse behavior on non-promise inputs. Whether complexity classes such as $\mathsf{BQP}^\Pi$ and $\mathsf{QMA}^\Pi$ are robust with respect to the notion of promise problem queries remains an interesting question for future work.\footnote{Nevertheless, we are not the first to apply this notion of promise problem queries to quantum complexity classes: Aaronson and Drucker \cite{AD14} use the same definition for $\mathsf{QMA}^\Pi$.}

If $\mathcal{C}$ and $\mathcal{D}$ are both complexity classes, then $\mathcal{C}^\mathcal{D} \coloneqq \bigcup_{L \in \mathcal{D}} \mathcal{C}^L$. A tower of relativized complexity classes such as $\mathcal{C}^{\mathcal{D}^\mathcal{O}}$ should be understood as $\mathcal{C}^{\left(\mathcal{D}^\mathcal{O}\right)}$, which is to say that notation for relativization, like exponentiation, is right associative. Nevertheless, in such cases, we can always view a language (or promise problem) in a complexity class such as $\mathcal{C}^{\mathcal{D}^\mathcal{O}}$ as being specified by a $\mathcal{C}^\mathcal{D}$ oracle machine. For example, a language in $\mathsf{P}^{\mathsf{NP}^\mathcal{O}}$ is uniquely specified by (1) a $\mathsf{P}$ oracle machine $A$ (a polynomial-time deterministic oracle Turing machine), (2) an $\mathsf{NP}$ oracle machine $B$ (a polynomial-time nondeterministic oracle Turing machine), and (3) the oracle $\mathcal{O}$. In such cases, we may refer to the pair $\langle A, B \rangle$ as a $\mathsf{P}^\mathsf{NP}$ oracle machine.

We define a quantum analogue of the polynomial hierarchy that we call $\mathsf{QMAH}$, and denote by $\mathsf{PromiseQMAH}$ the promise version of this class. Let $\mathsf{PromiseQMAH}_1 = \mathsf{PromiseQMA}$, and for $k > 1$ we recursively define:
\begin{align*}
\mathsf{PromiseQMAH}_k &\coloneqq \mathsf{PromiseQMA}^{\mathsf{PromiseQMAH}_{k-1}},
\end{align*}
Then, analogous to $\mathsf{PH}$, we take:
\begin{align*}
\mathsf{PromiseQMAH} &\coloneqq \bigcup_{k=1}^\infty \mathsf{PromiseQMAH}_k
\end{align*}
$\mathsf{QMAH}$ denotes the set of languages in $\mathsf{PromiseQMAH}$.
%In the above definition, we define queries to a promise problem in the standard way, as follows. Recall that a promise problem can be viewed as a partial function $\Pi: \{0,1\}^* \to \{0,1,\bot\}$. We say that a language $A: \{0,1\}^* \to \{0,1\}$ \textit{extends} $\Pi$ if for all $x \in \Dom(\Pi)$, $\Pi(x) = A(x)$. Then, a promise problem $P$ is in $\mathsf{PromiseQMA}^\Pi$ if there exists a single $\mathsf{QMA}$ verifier $V$ such that, for every language $A$ that extends $\Pi$, $V^A$ decides $P$. For a class $\mathcal{C}$ of promise problems, we let $\mathsf{PromiseQMA}^\mathcal{C} \coloneqq \bigcup_{\Pi \in \mathcal{C}} \mathsf{PromiseQMA}^\Pi$. This is essentially the same definition that was given by Aaronson and Drucker \cite{AD14}.

We note that there are many other possible ways to define a quantum analogue of the polynomial hierarchy (see \cite{GSS+18}), and that our definition appears to differ from all others that we are aware of. The definition we give is closest in spirit to a class called $\mathsf{BQPH}$ by Vinkhuijzen \cite{Vin18}, except that we allow recursive queries to $\mathsf{PromiseQMA}$ instead of just $\mathsf{QMA}$.

We next specify some terminology and notation that we will use for constructing oracles. We will often find it convenient to specify oracles as a union of disjoint regions. Formally, this means the following. Suppose we have an increasing sequence $n_1 < n_2 < ...$ and an associated sequence of functions $A_1, A_2, \ldots$, where $A_i: \{0,1\}^{n_i} \to \{0,1\}$. We call each $A_i$ a \textit{region}, and define the oracle $\mathcal{O}: \{0,1\}^* \to \{0,1\}$ constructed from these regions via:
\[
\mathcal{O}(x) \coloneqq \begin{cases}
A_i(x) & |x| = n_i\\
0 & \text{otherwise.}
\end{cases}
\]

We may also construct oracles by joining other oracles together. For example, if we have a pair of oracles $A, B: \{0,1\}^* \to \{0,1\}$, then $\mathcal{O} = (A, B)$ means that we define $\mathcal{O}: \{0,1\}^* \to \{0,1\}$ by:
\begin{align*}
\mathcal{O}(0x) &\coloneqq A(x)\\
\mathcal{O}(1x) &\coloneqq B(x).
\end{align*}

A \textit{random oracle} $\mathcal{O}$ is a uniformly random language where for each $x \in \{0,1\}^*$, $\mathcal{O}(x) = 0$ or $1$ with probability $\frac{1}{2}$ (independently for each $x$).

\subsection{Query Complexity and Related Measures}
We assume some familiarity with quantum and classical query complexity. We recommend a survey by Ambainis \cite{Amb18} for additional background and definitions. A quantum query to a string $x \in \{0,1\}^N$ is implemented via the unitary transformation $U_x$ that acts on basis states of the form $\ket{i}\ket{w}$ as $U_x\ket{i}\ket{w} = (-1)^{x_i}\ket{i}\ket{w}$, where $i \in [N]$ and $w$ is an index over a workspace register.

%Following standard notation, for a (possibly partial) function $f: \{0,1\}^N \to \{0,1,\bot\}$, we denote by $\Q(f)$ the bounded-error quantum query complexity of $f$, which means the fewest number of queries made by a quantum algorithm that computes $f$ with error probability at most $\frac{1}{3}$. We denote by $\D(f)$ the deterministic query complexity of $f$, which is also sometimes called decision tree complexity.

We start with some standard definitions for classical and quantum query complexity.

\begin{definition}[Decision tree complexity]
The \emph{decision tree complexity} of a function $f: \{0,1\}^N \to \{0,1, \bot\}$, also called the \emph{deterministic query complexity} of $f$ and denoted $\D(f)$, is the fewest number of queries made by any deterministic algorithm $\mathcal{A}(x)$ that satisfies, for all $x \in \Dom(f)$, $\mathcal{A}(x) = f(x)$.
\end{definition}

\begin{definition}[Quantum query complexity]
The \emph{(bounded-error) quantum query complexity} of a function $f: \{0,1\}^N \to \{0,1,\bot\}$, denoted $\Q(f)$, is the fewest number of queries made by any quantum query algorithm $\mathcal{A}(x)$ that satisfies, for all $x \in \Dom(f)$:
\begin{itemize}
\item If $f(x) = 1$, then $\Pr\left[\mathcal{A}(x) = 1 \right] \ge \frac{2}{3}$, and
\item If $f(x) = 0$, then $\Pr\left[\mathcal{A}(x) = 1 \right] \le \frac{1}{3}$.
\end{itemize}
\end{definition}

We define a notion of $\mathsf{QMA}$ (Quantum Merlin-Arthur) query complexity as follows.

\begin{definition}[$\mathsf{QMA}$ query complexity]
\label{def:qma_query_complexity}
The \emph{(bounded-error) $\mathsf{QMA}$ query complexity} of a function $f: \{0,1\}^N \to \{0,1,\bot\}$, denoted $\QMA(f)$, is the fewest number of queries made by any quantum query algorithm $\mathcal{V}(\ket{\psi}, x)$ that takes an auxiliary input state $\ket{\psi}$ and satisfies, for all $x \in \Dom(f)$:
\begin{itemize}
\item (Completeness) If $f(x) = 1$, then there exists a state $\ket{\psi}$ such that $\Pr\left[\mathcal{V}(\ket{\psi}, x) = 1 \right] \ge \frac{2}{3}$, and
\item (Soundness) If $f(x) = 0$, then for every state $\ket{\psi}$, $\Pr\left[\mathcal{V}(\ket{\psi}, x) = 1 \right] \le \frac{1}{3}$.
\end{itemize}
The algorithm $\mathcal{V}$ is sometimes called the \emph{verifier}, and the state $\ket{\psi}$ a \emph{witness}.
\end{definition}

Note that, in contrast to most previous works (c.f. \cite{RS04,AKKT20,ST19}), our definition of $\mathsf{QMA}$ query complexity completely ignores the number of qubits in the witness state $\ket{\psi}$. Our definition is more closely related to the \textit{quantum certificate complexity} $\QC(f)$ that was introduced by Aaronson \cite{Aar08}. The key difference between $\QMA(f)$ and $\QC(f)$ is that $\QMA(f)$ only requires the ability to query-efficiently verify $1$-inputs to the function, whereas $\QC(f)$ assumes the existence of a query-efficient verifier on both $0$- and $1$-inputs. Thus, one can view $\QMA(f)$ as a one-sided version of $\QC(f)$.\footnote{In other contexts, it might be preferable to denote ``one-sided quantum certificate complexity'' by $\QC_1(f)$, but in this work we exclusively use $\QMA(f)$ so as to emphasize the connection to quantum Merlin-Arthur protocols.

One can also show, as a consequence of \cite[Theorems 4 and 7]{Aar08}, that the verifier in the definition of $\QMA(f)$ can be replaced by a classical randomized verifier with perfect completeness, at the cost of a quadratic increase in the query complexity. We will not require this fact elsewhere in the paper, however.} %By a result of Aaronson \cite[Theorems 4 and 7]{Aar08}, the verifier in the definition of $\QMA(f)$ can be replaced by a classical randomized verifier with perfect completeness, at the cost of a quadratic increase in the query complexity. We will not require this fact elsewhere in the paper, however.

It is important to emphasize that $\mathsf{QMA}$ query complexity is not completely trivial: even though the witness can have unbounded length, the power of the verifier is still limited by the number of queries it makes. Indeed, in some cases, the proof does not help at all. For example, for the function $\AND_N$ on $N$ bits that outputs $1$ if all of the inputs are $1$, we have $\Q(\AND_N) = \QMA(\AND_N) = \Theta\left(\sqrt{N}\right)$, as observed by Raz and Shpilka \cite{RS04}.

We next define sensitivity, and the related $B$-block sensitivity.

\begin{definition}[Sensitivity]
The \emph{sensitivity} of a function $f: \{0,1\}^N \to \{0,1\}$ on input $x$ is defined as:
\[
\s^x(f) \coloneqq \left|\left\{i \in [n] : f(x) \neq f\left(x^{\oplus i}\right)\right\}\right|.
\]
The \emph{sensitivity of $f$} is defined as:
\[
\s(f) \coloneqq \max_{x \in \{0,1\}^N} \s^x(f).
\]
\end{definition}

\begin{definition}[$B$-block sensitivity]
\label{def:d_block_sensitivity}
Let $B = \{S_1,S_2,\ldots,S_k\}$ be collection of disjoint subsets of $[N]$. The \emph{block sensitivity} of a function $f: \{0,1\}^N \to \{0,1\}$ on input $x \in \{0,1\}^N$ with respect to $B$ is defined as:
\[
\bs_B^x(f) \coloneqq \left|\left\{i \in [k] : f(x) \neq f\left(x^{\oplus S_i}\right)\right\}\right|.
\]
The \emph{block sensitivity of $f$} with respect to $B$ is defined as:
\[
\bs_B(f) \coloneqq \max_{x \in \{0,1\}^N} \bs_B^x(f).
\]
\end{definition}

Note that sensitivity is the special case of $B$-block sensitivity in which $B$ is the partition into singletons.

We require a somewhat unusual definition of certificate complexity. Our definition agrees with the standard definition for \textit{total} functions, but may differ for partial functions.

\begin{definition}[Certificate complexity]
\label{def:certificate_complexity}
Let $f: \{0,1\}^N \to \{0,1,\bot\}$, and suppose $x \in \Dom(f)$. A \emph{certificate for $x$ on $f$}, also called an \emph{$f(x)$-certificate}, is a set $K \subseteq [N]$ such that for any $y \in \Dom(f)$ satisfying $y_i = x_i$ for all $i \in K$, $f(x) = f(y)$.

The \emph{certificate complexity of $f$ on $x$}, denoted $\C^x(f)$, is the minimum size of any certificate for $f$ on $x$. The \emph{certificate complexity of $f$} is defined as:
\[
\C(f) \coloneqq \max_{x \in \{0,1\}^N} \C^x(f).
\]
\end{definition}

Intuitively, in this definition of certificate complexity, a $b$-certificate for $b \in \{0,1\}$ witnesses that $f(x) \neq 1-b$, in contrast to the standard definition where a $b$-certificate witnesses that $f(x)=b$.

\subsection{Random Restrictions}
\label{sec:random_restrictions}
A \textit{restriction} is a function $\rho: [N] \to \{0,1,*\}$. A \textit{random restriction} with $\Pr[*] = p$ is a distribution over restrictions in which, for each $i \in [N]$, we independently sample:
\[
\rho(i) = \begin{cases}
0 & \text{ with probability } \frac{1-p}{2}\\
1 & \text{ with probability } \frac{1-p}{2}\\
* & \text{ with probability } p.
\end{cases}
\]

If $f: \{0,1\}^N \to \{0,1,\bot\}$ is a function and $\rho$ is a restriction where $S = \{i \in [N]: \rho(i) = *\}$, we denote by $f_\rho: \{0,1\}^S \to \{0,1,\bot\}$ the function obtained from $f$ by fixing the inputs where $\rho(i) \in \{0,1\}$. We call the remaining variables the \textit{unrestricted variables}. We sometimes apply restrictions to functions $f_i$ that have a subscript in the name, in which case we denote the restricted function by $f_{i|\rho}$ for notational clarity.

In this work, we also make use of \textit{projections}, which are a generalization of restrictions that were introduced in \cite{HRST17}. The exact definition of projections is unimportant for us, but intuitively, they are restrictions where certain unrestricted variables may be mapped to each other; see \cite{HRST17} for a more precise definition. We use the same notation $f_\rho$ for applying a projection $\rho$ to a function $f$ as we do for restrictions.

\subsection{Circuit Complexity}
\label{sec:circuit_complexity}
In this work, we consider Boolean circuits where the gates can be arbitrary partial Boolean functions $f: \{0,1\}^N \to \{0,1,\bot\}$. On input $x \in \{0,1,\bot\}^N$, a circuit gate labeled by $f$ evaluates to $b \in \{0,1\}$ if, for all $y \in \{0,1\}^N$ that extend $x$ (meaning, for all $i \in [n]$, $x_i \in \{0,1\}$ implies $y_i = x_i$), we have $f(y) = b$; otherwise, the gate evaluates to $\bot$.

We always specify the basis of gates allowed in Boolean circuits. Most commonly, we will consider Boolean circuits with $\AND$, $\OR$, and $\NOT$ gates where the $\AND$ and $\OR$ gates can have unbounded fan-in, but we will also consider e.g. circuits where the gates can be arbitrary functions of low $\mathsf{QMA}$ query complexity. 

The \textit{size} of a circuit is the number of gates of fan-in larger than $1$ (i.e. excluding $\NOT$ gates). The \textit{depth} of circuit is the length of the longest path from an input variable to the output gate, ignoring gates of fan-in $1$.

An $\AND$/$\OR$/$\NOT$ circuit is \textit{alternating} if all $\NOT$ gates are directly above the input, and all paths from the inputs to the output gate alternate between $\AND$ and $\OR$ gates. A circuit is \textit{layered} if for each gate $g$ in the circuit, the distance from $g$ to the output gate is the same along all paths. We denote by $\mathsf{AC^0}[s, d]$ the set of alternating, layered $\AND$/$\OR$/$\NOT$ circuits of size at most $s$ and depth at most $d$. By a folklore result, any $\AND$/$\OR$/$\NOT$ circuit can be turned into an alternating, layered circuit of the same depth at the cost of a small (constant multiplicative) increase in size.

A \textit{DNF formula}, also just called a DNF, is a depth-2 $\mathsf{AC^0}$ circuit where the top gate is an $\OR$ gate (i.e. the circuit is an $\OR$ of $\AND$s). The \textit{width} of a DNF is the maximum fan-in of any of the $\AND$ gates.

We require the following results in circuit complexity.
\begin{theorem}[{\cite[Lemma 7.8]{Has87}}]
\label{thm:hastad_parity_ac0_average}
For every constant $d$, there exists a constant $c$ such that for all sufficiently large $N$, for all $C \in \mathsf{AC^0}\left[2^{N^{c}},d \right]$, one has:
\[
\Pr_{x \sim \{0,1\}^N}\left[C(x) = \PARITY_N(x)\right] \le 0.6,
\]
where $\PARITY_N$ is the parity function on $N$ bits.
\end{theorem}

\begin{theorem}[{\cite[Proof of Theorem 10.1]{HRST17}}]
\label{thm:hrst}
For all constant $d \ge 2$ and all sufficiently large $m \in \Naturals$, there exists a function $\Sipser_d \in \mathsf{AC^0}\left[2^{\Theta(m)}, d\right]$ with $N = 2^{\Theta(m)}$ inputs, and a class $\mathcal{R}$ of random projections such that the following hold:
\begin{enumerate}[(a)]
\item For some value $b = 2^{-m}\left(1 - O\left(2^{-m/2}\right)\right)$, for any function $f: \{0,1\}^N \to \{0,1,\bot\}$,
\[
\Pr_{x \sim \{0,1\}^N}\left[f(x) = \Sipser_d(x)\right] = \Pr_{x \sim D, \rho \sim \mathcal{R}}\left[f_\rho(x) = \Sipser_{d|\rho} (x)\right],
\]
where $D$ is the distribution over bit strings in which each coordinate is $0$ with probability $b$ and $1$ with probability $1 - b$.
\end{enumerate}
Additionally, let $C \in \mathsf{AC^0}[s, d - 1]$. If we sample $\rho \sim \mathcal{R}$, then:
\begin{enumerate}[(a)]
\setcounter{enumi}{1}
\item Except with probability at most $s2^{-2^{m/2 - 4}}$, $\D\left(C_\rho\right) \le 2^{m/2 - 4}$.
\item Except with probability at most $O\left(2^{-m/2}\right)$, $\Sipser_{d|\rho}$ is reduced to an $\AND$ gate of fan-in $(\ln 2) \cdot 2^m \cdot \left(1 \pm O\left(2^{-m/4}\right)\right)$.
\end{enumerate}
\end{theorem}

An exact definition of the $\Sipser_d$ function is not important for us, but to help give some intuition, we mention a few of its other important properties. $\Sipser_d$ can be constructed as a depth-regular read-once monotone formula in which the gates at odd distance from the input are $\AND$ gates and the gates at even distance are $\OR$ gates. Additionally, for each $i \in [d]$, the gates at distance $i$ from the inputs all have the same fan-in $f_i$. These $f_i$s satisfy $f_1 = 2m$ and $f_i = 2^{\Theta(m)}$ for $i > 1$. This regularity allows for an appropriately scaled version of the $\Sipser_d$ function (or its negation) to be computed in $\mathsf{\Sigma}_{d-1}^\mathsf{P}$. Specifically, given an oracle $\mathcal{O}: \{0,1\}^{\lceil \log N \rceil} \to \{0,1\}$, and viewing the first $N$ bits of the truth table of $\mathcal{O}$ as an input $x$ to $\Sipser_d$, a  $\mathsf{\Sigma}_{d-1}^{\mathsf{P}^\mathcal{O}}$ machine can evaluate $\Sipser_d(x)$ (if $d$ is even, otherwise $1 - \Sipser_d(x)$ if $d$ is odd) in time $\polylog(N) = \poly(m)$. See \cite{RST15,HRST17} for further details.

\cite{HRST17} roughly explains the intuitive meaning of the above theorem as follows. Property (a) guarantees that the distribution $\mathcal{R}$ of random projections completes to the uniform distribution. Property (b) shows that the circuit $C$ simplifies with high probability under a random projection, while property (c) shows that $\Sipser_{d}$ retains structure under this distribution of projections with high probability. A simple corollary of these properties is the following:

\begin{corollary}[{\cite[Theorem 10.1]{HRST17}}]
\label{cor:hrst_ac0_depth_hierarchy}
Let $\Sipser_d$ be the function defined in \Cref{thm:hrst} on $N = 2^{\Theta(m)}$ bits. Let $C \in \mathsf{AC^0}[s, d-1]$. Then, for all sufficiently large $m$, we have:
\[
\Pr_{x \sim \{0,1\}^N}\left[C(x) = \Sipser_d(x)\right] \le \frac{1}{2} + O\left(2^{-m/4}\right) + s2^{-2^{m/2 - 4}}.
\]
\end{corollary}

Circuit complexity lower bounds are an indispensable tool for proving separations of relativized complexity classes, as was first observed by Furst, Saxe, and Sipser \cite{FSS84}. This connection can be formalized via the following lemma, which shows that the behavior of any $\mathsf{PH}$ oracle machine can be computed by an $\mathsf{AC^0}$ circuit whose inputs are the bits of the oracle string.

\begin{lemma}[{Implicit in \cite[Lemma 2.3]{FSS84}}]
\label{lem:furst-saxe-sipser}
Let $M$ be a $\mathsf{\Sigma}_k^{\mathsf{P}}$ oracle machine, and let $p(n)$ be a polynomial upper bound on the running time of $M$ on inputs of length $n$. Then for any $x \in \{0,1\}^n$, there exists a circuit $C \in \mathsf{AC^0}\left[2^{\poly(n)}, k+1\right]$ such that for any oracle $\mathcal{O}: \{0,1\}^* \to \{0,1\}$, we have:
\[
M^\mathcal{O}(x) = C\left(\mathcal{O}_{[p(n)]}\right),
\]
where $\mathcal{O}_{[p(n)]}$ denotes the concatenation of the bits of $\mathcal{O}$ on all strings of length at most $p(n)$.
\end{lemma}

Thus, lower bounds on $\mathsf{AC^0}$ circuit complexity give rise to oracle separations involving $\mathsf{PH}$. In particular, average-case lower bounds on the size of $\mathsf{AC^0}$ circuits can be used to construct separations relative to random oracles. For example, \Cref{thm:hastad_parity_ac0_average} implies that $\mathsf{P^{\# P}} \not\subset \mathsf{PH}$ relative to a random oracle \cite{Has87}, while \Cref{cor:hrst_ac0_depth_hierarchy} implies that $\mathsf{\Sigma}_{k+1}^\mathsf{P} \not\subset \mathsf{\Sigma}_k^\mathsf{P}$ relative to a random oracle \cite{HRST17,RST15}.

In most cases, when applying \Cref{lem:furst-saxe-sipser} to jump between $\mathsf{PH}$ oracle machines and $\mathsf{AC^0}$ circuits, we follow the convention of using $n$ to denote the length of an input to the $\mathsf{PH}$ machine, and $N$ to denote the (exponentially larger) size of the input to the corresponding $\mathsf{AC^0}$ circuit. For example, we might consider a $\mathsf{PH}$ machine that queries a function $f: \{0,1\}^{p(n)} \to \{0,1\}$, where $p(n) \le \poly(n)$. The truth table of $f$ can be interpreted as a string of length $N = 2^{p(n)}$. The corresponding $\mathsf{AC^0}$ circuit will have $N$ inputs and size $s \le 2^{\poly(n)}$. Thus, when we view $s$ as a function of the input size $N$ of the circuit, we have the bound $s \le \quasipoly(N)$.

\subsection{Other Background}

The form of the Raz-Tal Theorem stated below forms the basis for several of our results. It states that there exists a distribution that looks pseudorandom to small constant-depth $\mathsf{AC^0}$ circuits, but that is easily distinguishable from random by an efficient quantum algorithm.

\begin{theorem}[{\cite[Theorem 1.2]{RT19}}]
\label{thm:raz-tal}
For all sufficiently large $N$, there exists an explicit distribution $\mathcal{F}_N$ that we call the \emph{Forrelation distribution} over $\{0,1\}^N$ such that:
\begin{enumerate}
\item There exists a quantum algorithm $\mathcal{A}$ that makes $\polylog(N)$ queries and runs in time $\polylog(N)$ such that:
\[
\left|\Pr_{x \sim \mathcal{F}_N} [\mathcal{A}(x) = 1] - \Pr_{y \sim \{0,1\}^{N}} [\mathcal{A}(y) = 1]  \right| \ge 1 - \frac{1}{N^2}.
\]
\item For any $C \in \mathsf{AC^0}[\quasipoly(N),O(1)]$:
\[
\left|\Pr_{x \sim \mathcal{F}_N} [C(x) = 1] - \Pr_{y \sim \{0,1\}^{N}} [C(y) = 1]  \right| \le \frac{\polylog(N)}{\sqrt{N}}.
\]
\end{enumerate}
\end{theorem}
Note that, by standard amplification techniques, the $\frac{1}{N^2}$ in the above theorem can be replaced by any $\delta \le 2^{-\polylog(N)}$ at a cost of $\polylog(N)$ in the other parameters. For our purposes, the above theorem suffices as written. In some cases where $N$ is clear from context, we omit the subscript and write the distribution as $\mathcal{F}$. Additionally, in a slight abuse of notation, we sometimes informally call the decision problem of distinguishing a sample from $\mathcal{F}_N$ from a sample from the uniform distribution the $\Forrelation$ problem.

The next lemma was essentially shown in \cite{BBBV97}. We provide a proof for completeness.
\begin{lemma}
\label{lem:bbbv_query_magnitude}
Consider a quantum algorithm $Q$ that makes $T$ queries to $x \in \{0,1\}^N$. Write the state of the quantum algorithm immediately after $t$ queries to $x$ as:
\[\ket{\psi_t} = \sum_{i=1}^N \sum_{w} \alpha_{i,w,t}\ket{i,w},\]
where $w$ are indices over a workspace register. Define the \emph{query magnitude} $q_i$ of an input $i \in [N]$ by:
\[q_i \coloneqq \sum_{t=1}^{T} \sum_w \left| \alpha_{i,w,t} \right|^2. \]
Then for any $y \in \{0,1\}^N$, we have:
\[
\left|\Pr\left[Q(x) = 1 \right] - \Pr\left[Q(y) = 1 \right] \right| \le 8\sqrt{T} \cdot \sqrt{\sum_{i : x_i \neq y_i} q_i}.
\]
\end{lemma}

\begin{proof}
Denote by $\ket{\psi'_t}$ the state of the quantum algorithm after $t$ queries, where the first $t-1$ queries are to $x$ and the $t$th query is to $y$. For $t > 0$, we have:
\begin{align*}
|| \ket{\psi_t} - \ket{\psi'_t} || &= \left|\left| 2 \sum_{i : x_i \neq y_i} \sum_{w} \alpha_{i,w,t} \ket{i,w} \right|\right|\\
&= 2\sqrt{\sum_{i : x_i \neq y_i} \sum_{w} |\alpha_{i,w,t}|^2}.
\end{align*}

Hence, if we denote by $\ket{\varphi_t}$ the state of the quantum algorithm after $t$ queries, where all $t$ queries are to $y$, then:
\begin{align*}
\left|\Pr\left[Q(x) = 1 \right] - \Pr\left[Q(y) = 1 \right] \right| & \le 4||\ket{\psi_T} - \ket{\varphi_T}||\\ &\le \sum_{t=1}^T 8\sqrt{\sum_{i : x_i \neq y_i} \sum_{w} |\alpha_{i,w,t}|^2}\\
&\le 8\sqrt{T} \cdot \sqrt{\sum_{t=1}^T \sum_{i : x_i \neq y_i} \sum_w |\alpha_{i,w,t}|^2}\\
&= 8\sqrt{T} \cdot \sqrt{\sum_{i : x_i \neq y_i} q_i}.
\end{align*}
Above, the first line holds by \cite[Lemma 3.6]{BV97}; the second line is valid by the BBBV hybrid argument used in \cite[Theorem 3.3]{BBBV97}; the third line applies the Cauchy-Schwarz inequality, viewing the summation as the inner product between the all $1$s vector and the terms of the sum; and the last line substitutes the definition of $q_i$.
\end{proof}

\section{Consequences of the Raz-Tal Theorem}

In this section, we prove several oracle separations that build on the Raz-Tal Theorem (\Cref{thm:raz-tal}) and other known circuit lower bounds.

\subsection{Relativizing (Non-)Implications of \texorpdfstring{$\mathsf{NP}\subseteq\mathsf{BQP}$}{NP in BQP}}

Our first result proves the following:

\begin{theorem}
\label{thm:bqp=pp_ph_infinite}
There exists an oracle relative to which $\mathsf{BQP} = \mathsf{P^{\# P}}$ and $\mathsf{PH}$\ is infinite.
\end{theorem}

The proof idea is as follows. First, we take a random oracle, which makes $\mathsf{PH}$ infinite \cite{HRST17,RST15}. Then, we encode the answers to all possible $\mathsf{P^{\# P}}$ queries in instances of the $\Forrelation$ problem, allowing a $\mathsf{BQP}$ machine to efficiently decide any $\mathsf{P^{\# P}}$ language. We then leverage \Cref{thm:raz-tal} to argue that adding these $\Forrelation$ instances does not collapse $\mathsf{PH}$, because the $\Forrelation$ instances look random to $\mathsf{PH}$ algorithms. The formal proof is given below.

\begin{proof}[Proof of \Cref{thm:bqp=pp_ph_infinite}]
We will inductively construct this oracle $\mathcal{O}$, which will consist of two parts, $A$ and $B$. Denote the first part of the oracle $A$, and let this be a random oracle. For each $t\in\mathbb{N}$, we will add a region of $B$ called $B_t$ that will depend on the previously constructed parts of the oracle.  For convenience, we let $A_t$ denote the region of $A$ corresponding to inputs of length $t$, and we write $\mathcal{O}_t = (A_t, B_t)$.

Let $S_t$ be the set of all ordered pairs of the form $\langle M, x \rangle$ such that:
\begin{enumerate}
\item $M$ is a $\mathsf{P^{\# P}}$ oracle machine and $x$ is an input to $M$,
\item $\langle M, x \rangle$ takes less than $t$ bits to specify, and
\item $M$ is syntactically restricted to run in less than $t$ steps, and to query only the $\mathcal{O}_1, \mathcal{O}_2, \ldots, \mathcal{O}_{\lfloor \sqrt{t} \rfloor}$ regions of the oracle.
\end{enumerate}
Note that there are at most $2^t$ elements in $S_t$. Let $M_1, M_2, \ldots, M_{2^t}$ be an enumeration of $S_t$. For each $M_i$ in $S_t$, we add a function $f_i: \{0,1\}^{t^2} \to \{0,1\}$ into $B_t$. That is, we define $B_t: \{0,1\}^t \times \{0,1\}^{t^2} \to \{0,1\}$ by $B_t(i, x) \coloneqq f_i(x)$. The function $f_i$ is chosen subject to the following rules:
\begin{enumerate}
\item If $M_i$ accepts, then $f_i$ is drawn from the Forrelation distribution $\mathcal{F}_{2^{t^2}}$ (given in \Cref{thm:raz-tal}).
\item If $M_i$ rejects, then $f_i$ is uniformly random.
\end{enumerate}
Let $\mathcal{D}$ be the resulting distribution over oracles $\mathcal{O} = (A, B)$.

\begin{claim}
\label{claim:np_in_bqp}
$\mathsf{BQP}^\mathcal{O} = \mathsf{P}^{\mathsf{\# P}^\mathcal{O}}$ with probability $1$ over $\mathcal{O}$.
\end{claim}

\begin{proof}[Proof of Claim]
It suffices to show that $\mathsf{P}^{\mathsf{\# P}^\mathcal{O}} \subseteq \mathsf{BQP}^\mathcal{O}$, as the reverse containment holds relative to all oracles. Let $M$ be any $\mathsf{P}^{\mathsf{\# P}}$ oracle machine. Then given an input $x$ of size $n$, a quantum algorithm can decide whether $M^\mathcal{O}(x)$ accepts in $\poly(n)$ time by looking up the appropriate $B_t$, the one that contains a $\Forrelation$ instance $f_i: \{0,1\}^{t^2} \to \{0,1\}$ encoding the behavior of $\left<M,x\right>$, and then deciding whether $f_i$ is Forrelated or random by using the distinguishing algorithm $\mathcal{A}$ from \Cref{thm:raz-tal}.

In more detail, by \Cref{thm:raz-tal} we know that:
\[
\Pr_{\mathcal{O} \sim \mathcal{D}} \left[\mathcal{A}(f_i) \neq M^\mathcal{O}(x)\right] \le 2^{-2t^2},
\]
where the probability in the above expression is also taken over the randomness of $\mathcal{A}$. By Markov's inequality, we may conclude:
\[
\Pr_{\mathcal{O} \sim \mathcal{D}} \left[\Pr\left[\mathcal{A}(f_i) \neq M^\mathcal{O}(x)\right] \ge 1/3\right] \le 3 \cdot 2^{-2t^2}.
\]
Hence, the $\mathsf{BQP}$ promise problem defined by $\mathcal{A}$ agrees with the $\mathsf{P}^{\mathsf{\# P}}$ language on $\left<M,x\right>$, except with probability at most $3 \cdot 2^{-2t^2}$.

We now appeal to the Borel-Cantelli Lemma to argue that, with probability $1$ over $\mathcal{O} \sim \mathcal{D}$, $\mathcal{A}$ correctly decides $M^\mathcal{O}(x)$ for all but finitely many $x \in \{0,1\}^*$. Since there are at most $2^t$ inputs $\left<M,x\right>$ that take less than $t$ bits to specify, we have:
\[
\sum_{\langle x, M \rangle \in \{0,1\}^*} \Pr_{\mathcal{O} \sim \mathcal{D}}\left[\mathcal{A}^\mathcal{O} \text{ does not decide } M^\mathcal{O}(x)\right] \le \sum_{t=1}^{\infty} 2^t \cdot 3 \cdot 2^{-2t^2} < \infty
\]
Therefore, the probability that $\mathcal{A}$ fails on infinitely many inputs $\left<M,x\right>$ is $0$. Hence, $\mathcal{A}$ can be modified into a $\mathsf{BQP}$ algorithm that decides $M^\mathcal{O}(x)$ for \textit{all} $x \in \{0,1\}^*$, with probability $1$ over $\mathcal{O} \sim \mathcal{D}$.
\end{proof}

Now, we must show that $\mathsf{PH}^\mathcal{O}$ is infinite. We will accomplish this by proving, for all $k \in \Naturals$, $\mathsf{\Sigma}_{k}^{\mathsf{P}^\mathcal{O}}\neq\mathsf{\Sigma}_{k-1}^{\mathsf{P}^\mathcal{O}}$ with probability 1 over the choice of $\mathcal{O}$. Let $L^\mathcal{O}$ be the unary language used for the same purpose as in \cite{RST15,HRST17}. That is, $L^\mathcal{O}$ consists of strings $0^n$ such that, if we treat $n$ as an index into a portion of the random oracle $A_n$ that encodes a size-$2^n$ instance of the $\Sipser_{k+1}$ function, then that instance evaluates to $1$. By construction, $L^\mathcal{O} \in \mathsf{\Sigma}_{k}^{\mathsf{P}^\mathcal{O}}$ \cite{RST15,HRST17}. Furthermore, \cite{RST15,HRST17} show that $L^\mathcal{O}$ is not in $\mathsf{\Sigma}_{k-1}^{\mathsf{P}^A}$ with probability $1$ over the random oracle $A$. We need to argue that adding $B$ has probability $0$ of changing this situation. Fix any $\mathsf{\Sigma}_{k-1}^{\mathsf{P}^\mathcal{O}}$ oracle machine $M$. By the union bound, it suffices to show that

\[
\Pr_{\mathcal{O}\sim\mathcal{D}}\left[M^\mathcal{O} \text{ decides } L^\mathcal{O}\right]=0.
\]

Let $n_1<n_2<\cdots$ be an infinite sequence of input lengths, spaced far enough apart (e.g. $n_{i+1}=2^{n_i}$) such that $M\left(0^{n_i}\right)$ can query the oracle on strings of length $n_{i+1}$ or greater for at most finitely many values of $i$. Next, let
\[
p(M,i)\coloneqq\Pr_{\mathcal{O}\sim\mathcal{D}}\left[M^\mathcal{O} \text{ correctly decides } 0^{n_i}|M^\mathcal{O} \text{ correctly decided }0^{n_1},\dots,0^{n_{i-1}}\right]
\]
Then we have that
\[
\Pr_{\mathcal{O}\sim\mathcal{D}}\left[M^\mathcal{O}\text{ decides }L^\mathcal{O}\right]\leq\prod_{i=1}^\infty p(M,i).
\]

\noindent Thus it suffices to show that, for every fixed $M$, we have $p(M,i)\leq 0.7$ for all but finitely many $i$. To do this, we will consider a new quantity $q(M,i)$, which is defined exactly the same way as $p(M,i)$, except that now the oracle is chosen from a different distribution, which we call $D_i$. This $D_i$ is defined identically to $\mathcal{D}$ on $A$ and $B_1, \ldots, B_{n_i}$, but is uniformly random on $B_m$ for all $m > n_i$. It suffices to prove the following: for any fixed $M$,
\begin{enumerate}[(a)]
    \item $q(M,i)\leq 0.6$ for all but finitely many values of $i$, and
    \item $|q(M,i)-p(M,i)|\leq 0.1$ for all but finitely many values of $i$.
\end{enumerate}

Statement (a) essentially follows from the work of \cite{HRST17}. In more detail, the key observation is that the only portion of $\mathcal{O}$ that can depend on whether $0^{n_i}$ is in $L^\mathcal{O}$ is the input to the size-$2^{n_i}$ $\Sipser_{k+1}$ function that is encoded in $A$. All other portions of $\mathcal{O}$ are sampled independently from this region under $D_i$:
\begin{enumerate}
\item The rest of $A$ is sampled uniformly at random,
\item $B_1,\ldots,B_{n_i}$ are drawn from a distribution that cannot depend on any queries to $A$ on inputs of length $\lfloor \sqrt{n_i} \rfloor$ or greater (so in particular, they cannot depend on $A_{n_i}$), and
\item $B_{n_i+1}, B_{n_i + 2},\ldots$ are sampled uniformly at random.
\end{enumerate}
Hence, $M$ is forced to evaluate the size-$2^{n_i}$ $\Sipser_{k+1}$ function using only auxilliary and uncorrelated random bits. By the well-known connection between $\mathsf{\Sigma}_{k-1}^\mathsf{P}$ oracle machines and constant-depth circuits (\Cref{lem:furst-saxe-sipser}), $M$'s behavior on this size-$2^{n_i}$ string can be computed by an $\mathsf{AC^0}\left[2^{\poly(n_i)}, k\right]$ circuit. \Cref{cor:hrst_ac0_depth_hierarchy} shows that such a circuit correctly evaluates this $\Sipser_{k+1}$ function with probability greater than (say) 0.6 for at most finitely many $i$. This even holds conditioned on $M^\mathcal{O}$ correctly deciding $0^{n_1},\dots,0^{n_{i-1}}$, because the size-$2^{n_i}$ $\Sipser_{k+1}$ instance is chosen independently from the smaller instances, and because $M\left(0^{n_i}\right)$ can query the oracle on strings of length $n_{i+1}$ or greater for at most finitely many values of $i$.

For statement (b), we will prove this claim using a hybrid argument. We consider an infinite sequence of hybrids $\{D_{i,j} : j \in \Naturals\}$ between $D_i = D_{i,0}$ and $\mathcal{D}$, where in the $j$th hybrid $D_{i,j}$ we sample $A$ and $B_1,\ldots,B_{n_i + j}$ according to $\mathcal{D}$ and $B_{n_i + j + 1}, B_{n_i + j+2}, \ldots$ uniformly at random. The change between each $D_{i,j-1}$ and $D_{i,j}$ may be further decomposed into a sequence of smaller changes: from the uniform distribution $\mathcal{U}$ to the Forrelated $\mathcal{F}$, for each function $f: \{0,1\}^{(n_i + j)^2} \to \{0,1\}$ corresponding to a $\mathsf{P^{\# P}}$ oracle machine that happens to accept. 

Suppose we fix the values of $\mathcal{O}$ on everything except for $f$. \Cref{thm:raz-tal} implies that:
\begin{equation}
\label{eq:raz-tal-ph-no-triangle}
\left|\Pr_{f \sim \mathcal{F}}\left[M^\mathcal{O}\left(0^{n_i}\right) = 1\right]-\Pr_{f \sim \mathcal{U}}\left[M^\mathcal{O}\left(0^{n_i}\right) = 1 \right]\right| \le \frac{\poly(n_i)}{2^{(n_i+j)^2/2}}.
\end{equation}
This is because, again using \Cref{lem:furst-saxe-sipser}, there exists an $\mathsf{AC^0}\left[2^{\poly(n_i)} , k\right]$ circuit that takes the oracle string as input and evaluates to $M^\mathcal{O}\left(0^{n_i}\right)$. In fact, \eqref{eq:raz-tal-ph-no-triangle} \textit{also} holds even if the parts of $\mathcal{O}$ other than $f$ are not necessarily fixed, but are drawn from some distribution, by convexity (so long as the distribution is the same in both of the probabilities in \eqref{eq:raz-tal-ph-no-triangle}). In particular, using the fact that the $n_i$s are far enough apart for sufficiently large $i$, \eqref{eq:raz-tal-ph-no-triangle} also holds (for all sufficiently large $i$) when the parts of $\mathcal{O}$ other than $f$ are drawn from the distribution conditioned on $M^\mathcal{O}$ correctly deciding $0^{n_1},\dots,0^{n_{i-1}}$.

Now, recall that there are at most $2^{t}$ $\Forrelation$ instances in the $B_t$ part of the oracle. By the triangle inequality, bounding over each of these instances yields:
\[
\left|\Pr_{\mathcal{O} \sim D_{i, j-1}}\left[M^\mathcal{O}\left(0^{n_i}\right) = 1\right]-\Pr_{\mathcal{O} \sim D_{i, j}}\left[M^\mathcal{O}\left(0^{n_i}\right) = 1 \right]\right| \le 2^{n_i + j} \cdot \frac{\poly(n_i)}{2^{(n_i+j)^2/2}},
\]
where we implicitly condition on $M^\mathcal{O}$ correctly deciding $0^{n_1},\dots,0^{n_{i-1}}$ in both of the probabilities above, omitting it as written purely for notational simplicity. Hence, when we change \textit{all} of the hybrids, we obtain:
\begin{align*}
|q(M, i) - p(M, i)| &= \left|\Pr_{\mathcal{O} \sim D_i}\left[M^\mathcal{O}\left(0^{n_i}\right) = 1\right]-\Pr_{\mathcal{O} \sim \mathcal{D}}\left[M^\mathcal{O}\left(0^{n_i}\right) = 1 \right]\right|\\
&\le \sum_{j=1}^\infty 2^{n_i + j} \cdot \frac{\poly(n_i)}{2^{(n_i + j)^2 / 2}}\\
&\le \frac{\poly(n_i)}{2^{\Omega\left(n_i^2\right)}}\\
&\le 0.1
\end{align*}
for all but at most finitely many $i$.
\end{proof}

We conclude this section with a simple corollary.

\begin{corollary}
\label{cor:bqp_not_in_np/poly}
There exists an oracle relative to which $\mathsf{BQP}\not \subset
\mathsf{NP/poly}$.
\end{corollary}

\begin{proof}
It is known that for all oracles $\mathcal{O}$, $\mathsf{coNP}^\mathcal{O}\subset\mathsf{NP}^\mathcal{O}/\mathsf{poly}$ implies that $\mathsf{PH}^\mathcal{O}$ collapses to the third level \cite{Yap83}. %See also Feigenbaum-Fortnow "RANDOM-SELF-REDUCIBILITY OF COMPLETE SETS"
Let $\mathcal{O}$ be the oracle used in \Cref{thm:bqp=pp_ph_infinite}. Since $\mathsf{PH}^\mathcal{O}$ is infinite, $\mathsf{coNP}^\mathcal{O}\not\subset\mathsf{NP}^\mathcal{O}/\mathsf{poly}$. On the other hand, $\mathsf{coNP}^\mathcal{O}\subseteq\mathsf{BQP}^\mathcal{O}$, and hence $\mathsf{BQP}^\mathcal{O}\not\subset\mathsf{NP}^\mathcal{O}/\mathsf{poly}$.
\end{proof}

\subsection{Weak \texorpdfstring{$\mathsf{NP}$}{NP}, Strong \texorpdfstring{$\mathsf{BQP}$}{BQP}}

In this section, we prove the following:

\begin{theorem}
\label{thm:p=np_bqp=pp}
There exists an oracle relative to which $\mathsf{P}=\mathsf{NP}%
\neq\mathsf{BQP}=\mathsf{P}^{\mathsf{\#P}}$.
\end{theorem}

Note that, relative to any oracle, $\mathsf{P} = \mathsf{NP}$ implies $\mathsf{P} = \mathsf{PH}$. So, the Raz-Tal oracle separation of $\mathsf{BQP}$ and $\mathsf{PH}$ \cite{RT19} is necessary to prove \Cref{thm:p=np_bqp=pp}, in the sense that \Cref{thm:p=np_bqp=pp} is strictly stronger: any oracle $\mathcal{O}$ that satisfies \Cref{thm:p=np_bqp=pp} must also have $\mathsf{BQP}^\mathcal{O} \not\subset \mathsf{PH}^\mathcal{O}$. 

We follow a similar proof strategy to \Cref{thm:bqp=pp_ph_infinite}, with some additional steps. First, we take a random oracle, which separates $\mathsf{PH}$ from $\mathsf{P}^\mathsf{\#P}$ (morally, because $\PARITY$ is not approximable by $\mathsf{AC^0}$ circuits \cite{Has87}). We encode the answers to all possible $\mathsf{P}^\mathsf{\# P}$ queries in instances of the $\Forrelation$ problem, allowing a $\mathsf{BQP}$ machine to efficiently decide any $\mathsf{P}^\mathsf{\# P}$ language. Then, we add a region of the oracle that answers all $\mathsf{NP}$ queries, which collapses $\mathsf{PH}$ to $\mathsf{P}$. Finally, we leverage \Cref{thm:raz-tal} to argue that the $\Forrelation$ instances have no effect on the separation between $\mathsf{PH}$ and $\mathsf{P}^\mathsf{\#P}$, because the $\Forrelation$ instances look random to $\mathsf{PH}$ algorithms. The formal proof is given below.

\begin{proof}[Proof of \Cref{thm:p=np_bqp=pp}]
This oracle $\mathcal{O}$ will consist of three parts: a random oracle $A$, and oracles $B$ and $C$ that we will construct inductively. For each $t \in \Naturals$, we will add regions $B_t$ and $C_t$ that will depend on the previously constructed parts of the oracle. For convenience, we let $A_t$ denote the region of $A$ corresponding to inputs of length $t$, and we write $\mathcal{O}_t = (A_t, B_t, C_t)$.

We first describe $B$, which will effectively collapse $\mathsf{P}^\mathsf{\# P}$ to $\mathsf{BQP}$. For $t \in \Naturals$, let $S_t$ be the set of all ordered pairs of the form $\langle M, x \rangle$ such that:
\begin{enumerate}
\item $M$ is a $\mathsf{P}^\mathsf{\# P}$ oracle machine and $x$ is an input to $M$,
\item $\langle M, x \rangle$ takes less than $t$ bits to specify, and
\item $M$ is syntactically restricted to run in less than $t$ steps, and to query only the $\mathcal{O}_1, \mathcal{O}_2, \ldots, \mathcal{O}_{\lfloor \sqrt{t} \rfloor}$ regions of the oracle.
\end{enumerate}
Note that there are at most $2^t$ elements in $S_t$. Let $M_1, M_2, \ldots, M_{2^t}$ be an enumeration of $S_t$. For each $M_i$ in $S_t$, we add a function $f_i: \{0,1\}^{t^2} \to \{0,1\}$ into $B_t$. That is, we define $B_t: \{0,1\}^t \times \{0,1\}^{t^2} \to \{0,1\}$ by $B_t(i, x) \coloneqq f_i(x)$. The function $f_i$ is chosen subject to the following rules:
\begin{enumerate}
\item If $M_i$ accepts, then $f_i$ is drawn from the Forrelation distribution $\mathcal{F}_{2^{t^2}}$ (given in \Cref{thm:raz-tal}).
\item If $M_i$ rejects, then $f_i$ is uniformly random.
\end{enumerate}
We next describe $C$, which will effectively collapse $\mathsf{NP}$ to $\mathsf{P}$. For $t \in \Naturals$, define $T_t$ similarly to $S_t$, except that we take $\mathsf{NP}$ oracle machines instead of $\mathsf{P}^\mathsf{\# P}$ oracle machines. For each $M_i$ in $T_t$, we add a bit into $C_t$ that returns $M_i(x)$. That is, we define $C_t: \{0,1\}^t \to \{0,1\}$ by $C_t(i) \coloneqq M_i(x)$.

Let $\mathcal{D}$ be the resulting distribution over oracles $\mathcal{O} = (A, B, C)$. We will show that the statement of the theorem holds with probability $1$ over $\mathcal{O}$ sampled from $\mathcal{D}$.

\begin{claim}
\label{claim:p=np}
$\mathsf{P}^\mathcal{O} = \mathsf{NP}^\mathcal{O}$ with probability $1$ over $\mathcal{O}$.
\end{claim}

\begin{claim}
\label{claim:bqp=pp}
$\mathsf{BQP}^\mathcal{O} = \mathsf{P}^{\mathsf{\# P}^\mathcal{O}}$ with probability $1$ over $\mathcal{O}$.
\end{claim}

The proof of \Cref{claim:p=np} is trivial: given an $\mathsf{NP}^\mathcal{O}$ machine $M$ and input $x$, a polynomial time algorithm can decide $M(x)$ by simply looking up the bit in $C$ that encodes $M(x)$. The proof of \Cref{claim:bqp=pp} is identical to the proof of \Cref{claim:np_in_bqp} in \Cref{thm:bqp=pp_ph_infinite}, so we omit it.

To complete the proof, we will show that $\mathsf{NP}^\mathcal{O} \neq \mathsf{P}^{\mathsf{\# P}^\mathcal{O}}$ with probability $1$ over $\mathcal{O}$. Let $L^\mathcal{O}$ be the following unary language: $L^\mathcal{O}$ consists of strings $0^n$ such that, if we treat $n$ as an index into a portion of the random oracle $A_n$ of size $2^n$, then the parity of that length-$2^n$ string is $1$. By construction, $L^\mathcal{O} \in \mathsf{P}^{\mathsf{\# P}^\mathcal{O}}$. We will show that $L^\mathcal{O} \not\in \mathsf{NP}^\mathcal{O}$ with probability $1$ over $\mathcal{O}$.

Fix any $\mathsf{NP}^{\mathcal{O}}$ oracle machine $M$. By the union bound, it suffices to show that

\[
\Pr_{\mathcal{O}\sim\mathcal{D}}\left[M^\mathcal{O} \text{ decides } L^\mathcal{O}\right]=0.
\]

Let $n_1<n_2<\cdots$ be an infinite sequence of input lengths, spaced far enough apart (e.g. $n_{i+1}=2^{n_i}$) such that $M\left(0^{n_i}\right)$ can query the oracle on strings of length $n_{i+1}$ or greater for at most finitely many values of $i$. Next, let
\[
p(M,i)\coloneqq\Pr_{\mathcal{O}\sim\mathcal{D}}\left[M^\mathcal{O} \text{ correctly decides } 0^{n_i}|M^\mathcal{O} \text{ correctly decided }0^{n_1},\dots,0^{n_{i-1}}\right]
\]
Then we have that
\[
\Pr_{\mathcal{O}\sim\mathcal{D}}\left[M^\mathcal{O}\text{ decides }L^\mathcal{O}\right]\leq\prod_{i=1}^\infty p(M,i).
\]

\noindent Thus it suffices to show that, for every fixed $M$, we have $p(M,i)\leq 0.7$ for all but finitely many $i$. To do this, we will consider a new quantity $q(M,i)$, which is defined exactly the same way as $p(M,i)$, except that now the oracle is chosen from a different distribution, which we call $D_i$. This $D_i$ is defined identically to $\mathcal{D}$ on $A$, $C$, and $B_1, \ldots, B_{n_i}$, but is uniformly random on $B_m$ for all $m > n_i$. It suffices to prove the following: for any fixed $M$,
\begin{enumerate}[(a)]
    \item $q(M,i)\leq 0.6$ for all but finitely many values of $i$, and
    \item $|q(M,i)-p(M,i)|\leq 0.1$ for all but finitely many values of $i$.
\end{enumerate}

To prove these, we first need the following lemma, which essentially states that for any $t' \le \poly(t)$, any bit in $C_{t'}$ can be computed by a small (i.e. quasipolynomial in the input length) constant-depth $\mathsf{AC^0}$ circuit whose inputs do not depend on $C_i$ for any $i > t$.

\begin{lemma}
\label{lem:recursive_np_ac0}
Fix $t, d \in \Naturals$, and let $t' \le t^{2^d}$. For each $\langle M, x \rangle \in T_{t'}$, there exists an $\AND$/$\OR$/$\NOT$ circuit of size at most $2^{1 + t^{2^d}}$ and depth $2d$ that takes as input $A_1,A_2,\ldots,A_{t^{2^{d-1}}}$; $B_1,B_2,\ldots,B_{t^{2^{d-1}}}$; and $C_1,C_2,\ldots,C_{t}$, and outputs $M(x)$.
\end{lemma}

\begin{proof}[Proof of Lemma]
Assume $t \ge 2$ (otherwise the theorem is trivial). We proceed by induction on $d$. Consider the base case $d = 1$. By definition of $T_{t'}$, $\langle M, x \rangle$ is an $\mathsf{NP}$ oracle machine that runs in less than $t'$ steps and queries only the $\mathcal{O}_1,\mathcal{O}_2,\ldots,\mathcal{O}_{t}$ regions of the oracle, because $t' \le t^2$. Hence, $M(x)$ computes a function of certificate complexity (\Cref{def:certificate_complexity}) at most $t'$ in the bits of $\mathcal{O}_1,\mathcal{O}_2,\ldots,\mathcal{O}_{t}$. This function may thus be expressed as a DNF formula of width at most $t'$, which is in turn an $\AND$/$\OR$/$\NOT$ circuit of depth $2$ and size at most $2^{t'} + 1 \le 2^{t^2} + 1 \le 2^{1 + t^2}$.

For the inductive step, let $d \ge 2$. Similar to the base case, we use the definition of $T_{t'}$ to obtain a circuit of depth $2$ and size at most $2^{t'} + 1$ that takes as input $\mathcal{O}_1,\mathcal{O}_2,\ldots,\mathcal{O}_{t^{2^{d-1}}}$ and outputs $M(x)$. To complete the theorem, we use the inductive hypothesis to replace each of the inputs to this DNF formula from the regions $C_{t+1},C_{t+2},\ldots,C_{t^{2^{d-1}}}$ with the respective circuits that compute them. This yields a circuit of depth $2d$, and by the inductive hypothesis, the total number of gates in this circuit is at most:
\[
\left(2^{t'} + 1\right) + \sum_{i=t+1}^{t^{2^{d-1}}} 2^{i} \cdot 2^{1 + t^{2^{d-1}}},
\]
because for each $t+1 \le i \le t^{2^{d-1}}$, there are at most $2^i$ bits in $C_i$. The above quantity is upper bounded by:

\begin{align*}
\left(2^{t'} + 1\right) + \sum_{i=t+1}^{t^{2^{d-1}}} 2^{i} \cdot 2^{1 + t^{2^{d-1}}}
&\le \left(2^{t^{2^d}} + 1\right) + \sum_{i=t+1}^{t^{2^{d-1}}} 2^{i} \cdot 2^{1 + t^{2^{d-1}}}\\
&\le 2^{t^{2^d}} + \sum_{i=1}^{t^{2^{d-1}}} 2^{i} \cdot 2^{1 + t^{2^{d-1}}}\\
&\le 2^{t^{2^d}} + 2^{1 + t^{2^{d-1}}} \cdot 2^{1 + t^{2^{d-1}}}\\
&= 2^{t^{2^d}} + 2^{2 + 2t^{2^{d-1}}}\\
&\le 2^{t^{2^d}} + 2^{4t^{2^{d-1}}}\\
&\le 2^{t^{2^d}} + 2^{t^{2 + 2^{d-1}}}\\
&\le 2^{t^{2^d}} + 2^{t^{2^d}}\\
&= 2^{1 + t^{2^d}}.
\end{align*}
Above, the first inequality holds because $t' \le t^{2^d}$; the second inequality simply expands the range of the sum (which certainly increases the sum by at least $1$); the third inequality applies $\sum_{i=1}^j 2^i \le 2^{j+1}$; and the remaining inequalities hold because $t \ge 2$ and $d \ge 2$.
\end{proof}

%\begin{align*}
%\left(2^{t'} + 1\right) + \sum_{i=t+1}^{t^{2^{k-1}}} 2^{i} \cdot 2^{2t^{2^{k-1}}}
%&\le \left(2^{t^{2^k}} + 1\right) + \sum_{i=t+1}^{t^{2^{k-1}}} 2^{i} \cdot 2^{2t^{2^{k-1}}}\\
%&\le 2^{t^{2^k}} + \sum_{i=1}^{t^{2^{k-1}}} 2^{i} \cdot 2^{2t^{2^{k-1}}}\\
%&\le 2^{t^{2^k}} + 2^{1 + t^{2^{k-1}}} \cdot 2^{2t^{2^{k-1}}}\\
%&= 2^{t^{2^k}} + 2^{1 + 3t^{2^{k-1}}}\\
%&\le 2^{t^{2^k}} + 2^{4t^{2^{k-1}}}\\
%&\le 2^{t^{2^k}} + 2^{t^{2 + 2^{k-1}}}\\
%&\le 2^{t^{2^k}} + 2^{t^{2^k}}\\
%&\le 2^{2t^{2^k}}
%\end{align*}

Note that \Cref{lem:recursive_np_ac0} does not depend on the distribution of $A$ and $B$, but only on the way $C$ is defined recursively in terms of $A$ and $B$. Hence, it holds for both $\mathcal{O}$ drawn from $\mathcal{D}$ or drawn from any $D_i$. We also note that the circuit given in \Cref{lem:recursive_np_ac0} is not in $\mathsf{AC^0}$ normal form (i.e. it is not necessarily alternating and layered), but can be made so at the cost of a small increase in size.

Choosing specific parameters in \Cref{lem:recursive_np_ac0} gives the following simple corollary:
\begin{corollary}
\label{cor:recursive_np_ac0_M}
Fix an $\mathsf{NP}^\mathcal{O}$ oracle machine $M$. Let $p(n)$ be a polynomial upper bound on the running time of $M$ on inputs of length $n$, and also on the number of bits needed to specify $\langle M, 0^n \rangle$. Then there exists an $\mathsf{AC^0}\left[2^{p(n_i)^{O(1)}}, O(1) \right]$ circuit that takes as input $A$, $B$, and $C_1,\ldots,C_{n_i}$ and computes $M\left(0^{n_i}\right)$.
\end{corollary}

\begin{proof}[Proof of Corollary]
Let $t = n_i$ and $t' = p(n_i)^2$. Then $\langle M, 0^{n_i} \rangle \in T_{t'}$, because $M$ is restricted to query only $\mathcal{O}_1,\ldots,\mathcal{O}_{p(n)}$ by its time upper bound. Additionally, there exists $d = O(1)$ such that $t' \le t^{2^d}$ because $t' \le \poly(t)$. The corollary follows from \Cref{lem:recursive_np_ac0}.
\end{proof}

With \Cref{cor:recursive_np_ac0_M} in hand, the remainder of the proof closely follows the proof of \Cref{thm:bqp=pp_ph_infinite}. Statement (a) essentially follows from the work of \cite{Has87}.  In more detail, consider the circuit produced by \Cref{cor:recursive_np_ac0_M} that computes $M\left(0^{n_i}\right)$. The key observation is that the only portion of the input to this circuit that can depend on whether $0^{n_i}$ is in $L^\mathcal{O}$ is the input to the size-$2^{n_i}$ parity function that is encoded in $A$. All other portions of the input are sampled independently from this region under $D_i$:
\begin{enumerate}
\item The rest of $A$ is sampled uniformly at random,
\item $B_1,\ldots,B_{n_i}$ and $C_1,\ldots,C_{n_i}$ are drawn from a distribution that cannot depend on any queries to $A$ on inputs of length $\lfloor \sqrt{n_i} \rfloor$ or greater (so in particular, they cannot depend on $A_{n_i}$), and
\item $B_{n_i+1}, B_{n_i + 2},\ldots$ are sampled uniformly at random.
\end{enumerate}
Hence, $M$ is forced to evaluate the size-$2^{n_i}$ parity function using only auxiliary and uncorrelated random bits. By \Cref{thm:hastad_parity_ac0_average}, $M$ can do this with probability greater than 0.6 for at most finitely many $i$. This even holds conditioned on $M^\mathcal{O}$ correctly deciding $0^{n_1},\dots,0^{n_{i-1}}$, because the size-$2^{n_i}$ parity instance is chosen independently from the smaller instances, and because $M\left(0^{n_i}\right)$ can query the oracle on strings of length $n_{i+1}$ or greater for at most finitely many values of $i$.

For statement (b), we will prove this claim using a hybrid argument. We consider an infinite sequence of hybrids $\{D_{i,j} : j \in \Naturals\}$ between $D_i = D_{i,0}$ and $\mathcal{D}$, where in the $j$th hybrid $D_{i,j}$ we sample $A$, $C$, and $B_1,\ldots,B_{n_i + j}$ according to $\mathcal{D}$ and $B_{n_i + j + 1}, B_{n_i + j+2}, \ldots$ uniformly at random. The change between each $D_{i,j-1}$ and $D_{i,j}$ may be further decomposed into a sequence of smaller changes: from the uniform distribution $\mathcal{U}$ to the Forrelated $\mathcal{F}$, for each function $f: \{0,1\}^{(n_i + j)^2} \to \{0,1\}$ corresponding to a $\mathsf{P^{\# P}}$ oracle machine that happens to accept. 

Suppose we ``fix'' the values of $\mathcal{O}$ on everything except for $f$, in the following sense. We fix $A$ and $B$, except for some particular $f$ in $B$ that is allowed to vary. Then, we define $C$ recursively in terms of $A$ and $B$ in the usual, deterministic way (so that changing $f$ can affect $C$, but not the rest of $A$ and $B$). \Cref{thm:raz-tal} implies that:
\begin{equation}
\label{eq:raz-tal-ph-no-triangle-2}
\left|\Pr_{f \sim \mathcal{F}}\left[M^\mathcal{O}\left(0^{n_i}\right) = 1\right]-\Pr_{f \sim \mathcal{U}}\left[M^\mathcal{O}\left(0^{n_i}\right) = 1 \right]\right| \le \frac{\poly(n_i)}{2^{(n_i+j)^2/2}}.
\end{equation}
This is because, by \Cref{cor:recursive_np_ac0_M}, there exists an $\mathsf{AC^0}\left[2^{\poly(n_i)} , O(1)\right]$ circuit that takes $A$, $B$, and $C_1,\ldots,C_{n_i}$ as input and evaluates to $M^\mathcal{O}\left(0^{n_i}\right)$. In fact, \eqref{eq:raz-tal-ph-no-triangle-2} \textit{also} holds even if the parts of $A$, $B$, and $C_1,\ldots,C_{n_i}$ other than $f$ are not necessarily fixed, but are drawn from some distribution, by convexity (so long as the distribution is the same in both of the probabilities in \eqref{eq:raz-tal-ph-no-triangle-2}). In particular, using the fact that the $n_i$s are far enough apart for sufficiently large $i$, \eqref{eq:raz-tal-ph-no-triangle-2} also holds (for all sufficiently large $i$) when the parts of $A$, $B$, and $C_1,\ldots,C_{n_i}$ other than $f$ are drawn from the distribution conditioned on $M^\mathcal{O}$ correctly deciding $0^{n_1},\dots,0^{n_{i-1}}$.

Now, recall that there are at most $2^{t}$ $\Forrelation$ instances in the $B_t$ part of the oracle. By the triangle inequality, bounding over each of these instances yields:
\[
\left|\Pr_{\mathcal{O} \sim D_{i, j-1}}\left[M^\mathcal{O}\left(0^{n_i}\right) = 1\right]-\Pr_{\mathcal{O} \sim D_{i, j}}\left[M^\mathcal{O}\left(0^{n_i}\right) = 1 \right]\right| \le 2^{n_i + j} \cdot \frac{\poly(n_i)}{2^{(n_i+j)^2/2}},
\]
where we implicitly condition on $M^\mathcal{O}$ correctly deciding $0^{n_1},\dots,0^{n_{i-1}}$ in both of the probabilities above, omitting it as written purely for notational simplicity. Hence, when we change \textit{all} of the hybrids, we obtain:
\begin{align*}
|q(M, i) - p(M, i)| &= \left|\Pr_{\mathcal{O} \sim D_i}\left[M^\mathcal{O}\left(0^{n_i}\right) = 1\right]-\Pr_{\mathcal{O} \sim \mathcal{D}}\left[M^\mathcal{O}\left(0^{n_i}\right) = 1 \right]\right|\\
&\le \sum_{j=1}^\infty 2^{n_i + j} \cdot \frac{\poly(n_i)}{2^{(n_i + j)^2 / 2}}\\
&\le \frac{\poly(n_i)}{2^{\Omega\left(n_i^2\right)}}\\
&\le 0.1
\end{align*}
for all but at most finitely many $i$.
%The proof of statement (b) is essentially the same as the proof of the analogous statement in \Cref{thm:bqp=pp_ph_infinite}: we use \Cref{lem:recursive_np_ac0} to argue that there exists an $\mathsf{AC^0}\left[2^{\poly(n_i)},O(1)\right]$ circuit that takes the truth table of a $\Forrelation$ instance as input and evaluates to $M^\mathcal{O}\left(0^{n_i}\right)$, and then we invoke \Cref{thm:raz-tal} and a hybrid argument to argue that this circuit cannot detect the changes to $B_m$ for $m > n_i$ when switching between $\mathcal{D}$ and $D_i$.
\end{proof}

\section{Fine Control over \texorpdfstring{$\mathsf{BQP}$}{BQP}\ and \texorpdfstring{$\mathsf{PH}$}{PH}}
%\subsection{Generalizing \texorpdfstring{\cite{HRST17}}{[HRST17]}}
\subsection{\texorpdfstring{$\mathsf{BQP}^{\mathsf{PH}}$}{BQP\^{}PH} Lower Bounds for \texorpdfstring{$\Sipser$}{Sipser} Functions}
In this section, we prove that $\mathsf{BQP}^{\mathsf{\Sigma}_k^\mathsf{P}}$ does not contain $\mathsf{\Sigma}_{k+1}^\mathsf{P}$ relative to a random oracle, generalizing the known result that $\mathsf{PH}$ is infinite relative to a random oracle \cite{HRST17,RST15}.

We first require the following form of the BBBV Theorem \cite{BBBV97}. It essentially states that a quantum algorithm that makes few queries to its input is unlikely to detect small random changes to the input. Viewed another way, \Cref{lem:average_case_bbbv} is just a probabilistic version of \Cref{lem:bbbv_query_magnitude}.

\begin{lemma}
\label{lem:average_case_bbbv}
Consider a quantum algorithm $Q$ that makes $T$ queries to $x \in \{0,1\}^N$. Let $y \in \{0,1\}^N$ be drawn from some distribution such that, for all $i \in [N]$, $\Pr_{y}\left[x_i \neq y_i\right] \le p$. Then for any $r > 0$:
\[
\Pr_{y} \left[ \left|\Pr\left[Q(y) = 1 \right] - \Pr\left[Q(x) = 1 \right] \right| \ge r \right] \le \frac{64pT^2}{r^2}
\]
\end{lemma}
\begin{proof}
%Let $Q$ be the quantum query algorithm corresponding to $f$. 
By \Cref{lem:bbbv_query_magnitude}, we have that for any fixed $y$:
\[
\left|\Pr\left[Q(x) = 1 \right] - \Pr\left[Q(y) = 1 \right] \right| \le 8\sqrt{T} \cdot \sqrt{\sum_{i : x_i \neq y_i} q_i},
\]
where we recall the definition of the \textit{query magnitudes} $q_i$ which depend on the algorithm's behavior on input $x$, and which satisfy $\sum_{i=1}^n q_i = T$. This implies that:
\begin{align*}
\Pr_{y} \left[ \left|\Pr\left[Q(y) = 1 \right] - \Pr\left[Q(x) = 1 \right] \right| \ge r \right]
&\le \Pr_{y} \left[8\sqrt{T} \cdot \sqrt{\sum_{i : x_i \neq y_i} q_i} \ge r\right]\\
&= \Pr_{y} \left[\sum_{i : x_i \neq y_i} q_i \ge \frac{r^2}{64T} \right]\\
&\le \frac{64T}{r^2} \cdot \E_{y}\left[\sum_{i : x_i \neq y_i} q_i\right] \\
&= \frac{64T}{r^2} \cdot \sum_{i=1}^n q_i \cdot \Pr_{y}[y_i \neq x_i]\\
&\le \frac{64pT^2}{r^2},
\end{align*}
where the third line applies Markov's inequality (the $q_i$s are nonnegative), and the last two lines use linearity of expectation along with the fact that the $q_i$s sum to $T$.
\end{proof}

\begin{corollary}
\label{cor:average_case_bbbv_function}
Let $f: \{0,1\}^N \to \{0,1,\bot\}$ be a partial function with $\Q(f) \le T$. Fix $x \in \{0,1\}^N$, and let $y \in \{0,1\}^N$ be drawn from some distribution such that, for all $i \in [N]$, $\Pr_{y}\left[x_i \neq y_i\right] \le p$. Then for some $i \in \{0,1\}$, $\Pr_y[f(y) = i] \le 2304pT^2$.
\end{corollary}

\begin{proof}
Let $Q$ be the quantum query algorithm corresponding to $f$. We choose $i = 1$ if $\Pr\left[Q(x) = 1\right] \le \frac{1}{2}$, and $i = 0$ otherwise. Then the claim follows from \Cref{lem:average_case_bbbv} with $r = \frac{1}{6}$, just because $Q$ computes $f$ with error at most $\frac{1}{3}$.
\end{proof}

We now prove a query complexity version of the main result of this section.

\begin{theorem}
\label{thm:bqp_sigma_k_query_version}
Let $\Sipser_d$ be the function defined in \Cref{thm:hrst} for some choice of $d$, $m$, and $N$. Let $f: \{0,1\}^N \to \{0,1,\bot\}$ be computable by a depth-$d$ circuit of size $s$ in which the top gate has bounded-error quantum query complexity $T$, and all of the sub-circuits of the top gate are depth-$(d-1)$ $\mathsf{AC^0}$ circuits. Then:
\[
\Pr_{x \sim \{0,1\}^N}\left[f(x) = \Sipser_d(x)\right] \le \frac{1}{2} + O\left(2^{-m/4}\right) + s2^{-2^{m/2 - 4}} + 2304T^22^{-m/2 -4}.
\]
\end{theorem}

The proof of this theorem largely follows Theorem 10.1 of \cite{HRST17}, the relevant parts of which are quoted in \Cref{thm:hrst}. We take a distribution $\rho \sim \mathcal{R}$ of random projections with the property that (a) $\mathcal{R}$ completes to the uniform distribution, (b) $f_\rho$ simplifies with high probability over $\rho$, and (c) $\Sipser_{d|\rho}$ retains structure with high probability over $\rho$. Essentially the only difference compared to \cite{HRST17} is that we must incorporate the BBBV Theorem (in the form of \Cref{cor:average_case_bbbv_function}) in order to argue (b).

\begin{proof}[Proof of \Cref{thm:bqp_sigma_k_query_version}]
By \Cref{thm:hrst}(a),
\[
\Pr_{x \sim \{0,1\}^N}\left[f(x) = \Sipser_d(x)\right]
= \Pr_{x \sim D, \rho \sim \mathcal{R}}\left[f_\rho(x) = \Sipser_{d|\rho} (x)\right].
\]
\Cref{thm:hrst}(b) and a union bound over all of the sub-circuits of the top gate imply that, except with probability at most $s2^{-2^{m/2 - 4}}$ over $\rho \sim \mathcal{R}$, $f_\rho$ can be computed by a depth-$2$ circuit where the top gate has bounded-error quantum query complexity $T$, and all of the gates below the top gate have deterministic query complexity at most $2^{m/2-4}$. In this case, we say for brevity that $f_\rho$ ``simplifies''. By \Cref{thm:hrst}(c) and a union bound, except with probability at most $O\left(2^{-m/2}\right) + s2^{-2^{m/4 - 2}}$ over $\rho \sim \mathcal{R}$, $f_\rho$ simplifies and $\Sipser_{d|\rho}$ is an $\AND$ of fan-in $(\ln 2) \cdot 2^m \cdot \left(1 \pm O\left(2^{-m/4}\right)\right)$.

An $\AND$ of fan-in $(\ln 2) \cdot 2^m \cdot \left(1 \pm O\left(2^{-m/4}\right)\right)$ evaluates to $1$ with probability $\frac{1}{2}\left(1 \pm O\left(2^{-m/4}\right)\right)$ on an input sampled from $D$. On the other hand, if $f_\rho$ simplifies, then for each sub-circuit $C$ of the top gate, $\Pr_{x \sim D}\left[C(x) \neq C\left(1^{|x|}\right)\right] \le b2^{m/2-4}$, just because $\D(C) \le 2^{m/2-4}$ and each bit of $x$ is $0$ with probability at most $b$. Hence, by \Cref{cor:average_case_bbbv_function} with $p=b2^{m/2-4}$, if $f_\rho$ simplifies, then for some $i \in \{0,1\}$, $\Pr_{x \sim D}\left[f_\rho(x) = i\right] \le 2304bT^22^{m/2-4}$. Since $b \le 2^{-m}$, putting these together gives us that:
\[
\Pr_{x \sim D, \rho \sim \mathcal{R}}\left[f_\rho(x) = \Sipser_{d|\rho} (x)\right] \le \frac{1}{2} + O\left(2^{-m/4}\right) + s2^{-2^{m/2-4}} + 2304 T^22^{-m/2-4} .\qedhere
\]
\end{proof}

To complete this section, we require the following extension of Furst-Saxe-Siper \cite{FSS84} (\Cref{lem:furst-saxe-sipser}) to $\mathsf{BQP}^{\mathsf{\Sigma}_k^\mathsf{P}}$ machines.

\begin{proposition}
\label{prop:fss_bqp^ph}
Let $M$ be a $\mathsf{BQP}^{\mathsf{\Sigma}_k^\mathsf{P}}$ oracle machine (i.e. a pair $\langle A, B \rangle$ of a $\mathsf{BQP}$ oracle machine $A$ and a $\mathsf{\Sigma}_k^\mathsf{P}$ oracle machine $B$). Let $p(n)$ be a polynomial upper bound on the runtime of $A$ and $B$ on inputs of length $n$. Then for any $x \in \{0,1\}^n$, there is a depth-$(k+2)$ circuit $C$ of size at most $2^{\poly(n)}$ in which the top gate has bounded-error quantum query complexity at most $p(n)$, and all of the sub-circuits of the top gate are $\mathsf{AC^0}\left[2^{\poly(n)}, k+1\right]$ circuits, such that for any oracle $\mathcal{O}: \{0,1\}^* \to \{0,1\}$ we have:
\[
M^\mathcal{O}(x) = C\left( \mathcal{O}_{[p(p(n))]} \right),
\]
where $\mathcal{O}_{[p(p(n))]}$ denotes the concatenation of the bits of $\mathcal{O}$ on all strings of length at most $p(p(n))$.
\end{proposition}

\begin{proof}
For convenience, denote by $L$ the language decided by $B^\mathcal{O}$. $M^\mathcal{O}(x) = A^{L}(x)$ is a function of bounded-error quantum query complexity at most $p(n)$ in the bits of $L$. We take this function to be the top gate of our circuit, and use \Cref{lem:furst-saxe-sipser} to replace the inputs to this gate, the bits of $L$, with $\mathsf{AC^0}$ circuits.

Since $A^{L}(x)$ runs in time at most $p(n)$, it can only query the evaluation of $B^\mathcal{O}$ on inputs up to length at most $p(n)$. By \Cref{lem:furst-saxe-sipser}, for each $y \in \{0,1\}^{m}$ with $m \le p(n)$, $B^\mathcal{O}(y)$ is computed by an $\mathsf{AC^0}\left[s, k+1 \right]$ circuit for some $s \le 2^{\poly(p(n))} \le 2^{\poly(n)}$, where the inputs to this circuit are the bits of $\mathcal{O}$ on inputs of length at most $p(p(n))$. The resulting circuit obtained by composing the quantum gate with these $\mathsf{AC^0}$ circuits has total number of gates bounded above by:
\[
1 + \sum_{m=0}^{p(n)} 2^{m} \cdot 2^{\poly(m)} \le 2^{\poly(n)},
\]
and thus it satisfies the statement of the proposition.
\end{proof}

By standard techniques, we obtain the main result of this section, stated in terms of oracles instead of query complexity.

\begin{corollary}
\label{cor:bqp_sigma_k}
For all $k$, $\mathsf{\Sigma}_{k+1}^{\mathsf{P}^\mathcal{O}} \not\subset \mathsf{BQP}^{\mathsf{\Sigma}_k^{\mathsf{P}^\mathcal{O}}}$ with probability $1$ over a random oracle $\mathcal{O}$.
\end{corollary}

\begin{proof}
Let $d = k + 2$. Let $L^\mathcal{O}$ be the unary language used for the same purpose as in \cite{RST15,HRST17}. That is, $L^\mathcal{O}$ consists of strings $0^n$ such that, if we treat $n$ as an index into a portion of the random oracle $\mathcal{O}$ that encodes a size-$2^n$ instance of the $\Sipser_{d}$ function, then that instance evaluates to $1$. By construction, $L^\mathcal{O} \in \mathsf{\Sigma}_{k+1}^{\mathsf{P}^\mathcal{O}}$ \cite{RST15,HRST17}.

It remains to show that $L^\mathcal{O} \not\in \mathsf{BQP}^{\mathsf{\Sigma}_k^{\mathsf{P}^\mathcal{O}}}$. Fix a $\mathsf{BQP}^{\mathsf{\Sigma}_k^{\mathsf{P}}}$ oracle machine $M$. By the union bound, it suffices to show that
\[
\Pr_{\mathcal{O}}\left[M^\mathcal{O} \text{ decides } L^\mathcal{O}\right]=0.
\]

Let $n_1<n_2<\cdots$ be an infinite sequence of input lengths, spaced far enough apart (e.g. $n_{i+1}=2^{n_i}$) such that $M\left(0^{n_i}\right)$ can query the oracle on strings of length $n_{i+1}$ or greater for at most finitely many values of $i$. Next, let
\[
p(M,i)\coloneqq\Pr_{\mathcal{O}}\left[M^\mathcal{O} \text{ correctly decides } 0^{n_i}|M^\mathcal{O} \text{ correctly decided }0^{n_1},\dots,0^{n_{i-1}}\right]
\]
Then we have that
\[
\Pr_{\mathcal{O}}\left[M^\mathcal{O}\text{ decides }L^\mathcal{O}\right]\leq\prod_{i=1}^\infty p(M,i).
\]

\noindent Thus it suffices to show that, for every fixed $M$, we have $p(M,i)\leq 0.7$ for all but finitely many $i$. \Cref{prop:fss_bqp^ph} shows that $M$'s behavior on $0^{n_i}$ can be computed by a circuit of size $s \le 2^{\poly(n_i)}$ in which the top gate has bounded-error quantum query complexity $T \le \poly(n_i)$, and all of the sub-circuits of the top gate are depth-$(d-1)$ $\mathsf{AC^0}$ circuits. \Cref{thm:bqp_sigma_k_query_version} with $N = 2^{n_i}$ and $m = \Theta(n_i)$ shows that such a circuit correctly evaluates the $\Sipser_d$ function with probability greater than (say) $0.7$ for at most finitely many $i$. This even holds conditioned on $M^\mathcal{O}$ correctly deciding $0^{n_1},\ldots,0^{n_{i-1}}$, because the size-$2^{n_i}$ $\Sipser_d$ instance is chosen independently of the smaller instances, and because $M\left(0^{n_i}\right)$ can query the oracle on strings of length $n_{i+1}$ or greater for at most finitely many values of $i$.
\end{proof}

\subsection{\texorpdfstring{$\mathsf{BQP}^{\mathsf{PH}}$}{BQP\^{}PH} Lower Bounds for \texorpdfstring{$\OR \circ \Forrelation$}{OR of FORRELATIONs}}
\label{sec:implications_of_ac0_concentration}
In this section, we use tail bounds on the sensitivity of $\mathsf{AC^0}$ circuits to construct an oracle relative to which $\mathsf{NP}^\mathsf{BQP} \not\subset \mathsf{BQP}^\mathsf{PH}$. Such bounds are given implicitly in Section 3 of \cite{GSTW16}. For completeness, we derive our own bound on the sensitivity tails of $\mathsf{AC^0}$ circuits, though our bound is probably quantitatively suboptimal.

To prove that the sensitivity of $\mathsf{AC^0}$ circuits concentrates well, the first ingredient we need is a random restriction lemma for $\mathsf{AC^0}$ circuits. We use the following form, due to Rossman \cite{Ros17}.

\begin{theorem}[{\cite{Ros17}}]
\label{thm:rossman_ac0_restriction}
Let $f \in \mathsf{AC^0}[s, d]$, and let $\rho$ be a random restriction with $\Pr[*] = p$. Then for any $t > 0$:
\[
\Pr_\rho\left[\D(f_\rho) \ge t \right] \le \left(p \cdot O\left(\log s\right)^{d-1}\right)^t.
\]
\end{theorem}

With this in hand, it is straightforward to derive our sensitivity tail bound.

\begin{lemma}
\label{lem:ac0_sensitivity_tail_bound}
Let $f: \{0,1\}^N \to \{0,1\}$ be a circuit in $\mathsf{AC^0}[s, d]$. Then for any $t > 0$,
\[
\Pr_{x \sim \{0,1\}^N} \left[ \s^x(f) \ge t \right] \le 2N\cdot 2^{-\Omega\left(\frac{t}{(\log s)^{d-1}}\right)}.
\]
\end{lemma}

\begin{proof}
Let $\rho$ be a random restriction with $\Pr[*] = p$, for some $p$ to be chosen later. It will be convenient to view the choice of $\rho$ as follows: we choose a string $x \in \{0,1\}^N$ uniformly at random, and then we choose a set $S \subseteq [N]$ wherein each $i \in [N]$ is included in $S$ independently with probability $p$. Then, we take $\rho$ to be:
\[\rho(i) = \begin{cases} * & i \in S\\ x_i & i \not\in S. \end{cases}\]
Thus, by definition, it holds that $f_\rho(x|_S) = f(x)$.

Observe that for any fixed $x \in \{0,1\}^N$ and $j > 0$, we have:
\begin{align}
\s^x(f)
&=\frac{1}{p}\E_S\left[\s^{x|_S}(f_\rho) \right]\nonumber\\
&\le \frac{1}{p}\E_S\left[\s(f_\rho) \right]\nonumber\\
&\le \frac{1}{p}\E_S\left[\D(f_\rho) \right]\nonumber\\
&\le \frac{1}{p} \cdot \left(j + N \cdot \Pr_S\left[\D(f_\rho) \ge j\right] \right),\label{eq:sensitivity_restriction_bound}
\end{align}
where the first line holds because each sensitive bit of $f$ on $x$ is kept unrestricted with probability $p$; the second line holds by the definition of sensitivity; the third line holds by known relations between query measures; and the last line holds because $\D(f_\rho) \le N$ always holds.

With this in hand, we derive:
\begin{align*}
\Pr_x\left[\s^x(f) \ge t\right]
&\le \Pr_x\left[ \frac{1}{p} \cdot \left(j + N \cdot \Pr_S\left[\D(f_\rho) \ge j\right] \right) \ge t \right]\\
&= \Pr_x\left[\Pr_S\left[\D(f_\rho) \ge j\right] \ge \frac{pt - j}{N} \right]\\
&\le  \Pr_{x,S}\left[\D(f_\rho) \ge j\right] \cdot \frac{N}{pt - j}\\
&\le \left(p \cdot O\left(\log s\right)^{d-1}\right)^j \cdot \frac{N}{pt - j}.
\end{align*}
Above, the first line applies \eqref{eq:sensitivity_restriction_bound}; the third line holds by Markov's inequality; and the last line applies \Cref{thm:rossman_ac0_restriction}.

Choose $p = O\left(\log s\right)^{1-d}$ so that the above expression simplifies to $2^{-j} \cdot \frac{N}{pt - j}$. Then, set $j = pt - 1$ and the corollary follows.
\end{proof}

The sensitivity tail bound above immediately implies a tail bound on the \textit{block} sensitivity of $\mathsf{AC^0}$ circuits. We thank Avishay Tal for providing us with a proof of this fact.

\begin{corollary}
\label{cor:ac0_block_sensitivity_tail_bound}
Let $f: \{0,1\}^N \to \{0,1\}$ be a circuit in $\mathsf{AC^0}[s, d]$, and let $B = \{B_1,B_2,\ldots,B_k\}$ be a collection of disjoint subsets of $[N]$. Then for any $t$,
\[
\Pr_{x \sim \{0,1\}^N} \left[ \bs_B^x(f) \ge t \right] \le 4N\cdot 2^{-\Omega\left(\frac{t}{(\log (s + N))^{d}}\right)}.
\]
\end{corollary}
\begin{proof}
Consider the function $g: \{0,1\}^{N + k}$ defined by
\[
g(y, z) \coloneqq f\left(y \oplus z_1 \cdot B_1 \oplus z_2 \cdot B_2 \oplus \ldots \oplus z_k \cdot B_k\right),
\]
where $z_i \cdot B_i$ denotes the all zeros string if $z_i = 0$, and otherwise is the indicator string of $B_i$.

We claim that $g \in \mathsf{AC^0}[s + O(N), d+1]$. Let $x = y \oplus z_1 \cdot B_1 \oplus z_2 \cdot B_2 \oplus \ldots \oplus z_k \cdot B_k$. Notice that each bit of $x$ is either a bit in $y$, or else the $\XOR$ of a bit of $y$ with a bit of $z$. Hence, we can compute $g$ by feeding in at most $N$ $\XOR$ gates and their negations into $f$. The $\XOR$ function can be written as either an $\OR$ of $\AND$s or an $\AND$ of $\OR$s: $a \oplus b = (a \lor b) \land (\lnot a \lor \lnot b) = (a \land \lnot b) \lor (\lnot a \land  b)$. Hence, we can absorb one layer of $\AND$ or $\OR$ gates into the bottom layer of the circuit that computes $f$, thus obtaining a circuit of depth $d+1$.

Notice that for any $x \in \{0,1\}^N$, there are exactly $2^k$ strings $(y, z) \in \{0,1\}^{N + k}$ such that $x = y \oplus z_1 \cdot B_1 \oplus z_2 \cdot B_2 \oplus \ldots \oplus z_k \cdot B_k$. Moreover, for any such $(y, z)$ we have that $\bs_B^x(f) \le \s^{(y,z)}(g)$. Thus, \Cref{lem:ac0_sensitivity_tail_bound} implies that:
\[
\Pr_{x \sim \{0,1\}^N} \left[ \bs_B^x(f) \ge t \right] \le \Pr_{(y, z) \sim \{0,1\}^{N + k}} \left[ \s^{(y, z)}(g) \ge t \right] \le 2(N + k)\cdot 2^{-\Omega\left(\frac{t}{(\log (s + N))^{d}}\right)},
\]
and the corollary follows because $k \le N$.
\end{proof}

The rough idea of the proof going forward is as follows: an $\mathsf{NP}^\mathsf{BQP}$ machine can easily distinguish (1) a uniformly random $M \times N$ array of bits, and (2) an $M \times N$ array which contains a single row drawn from the Forrelation distribution, and is otherwise random. We want to show that a $\mathsf{BQP}^\mathsf{PH}$ machine cannot distinguish (1) and (2). To prove this, we first use our block sensitivity tail bound to argue in \Cref{lem:ac0_block_averaged_sensitivity_tail_bound} below that for most uniformly random strings $x$, an $\mathsf{AC^0}$ circuit is unlikely to detect a change to $x$ made by uniformly randomly resampling a single row of $x$. Then, we use the Raz-Tal Theorem to argue in \Cref{lem:ac0_single_forrelated_block_indistinguishable} that the same holds if we instead resample a single row of $x$ from the Forrelation distribution, rather than the uniform distribution. Finally, in \Cref{thm:bqp_ac0_single_forrelated_block_indistinguishable} we apply the BBBV Theorem to argue that a $\mathsf{BQP}^\mathsf{PH}$ oracle machine cannot distinguish cases (1) and (2).

\begin{lemma}
\label{lem:ac0_block_averaged_sensitivity_tail_bound}
Let $f: \{0,1\}^{MN} \to \{0,1\}$ be a circuit in $\mathsf{AC^0}[s, d]$. Let $x \in \{0,1\}^{MN}$ be an input, viewed as an $M \times N$ array with $M$ rows and $N$ columns. Let $y$ be sampled depending on $x$ as follows: uniformly select one of the rows of $x$, randomly reassign all of the bits of that row, and leave the other rows of $x$ unchanged. Then for any $\eps > 0$:
\[
\Pr_{x \sim \{0,1\}^{MN}}\left[\Pr_y\left[f(x) \neq f(y)\right] \ge \eps \right] \le 8M^2N \cdot 2^{-\Omega\left(\frac{\eps M}{(\log(s + MN))^d}\right)}.
\]
\end{lemma}

\begin{proof}
Let $\mathcal{B}$ be the distribution over collections $B = \{S_1,\ldots,S_M\}$ of subsets of $[MN]$ wherein each $S_i$ is a uniformly random subset of the $i$th row. Notice that for any fixed $x \in \{0,1\}^{MN}$ and $j > 0$, we have:
\begin{align}
\Pr_y\left[f(x) \neq f(y)\right]
&= \frac{1}{M} \cdot \E_{B \sim \mathcal{B}}\left[\bs^x_B(f)\right]\nonumber\\
&\le \frac{1}{M} \cdot \left(j + M \cdot \Pr_{B \sim \mathcal{B}}\left[\bs^x_B(f) \ge j\right] \right)\nonumber\\
&= \frac{j}{M} + \Pr_{B \sim \mathcal{B}}\left[\bs^x_B(f) \ge j\right],\label{eq:block_averaged_sensitivity_bound}
\end{align}
just because one can sample $y$ by drawing $B \sim \mathcal{B}$, $i \sim [M]$, and taking $y = x^{\oplus S_i}$. The inequality in the second line holds because $\bs^x_B(f) \le |B| = M$.

With this in hand, we derive:
\begin{align*}
\Pr_{x \sim \{0,1\}^{MN}}\left[\Pr_y\left[f(x) \neq f(y)\right] \ge \eps \right]
&\le \Pr_{x}\left[\frac{j}{M} + \Pr_{B \sim \mathcal{B}}\left[\bs^x_B(f) \ge j\right] \ge \eps \right]\\
&= \Pr_{x}\left[\Pr_{B \sim \mathcal{B}}\left[\bs^x_B(f) \ge j\right] \ge \eps - \frac{j}{M} \right]\\
&\le \Pr_{x,B}\left[\bs^x_B(f) \ge j\right] \cdot \frac{M}{\eps M - j}\\
&\le 4MN \cdot 2^{-\Omega\left(\frac{j}{(\log(s + MN))^d}\right)} \cdot \frac{M}{\eps M - j}\\
&\le \frac{4M^2N}{\eps M - j} \cdot 2^{-\Omega\left(\frac{j}{(\log(s + MN))^d}\right)}.
\end{align*}
Above, the first line applies \eqref{eq:block_averaged_sensitivity_bound}; the third line holds by Markov's inequality; and the fourth line applies \Cref{cor:ac0_block_sensitivity_tail_bound}. Choosing $j = \eps M - 1$ completes the proof.
\end{proof}

\begin{lemma}
\label{lem:ac0_single_forrelated_block_indistinguishable}
Let $M \le \quasipoly(N)$, and suppose that $f: \{0,1\}^{MN} \to \{0,1\}$ is a circuit in $\mathsf{AC^0}[\quasipoly(N), O(1)]$. Let $x \in \{0,1\}^{MN}$ be an input, viewed as an $M \times N$ array with $M$ rows and $N$ columns. Let $y$ be sampled depending on $x$ as follows: uniformly select one of the rows of $x$, randomly sample that row from the Forrelation distribution $\mathcal{F}_N$, and leave the other rows of $x$ unchanged. Then for some $\eps = \frac{\polylog(N)}{\sqrt{N}}$, we have:
\[
\Pr_{x \sim \{0,1\}^{MN}}\left[\Pr_y\left[f(x) \neq f(y)\right] \ge \eps \right] \le 8M^2N \cdot 2^{-\Omega\left(\frac{M}{\sqrt{N} \polylog(N)}\right)}.
\]

\end{lemma}

\begin{proof}
Consider a Boolean function $C(x, z, i)$ that takes inputs $x \in \{0,1\}^{MN}$, $z \in \{0,1\}^N$, and $i \in [M]$. Let $\tilde{y}$ be the string obtained from $x$ by replacing the $i$th row with $z$. Let $C$ output $1$ if $f(x) \neq f(\tilde{y})$, and $0$ otherwise. Clearly, $C \in \mathsf{AC^0}[\quasipoly(N), O(1)]$. Observe that for any fixed $x$:
\begin{equation}
\label{eq:forrelation_block_same_distribution}
\Pr_{i \sim [M],z \sim \mathcal{F}_N}\left[C(x, z, i) = 1\right] = \Pr_{y}[f(x) \neq f(y)].
\end{equation}

By \Cref{thm:raz-tal}, for some $\eps = \frac{\polylog(N)}{\sqrt{N}}$ we have:
\begin{equation}
\label{eq:forrelation_block_indistinguishable}
\left|\Pr_{i \sim [M],z \sim \mathcal{F}_N}\left[C(x, z, i) = 1\right] - \Pr_{i \sim [M],z \sim \{0,1\}^N}\left[C(x, z, i) = 1\right] \right| \le \frac{\eps}{2}.
\end{equation}

Putting these together, we obtain:
\begin{align*}
\Pr_{x \sim \{0,1\}^{MN}}\left[\Pr_y\left[f(x) \neq f(y)\right] \ge \eps \right]
&= \Pr_{x \sim \{0,1\}^{MN}}\left[\Pr_{i \sim [M], z \sim \mathcal{F}_N}\left[C(x, z, i) = 1\right] \ge \eps \right]\\
&\le \Pr_{x \sim \{0,1\}^{MN}}\left[\Pr_{i \sim [M], z \{0,1\}^N}\left[C(x, z, i) = 1\right] \ge \frac{\eps}{2} \right]\\
&= \Pr_{x \sim \{0,1\}^{MN}}\left[
\Pr_{i \sim [M],z \sim \{0,1\}^N}\left[f(x) \neq f(\tilde{y})\right]
\ge \frac{\eps}{2}
\right]\\
&\le 8M^2N \cdot 2^{-\Omega\left(\frac{\eps M}{(\log(s + MN))^d}\right)}\\
&\le 8M^2N \cdot 2^{-\Omega\left(\frac{M}{\sqrt{N} \polylog(N)}\right)},
\end{align*}
where the first line substitutes \eqref{eq:forrelation_block_same_distribution}; the second line holds by \eqref{eq:forrelation_block_indistinguishable} and the triangle inequality; the third line holds by the definition of $C$ and $\tilde{y}$ in terms of $i$ and $z$; the fourth line invokes \Cref{lem:ac0_block_averaged_sensitivity_tail_bound} for some $s = \quasipoly(N)$ and $d = O(1)$; and the last line uses these bounds on $s$ and $d$ along with the assumption that $M \le \quasipoly(N)$.
\end{proof}

The next theorem essentially shows that no $\mathsf{BQP}^\mathsf{PH}$ oracle machine can solve the $\OR \circ \Forrelation$ problem (i.e. given a list of $\Forrelation$ instances, decide if one of them is Forrelated, or if they are all uniform).

\begin{theorem}
\label{thm:bqp_ac0_single_forrelated_block_indistinguishable}
Let $M \le \quasipoly(N)$, and let $f: \{0,1\}^{MN} \to \{0,1,\bot\}$ be computable by a circuit of size $\quasipoly(N)$ in which the top gate has bounded-error quantum query complexity $T$, and all of the sub-circuits of the top gate are $\mathsf{AC^0}[\quasipoly(N), O(1)]$ circuits.

Let $b \sim \{0,1\}$ be a uniformly random bit. Suppose $z \in \{0,1\}^{MN}$ is sampled such that:
\begin{itemize}
\item If $b = 0$, then $z$ is uniformly random.
\item If $b = 1$, then a single uniformly chosen row of $z$ is sampled from the Forrelation distribution $\mathcal{F}_N$, and the remaining $M - 1$ rows of $z$ are uniformly random.
\end{itemize}
Then:
\[
\Pr_{b, z} [f(z) = b] \le \frac{1}{2} + \quasipoly(N) \cdot 2^{-\Omega\left(\frac{M}{\sqrt{N} \polylog(N)}\right)} + \frac{T^2 \polylog(N)}{\sqrt{N}}.
\]
\end{theorem}
\begin{proof}
We can think of sampling $z$ as follows. First, we choose a string $x_0 \sim \{0,1\}^{MN}$. Then, we sample $x_1$ by uniformly at random replacing one of the rows of $x_0$ with a sample from $\mathcal{F}_N$. Finally, we sample $b \sim \{0,1\}$ and set $z = x_b$.

Call a fixed $x_0$ ``bad'' if, for one of the sub-circuits $C$ of the top gate, we have $\Pr_{x_1}[C(x_0) \neq C(x_1)] \ge \eps$, where $\eps \le \frac{\polylog(N)}{\sqrt{N}}$ is the parameter given in \Cref{lem:ac0_single_forrelated_block_indistinguishable}. \Cref{lem:ac0_single_forrelated_block_indistinguishable}, combined with a union bound over the $\quasipoly(N)$-many such sub-circuits, implies that:
\[
\Pr_{x_0 \sim \{0,1\}^{MN}}\left[x_0 \text{ is bad}\right] \le \quasipoly(N) \cdot 2^{-\Omega\left(\frac{M}{\sqrt{N} \polylog(N)}\right)}.
\]

Clearly, it holds that:
\[
\Pr_{b,z}[f(z) = b] \le \Pr_{x_0 \sim \{0,1\}^{MN}}\left[x_0 \text{ is bad}\right] + \Pr_{x_1,b}\left[f(x_b) = b | x_0 \text{ is good} \right].
\]
$b$ is uniformly random, even conditioned on $x_0$ being good. On the other hand, \Cref{cor:average_case_bbbv_function} implies that for some $i \in \{0,1\}$ (depending on $x_0$), $\Pr_{x_1,b}\left[f(x_b) = i | x_0 \text{ is good} \right] \le 2304\eps T^2$. Thus, it holds that:
\[
\Pr_{x_1,b}\left[f(x_b) = b | x_0 \text{ is good} \right] \le \frac{1}{2} + \frac{T^2 \polylog(N)}{\sqrt{N}}.
\]
Putting these bounds together implies the statement of the theorem.
\end{proof}

Via standard complexity-theoretic techinques, \Cref{thm:bqp_ac0_single_forrelated_block_indistinguishable} implies the following oracle separation, which resolves the question of Fortnow \cite{For05}.

\begin{corollary}
\label{cor:np^bqp_not_in_bqp^ph}
There exists an oracle relative to which $\mathsf{NP}^{\mathsf{BQP}%
}\not \subset \mathsf{BQP}^{\mathsf{PH}}$.
\end{corollary}

\begin{proof}
We construct an oracle $A$ as follows. Let $L^A$ be a uniformly random unary language. For each $n \in \Naturals$, we add into $A$ a region consisting of a function $f_n: \{0,1\}^{2n^2} \to \{0,1\}$. Choose $f_n$ as follows:
\begin{itemize}
\item If $L^A\left(0^n\right) = 0$, then $f_n$ is uniformly random.
\item If $L^A\left(0^n\right) = 1$, then viewing the truth table of $f_n$ as consisting of $2^{n^2}$ rows of length $2^{n^2}$, we pick a single row at random and sample it from the Forrelation distribution $\mathcal{F}_{2^{n^2}}$, and sample the remaining $2^{n^2} - 1$ rows from the uniform distribution.
\end{itemize}
Let $\mathcal{D}$ be the resulting distribution over oracles $A$.

We first show that $L^A \in \mathsf{NP}^{\mathsf{BQP}^A}$ with probability $1$ over $A \sim \mathcal{D}$. To do so, we define a language $P^A$ as follows: for a string $x \in \{0,1\}^*$, $P^A(x) = 1$ if $|x| = n^2$ and the $x$th row of $f_n$ was drawn from the Forrelation distribution; otherwise $P^A(x) = 0$. Clearly, $L^A \in \mathsf{NP}^{P^A}$: to determine if $0^n \in L^A$, nondeterministically guess a string $x \in \{0,1\}^{n^2}$ and check if $x \in P^A$. Thus, it suffices to show that $P^A \in \mathsf{BQP}^A$, which we do below. (Note that this proof shares large parts with the proof of \Cref{claim:np_in_bqp}, only modifying a few parameters.)

\begin{claim}
\label{claim:m_in_bqp_o}
$P^A \in \mathsf{BQP}^A$ with probability $1$ over $A$.
\end{claim}

\begin{proof}[Proof of Claim]
Given an input $x$ of length $n^2$, a quantum algorithm can decide whether $x \in P^A$ in $\poly(n)$ time by looking up $x$th row of $f_n$, and then deciding whether it is Forrelated or random by using the distinguishing algorithm $\mathcal{A}$ from \Cref{thm:raz-tal}.

In more detail, let $g_x: \{0,1\}^{n^2} \to \{0,1\}$ denote the $x$th row of $f_n$ (i.e. $g_x(y) \coloneqq f_n(x, y)$). By \Cref{thm:raz-tal} we know that:
\[
\Pr_{A \sim \mathcal{D}} \left[\mathcal{A}(g_x) \neq P^A(x)\right] \le 2^{-2n^2},
\]
where the probability in the above expression is also taken over the randomness of $\mathcal{A}$. By Markov's inequality, we may conclude:
\[
\Pr_{A \sim \mathcal{D}} \left[\Pr\left[\mathcal{A}(g_x) \neq P^A(x)\right] \ge 1/3\right] \le 3 \cdot 2^{-2n^2}.
\]
Hence, the $\mathsf{BQP}$ promise problem defined by $\mathcal{A}$ agrees with $P^A$ on $x$, except with probability at most $3 \cdot 2^{-2n^2}$.

We now appeal to the Borel-Cantelli Lemma to argue that, with probability $1$ over $A$, $\mathcal{A}$ correctly decides $P^A(x)$ for all but finitely many $x \in \{0,1\}^*$. Since there are exactly $2^{n^2}$ inputs $x$ of length $\{0,1\}^{n^2}$, we have:
\[
\sum_{x \in \{0,1\}^*} \Pr_{A\sim \mathcal{D}}\left[\mathcal{A}^A \text{ does not decide } P^A(x)\right] \le \sum_{n=1}^{\infty} \sum_{x \in \{0,1\}^{n^2}} 3 \cdot 2^{-2n^2} \le \sum_{n=1}^{\infty} 2^{n^2} \cdot 3 \cdot 2^{-2n^2} < \infty.
\]

Therefore, the probability that $\mathcal{A}$ fails on infinitely many inputs $x$ is $0$. Hence, $\mathcal{A}$ can be modified into a $\mathsf{BQP}$ algorithm that decides $P^A(x)$ for \textit{all} $x \in \{0,1\}^*$, with probability $1$ over $A \sim \mathcal{D}$.
\end{proof}

It remains to show that $L^A \not\in \mathsf{BQP}^{\mathsf{PH}^A}$ with probability $1$ over $A$. As we will show, this follows from \Cref{thm:bqp_ac0_single_forrelated_block_indistinguishable} in the much same way that \Cref{cor:bqp_sigma_k} follows from \Cref{thm:bqp_sigma_k_query_version}. Fix a $\mathsf{BQP}^{\mathsf{\Sigma}_k^{\mathsf{P}}}$ oracle machine $M$. By the union bound, it suffices to show that
\[
\Pr_{A \sim \mathcal{D}}\left[M^A \text{ decides } L^A\right]=0.
\]

Let $n_1<n_2<\cdots$ be an infinite sequence of input lengths, spaced far enough apart (e.g. $n_{i+1}=2^{n_i}$) such that $M\left(0^{n_i}\right)$ can query the oracle on strings of length $n_{i+1}$ or greater for at most finitely many values of $i$. Next, let
\[
p(M,i)\coloneqq\Pr_{A \sim \mathcal{D}}\left[M^A \text{ correctly decides } 0^{n_i}|M^A \text{ correctly decided }0^{n_1},\dots,0^{n_{i-1}}\right]
\]
Then we have that
\[
\Pr_{A \sim \mathcal{D}}\left[M^A\text{ decides }L^A\right]\leq\prod_{i=1}^\infty p(M,i).
\]

\noindent Thus it suffices to show that, for every fixed $M$, we have $p(M,i)\leq 0.7$ for all but finitely many $i$. \Cref{prop:fss_bqp^ph} shows that $M$'s behavior on $0^{n_i}$ can be computed by a circuit of size $2^{\poly(n_i)}$ in which the top gate has bounded-error quantum query complexity $T \le \poly(n_i)$, and all of the sub-circuits of the top gate are $\mathsf{AC^0}\left[2^{\poly(n_i)}, k+1\right]$ circuits. \Cref{thm:bqp_ac0_single_forrelated_block_indistinguishable} with $M = N = 2^{n_i^2}$ shows that such a circuit correctly evaluates the $\OR \circ \Forrelation$ function with probability greater than (say) $0.7$ for at most finitely many $i$. This even holds conditioned on $M^A$ correctly deciding $0^{n_1},\ldots,0^{n_{i-1}}$, because the size-$2^{2n_i^2}$ $\OR \circ \Forrelation$ instance is chosen independently of the smaller instances, and because $M\left(0^{n_i}\right)$ can query the oracle on strings of length $n_{i+1}$ or greater for at most finitely many values of $i$.
\end{proof}

Using techniques analogous to \Cref{thm:p=np_bqp=pp}, we obtain the following stronger oracle separation.

\begin{corollary}
\label{cor:p=np_bqp!=qcma}
There exists an oracle relative to which $\mathsf{P} = \mathsf{NP}$ but $\mathsf{BQP} \neq \mathsf{QCMA}$.
\end{corollary}

\begin{proof}[Proof sketch]
This oracle $\mathcal{O}$ will consist of two parts: an oracle $A$ drawn from the same distribution as the oracle $A$ in \Cref{cor:np^bqp_not_in_bqp^ph}, and an oracle $B$ that we will construct inductively. For each $t \in \Naturals$, we add a region $B_t$ that will depend on the previously constructed parts of the oracle. For convenience, we denote by $A_t$ the region of $A$ corresponding to inputs of length $t$, and write $\mathcal{O}_t = (A_t, B_t)$.

Similarly to \Cref{thm:p=np_bqp=pp}, we define $S_t$ as the set of all $\mathsf{NP}$ machines that take less than $t$ bits to specify, run in at most $t$ steps, and query only the $\mathcal{O}_1,\ldots,\mathcal{O}_{\lfloor \sqrt{t} \rfloor}$ regions of the oracle. Then, we encode into $B_t$ the answers to all machines in $S_t$. This has the effect of making $\mathsf{P}^\mathcal{O} = \mathsf{NP}^\mathcal{O}$, as a polynomial-time algorithm can decide the behavior of any $\mathsf{NP}^\mathcal{O}$ machine $M$ by looking up the relevant bit in $B$ that encodes $M$'s behavior.

It remains to show that $\mathsf{BQP}^\mathcal{O} \neq \mathsf{QCMA}^\mathcal{O}$ with probability $1$ over $\mathcal{O}$. We achieve this by taking the language $L^A$ defined in \Cref{cor:np^bqp_not_in_bqp^ph}, which is clearly in $\mathsf{QCMA}^\mathcal{O}$, and showing that $L^A \not \in \mathsf{BQP}^\mathcal{O}$ with probability $1$ over $\mathcal{O}$.

Analogous to \Cref{lem:recursive_np_ac0}, one can show that for any $t' \le \poly(t)$, any bit of $B_{t'}$ can be computed by an $\mathsf{AC^0}\left[2^{\poly(t)}, O(1)\right]$ circuit whose inputs depend only on $A$ and $B_1,B_2,\ldots,B_t$. Hence, any $\mathsf{BQP}^\mathcal{O}$ machine that runs in time $\poly(t)$ can be computed by a circuit of size $2^{\poly(t)}$ in which the top gate has bounded-error quantum query complexity $\poly(t)$, all of the sub-circuits of the top gate are $\mathsf{AC^0}\left[2^{\poly(t)}, O(1)\right]$ circuits, and the inputs are $A$ and $B_1,B_2,\ldots,B_t$. In particular, if $t = n$, then all of the inputs are uncorrelated with $L^A(0^n)$, except for $A_{2n^2}$, the region whose $\OR \circ \Forrelation$ instance encodes $L^A(0^n)$. But in that case, we can again appeal to \Cref{thm:bqp_ac0_single_forrelated_block_indistinguishable} with $M = N = 2^{n^2}$ and $T = \poly(n)$ to argue that such a circuit correctly decides $L^A(0^n)$ with probability at most $0.7$ for infinitely many $n$.
\end{proof}

\subsection{\texorpdfstring{$\mathsf{PH}^{\mathsf{BQP}}$}{PH\^{}BQP} Lower Bounds for \texorpdfstring{$\Forrelation \circ \OR$}{FORRELATION of ORs}}
In this section, we construct an oracle relative to which $\mathsf{BQP}^\mathsf{NP} \not\subset \mathsf{PH}^\mathsf{PromiseBQP}$.

Within this section, for a string $z \in \{0,1\}^M$ and some choice of $N$, let $\mathcal{D}_{z,N}$ denote the following distribution over $\{0,1\}^{MN}$. View $x \sim \mathcal{D}_{z,N}$ as an $M \times N$ array of bits sampled as follows: if $z_i = 0$, then the $i$th row of $x$ is all $0$s, while if $z_i = 1$, then the $i$th row of $x$ has a single $1$ chosen uniformly at random and $0$s everywhere else.

Our first key lemma shows that, for a string $x \sim \mathcal{D}_{z,N}$, a quantum algorithm that queries $x$ can be efficiently simulated by a classical query algorithm, with high probability over $x$. As a warmup, we start with a version of this lemma in which the quantum algorithm makes only a single query.

\begin{lemma}
\label{lem:sparse_aaronson_ambainis_single_query}
Consider a quantum algorithm $Q$ that makes $1$ query to $x \in \{0,1\}^{MN}$ to produce a state $\ket{\psi}$. Then for any $K \in \Naturals$, there exists a deterministic classical algorithm that makes $K$ queries to $x$, and outputs a description of a state $\ket{\varphi}$ such that for any $\alpha \ge \sqrt{\frac{8}{N}}$ and any $z \in \{0,1\}^M$:
\[
\Pr_{x \sim \mathcal{D}_{z,N}} \left[|| \ket{\psi} - \ket{\varphi}|| \ge \alpha \right] \le e^{-\frac{\alpha^4 K}{32}}.
\]
\end{lemma}

\begin{proof}
Call $\ket{\psi_0}$ the initial state of $Q$, and suppose that $Q$ queries the phase oracle $U_x$ corresponding to $x$, then applies a unitary $W$, so that the output state of the algorithm is
\[
\ket{\psi} = W U_x \ket{\psi_0}.
\]
Without loss of generality, we may assume that $W$ is the identity, because
\[
|| \ket{\psi} - \ket{\varphi}|| = || W\ket{\psi} - W\ket{\varphi}||.
\]

Analogous to \Cref{lem:bbbv_query_magnitude}, let $q_{i,j}$ be the query magnitude (i.e. probability) with which $Q$ queries $x_{i,j}$ during its single query. That is, if the initial state $\ket{\psi_0}$ of $Q$ has the form:
\[
\ket{\psi_0} = \sum_{i=1}^M \sum_{j=1}^N \alpha_{i,j,w} \ket{i,j,w},
\]
where $w$ are indices over a workspace register, then
\[
q_{i,j} \coloneqq \sum_w |\alpha_{i,j,w}|^2,
\]
so that $\sum_{i=1}^M \sum_{j=1}^N q_{i,j} = 1$.

The classical algorithm is simply the following: query every $x_{i,j}$ such that $q_{i,j} \ge \frac{1}{K}$. Clearly there are at most $K$ such bits, so the algorithm makes at most $K$ queries. Then, calculate $Q$'s post-query state, assuming that all of the unqueried bits are $0$. Call this state $\ket{\varphi}$.

We now argue that the classical algorithm achieves the desired approximation to $\ket{\psi}$ with the correct probability. Fix some $z \in \{0,1\}^M$. For each row $i$ with $z_i = 1$, let $j(i, x)$ be the unique column $j$ such that $x_{i,j} = 1$. Now define a random variable $w(i,x)$ that measures the contribution of row $i$ to the error of our classical simulation:
\begin{equation}
\label{eq:wij_def}
w(i,x) \coloneqq \begin{cases}
q_{i,j(i,x)} & z_i = 1 \text{ and } q_{i,j(i,x)} < \frac{1}{K},\\
0 & \text{otherwise}.
\end{cases}
\end{equation}

Note that the $w(i,x)$'s are independent random variables, and also satisfy
\begin{align*}
\E_x \left[\sum_{i=1}^M w(i,x) \right] &\le \E_x \left[\sum_{i: z_i = 1} q_{i,j(i,x)} \right]\\
&= \sum_{i=1}^M \sum_{j=1}^N \Pr[x_{i,j} = 1] q_{i,j}\\
&\le \sum_{i=1}^M \sum_{j=1}^N \frac{q_{i,j}}{N} \\
&\le \frac{1}{N}
\end{align*}
by \eqref{eq:wij_def} and linearity of expectation. We also have $w(i, x) \le \frac{1}{K}$ for all $i$, but we will actually need a stronger upper bound: namely $w(i, x) \le m_i$, where
\[
m_i \coloneqq \min\left\{\frac{1}{K},\ \max_j q_{i,j}\right\}.
\]
Note that $m_i \le \frac{1}{K}$ for all $i$, and also that
\[
\sum_{i=1}^M m_i \le \sum_{i=1}^M \sum_{j=1}^N q_{i,j} = 1,
\]
which together imply that
\begin{equation}
\label{eq:sum_mi_squares}
\sum_{i=1}^M m_i^2 \le \sum_{i=1}^M m_i \cdot \frac{1}{K} \le \frac{1}{K}.
\end{equation}

Recall that we wish to bound the distance between $\ket{\psi}$ and $\ket{\varphi}$. \eqref{eq:wij_def} implies that
\begin{equation}
|| \ket{\psi} - \ket{\varphi}|| = 2\sqrt{\sum_{i=1}^M w(i,x)},
\label{eq:euclidean_dist}
\end{equation}
and therefore
\[
\Pr_{x \sim \mathcal{D}_{z,N}} \left[|| \ket{\psi} - \ket{\varphi}|| \ge \alpha \right]
= \Pr_{x \sim \mathcal{D}_{z,N}} \left[\sum_{i=1}^M w(i,x) \ge \frac{\alpha^2}{4} \right].
\]
We finally appeal to Hoeffding's inequality (\Cref{fact:hoeffding}) to bound this quantity. Set $\mu \coloneqq \frac{1}{N}$ and $\delta \coloneqq \frac{\alpha^2 N}{4} - 1$. Then
\begin{align*}
\Pr_{x \sim \mathcal{D}_{z,N}} \left[|| \ket{\psi} - \ket{\varphi}|| \ge \alpha \right]
&= \Pr_{x \sim \mathcal{D}_{z,N}} \left[\sum_{i=1}^M w(i,x) \ge (1 + \delta)\mu \right]\\
&\le \exp\left(-\frac{2\delta^2 \mu^2}{\sum_{i=1}^M m_i^2} \right)\\
&\le \exp\left(-\frac{2\left(\frac{\alpha^2N}{4} - 1\right)^2 \left(\frac{1}{N}\right)^2}{\frac{1}{K}} \right)\\
&\le \exp\left(-\frac{2\left(\frac{\alpha^2N}{8}\right)^2 \left(\frac{1}{N}\right)^2}{\frac{1}{K}} \right)\\
&= e^{-\frac{\alpha^4 K}{32}},
\end{align*}
where the first line applies \eqref{eq:euclidean_dist}; the second line applies \Cref{fact:hoeffding}; the third line substitutes \eqref{eq:sum_mi_squares}; and the fourth line uses the assumption that $\alpha \ge \sqrt{\frac{8}{N}}$.
\end{proof}

Ultimately, we will want to apply \Cref{lem:sparse_aaronson_ambainis_single_query} many times in succession to simulate the output of quantum algorithms that make multiple queries to a string $x \sim \mathcal{D}_{z, N}$. To do so, we need the following strengthening of \Cref{lem:sparse_aaronson_ambainis_single_query} whose proof is nearly the same. It shows that a similar statement to \Cref{lem:sparse_aaronson_ambainis_single_query} holds if we condition on knowing the values of $x$ at a few locations. Below, the set $S$ captures the previously queried locations, and $f$ records their values.\footnote{An earlier version of this manuscript did not contain this generalization of \Cref{lem:sparse_aaronson_ambainis_single_query}, and consequently some proofs of the following lemmas were erroneous. We thank Chinmay Nirkhe for pointing this out to us.}

\begin{lemma}
\label{lem:sparse_aaronson_ambainis_single_query_2}
Consider a quantum algorithm $Q$ that makes $1$ query to $x \in \{0,1\}^{MN}$ to produce a state $\ket{\psi}$. Let $S \subseteq [M] \times [N]$ be such that for all $i \in [M]$, either $|\{j: (i, j) \in S\}| \le \frac{N}{2}$ or $|\{j: (i, j) \in S\}| = N$. Let $f: S \to \{0,1\}$. Then for any $K \in \Naturals$, there exists a deterministic classical algorithm that, given $f$, makes $K$ queries to $x$, and outputs a description of a state $\ket{\varphi}$ such that for any $\alpha \ge \sqrt{\frac{16}{N}}$ and any $z \in \{0,1\}^M$:
\[
\Pr_{x \sim \mathcal{D}_{z,N}} \left[|| \ket{\psi} - \ket{\varphi}|| \ge \alpha \mid x_{i,j} = f(i, j)\ \forall (i, j) \in S\right] \le e^{-\frac{\alpha^4 K}{32}}.
\]
\end{lemma}

\begin{proof}
For notational simplicity, instead of fully writing out $x_{i,j} = f(i, j)\ \forall (i, j) \in S$ in conditional probabilities, we abbreviate the condition by $f$. For example, in this notation, the statement we want to prove is that
\[
\Pr_{x \sim \mathcal{D}_{z,N}} \left[|| \ket{\psi} - \ket{\varphi}|| \ge \alpha \mid f\right] \le e^{-\frac{\alpha^4 K}{32}}.
\]

Call $\ket{\psi_0}$ the initial state of $Q$, and suppose that $Q$ queries the phase oracle $U_x$ corresponding to $x$, then applies a unitary $W$, so that the output state of the algorithm is
\[
\ket{\psi} = W U_x \ket{\psi_0}.
\]
Without loss of generality, we may assume that $W$ is the identity, because
\[
|| \ket{\psi} - \ket{\varphi}|| = || W\ket{\psi} - W\ket{\varphi}||.
\]

Analogous to \Cref{lem:bbbv_query_magnitude}, let $q_{i,j}$ be the query magnitude (i.e. probability) with which $Q$ queries $x_{i,j}$ during its single query. That is, if the initial state $\ket{\psi_0}$ of $Q$ has the form:
\[
\ket{\psi_0} = \sum_{i=1}^M \sum_{j=1}^N \alpha_{i,j,w} \ket{i,j,w},
\]
where $w$ are indices over a workspace register, then
\[
q_{i,j} \coloneqq \sum_w |\alpha_{i,j,w}|^2,
\]
so that $\sum_{i=1}^M \sum_{j=1}^N q_{i,j} = 1$.

The classical algorithm is simply the following: query every $x_{i,j}$ such that $q_{i,j} \ge \frac{1}{K}$. Clearly there are at most $K$ such bits, so the algorithm makes at most $K$ queries. Then, calculate $Q$'s post-query state, assuming that all of the unqueried bits are $0$, except for any $(i, j) \in S$ for which $f(i, j) = 1$. Call this state $\ket{\varphi}$.

We now argue that the classical algorithm achieves the desired approximation to $\ket{\psi}$ with the correct conditional probability. Fix some $z \in \{0,1\}^M$. For each row $i$ with $z_i = 1$, let $j(i, x)$ be the unique column $j$ such that $x_{i,j} = 1$. Now define a random variable $w(i,x)$ that measures the contribution of row $i$ to the error of our classical simulation:
\begin{equation}
\label{eq:wij_def_2}
w(i,x) \coloneqq \begin{cases}
q_{i,j(i,x)} & z_i = 1 \text{ and } q_{i,j(i,x)} < \frac{1}{K} \text{ and } (i, j(i, x)) \not \in S,\\
0 & \text{otherwise}.
\end{cases}
\end{equation}

Note that the $w(i,x)$'s are independent random variables (even conditioned on $f$), and also satisfy
\begin{align*}
\E_x \left[\sum_{i=1}^M w(i,x) \mid f\right] &\le \E_x \left[\sum_{i: z_i = 1 \land (i, j(i, x)) \not\in S} q_{i,j(i,x)} \mid f \right]\\
&\le \sum_{i=1}^M \sum_{j : (i, j) \not\in S} \Pr[x_{i,j} = 1 \mid f]q_{i,j}\\
&\le \sum_{i=1}^M \sum_{j=1}^N \frac{2q_{i,j}}{N}\\
&\le \frac{2}{N}
\end{align*}
by \eqref{eq:wij_def_2} and linearity of expectation, using in the third line the assumption that either $|\{j: (i, j) \in S\}| \le \frac{N}{2}$ or $|\{j: (i, j) \in S\}| = N$. We also have $w(i, x) \le \frac{1}{K}$ for all $i$, but we will actually need a stronger upper bound: namely $w(i, x) \le m_i$, where
\[
m_i \coloneqq \min\left\{\frac{1}{K},\ \max_j q_{i,j}\right\}.
\]
Note that $m_i \le \frac{1}{K}$ for all $i$, and also that
\[
\sum_{i=1}^M m_i \le \sum_{i=1}^M \sum_{j=1}^N q_{i,j} = 1,
\]
which together imply that
\begin{equation}
\label{eq:sum_mi_squares_2}
\sum_{i=1}^M m_i^2 \le \sum_{i=1}^M m_i \cdot \frac{1}{K} \le \frac{1}{K}.
\end{equation}

Recall that we wish to bound the distance between $\ket{\psi}$ and $\ket{\varphi}$. \eqref{eq:wij_def_2} implies that
\begin{equation}
|| \ket{\psi} - \ket{\varphi}|| = 2\sqrt{\sum_{i=1}^M w(i,x)},
\label{eq:euclidean_dist_2}
\end{equation}
and therefore
\[
\Pr_{x \sim \mathcal{D}_{z,N}} \left[|| \ket{\psi} - \ket{\varphi}|| \ge \alpha \mid f\right]
= \Pr_{x \sim \mathcal{D}_{z,N}} \left[\sum_{i=1}^M w(i,x) \ge \frac{\alpha^2}{4} \mid f \right].
\]
We finally appeal to Hoeffding's inequality (\Cref{fact:hoeffding}) to bound this quantity. Set $\mu \coloneqq \frac{1}{N}$ and $\delta \coloneqq \frac{\alpha^2 N}{8} - 1$. Then
\begin{align*}
\Pr_{x \sim \mathcal{D}_{z,N}} \left[|| \ket{\psi} - \ket{\varphi}|| \ge \alpha \mid f\right]
&= \Pr_{x \sim \mathcal{D}_{z,N}} \left[\sum_{i=1}^M w(i,x) \ge (1 + \delta)\mu \mid f \right]\\
&\le \exp\left(-\frac{2\delta^2 \mu^2}{\sum_{i=1}^M m_i^2} \right)\\
&\le \exp\left(-\frac{2\left(\frac{\alpha^2N}{8} - 1\right)^2 \left(\frac{2}{N}\right)^2}{\frac{1}{K}} \right)\\
&\le \exp\left(-\frac{2\left(\frac{\alpha^2N}{16}\right)^2 \left(\frac{2}{N}\right)^2}{\frac{1}{K}} \right)\\
&= e^{-\frac{\alpha^4 K}{32}},
\end{align*}
where the first line applies \eqref{eq:euclidean_dist_2}; the second line applies \Cref{fact:hoeffding} (using that the $w(i,x)$'s are conditionally independent given $f$); the third line substitutes \eqref{eq:sum_mi_squares_2}; and the fourth line uses the assumption that $\alpha \ge \sqrt{\frac{16}{N}}$.
\end{proof}

Next, by repeated application of \Cref{lem:sparse_aaronson_ambainis_single_query_2}, we generalize \Cref{lem:sparse_aaronson_ambainis_single_query} to quantum algorithms that make multiple queries.

\begin{lemma}
\label{lem:sparse_aaronson_ambainis_states}
Consider a quantum algorithm $Q$ that makes $T$ queries to $x \in \{0,1\}^{MN}$ to produce a state $\ket{\psi_T}$. Then for any $K \in \Naturals$, there exists a classical algorithm that makes at most $2KT$ queries to $x$, and outputs a description of a state $\ket{\varphi_T}$ such that for any $\alpha \ge \sqrt{\frac{16}{N}}$ and any $z \in \{0,1\}^M$:
\[
\Pr_{x \sim \mathcal{D}_{z,N}} \left[|| \ket{\psi_T} - \ket{\varphi_T}|| \ge \alpha T \right] \le T \cdot e^{-\frac{\alpha^4 K}{32}}.
\]
\end{lemma}

\begin{proof}
Call $\ket{\psi_0}$ the initial state of $Q$, and suppose that $Q$ applies a sequence of unitaries $W_1, \ldots, W_t$ interspersed by queries to the phase oracle $U_x$ corresponding to $x$. For $t \le T$, denote by
\[
\ket{\psi_{t, x}} \coloneqq W_t U_x W_{t-1} U_x \cdots W_1 U_x \ket{\psi_0}
\]
the state of $Q$ after $t$ queries, with the convention that $\ket{\psi_{0,x}} = \ket{\psi_0}$ (even though $\ket{\psi_0}$ is, of course, independent of $x$).

A detailed description of the classical algorithm is given in \Cref{alg:sparse_aa}. Intuitively speaking, the simulation algorithm simply applies the algorithm from \Cref{lem:sparse_aaronson_ambainis_single_query_2} $T$ times consecutively, recording into $f$ the queries that have been made so far. Additionally, any time the algorithm encounters a row of $x$ of which more than $N/2$ queries have already been made, it queries the remainder of the row. Thus, it is clear that the query complexity of \Cref{alg:sparse_aa} is at most $2KT$.

\begin{algorithm}
\caption{Sparse oracle classical simulation}\label{alg:sparse_aa}
\KwInput{Oracle access to $x \in \{0,1\}^{MN}$, $T$-query quantum query algorithm $Q$, $K \in \Naturals$}
\KwOutput{Approximation to the output state of $Q^x$}

$f, S \leftarrow \emptyset$

$\ket{\varphi_{0,x}} \leftarrow \ket{\psi_0}$

\For{$t \leftarrow 1$ \KwTo $T$}{
    Run the algorithm from \Cref{lem:sparse_aaronson_ambainis_single_query_2} corresponding to $\ket{\psi} = W_t U_x \ket{\varphi_{t-1,x}}$ given $S$ and $f$.

    Let $\ket{\varphi_{t,x}}$ be the output, and let $S'$ be the set of locations queried by the algorithm.

    \For(\tcc*[f]{Record the queried locations}){$(i, j) \leftarrow S'$}{
        $S \leftarrow S \cup \{(i, j)\}$
        
        $f(i, j) \leftarrow x_{i,j}$
    }
    \For(\tcc*[f]{Query rows with at least $N/2$ queries already}){$i \leftarrow 1$ \KwTo $M$}{
        \If{$|\{j : (i, j) \in S\}| \ge N / 2$}{
            \For{$j \leftarrow 1$ \KwTo $N$}{
                $S \leftarrow S \cup \{(i, j)\}$
                
                $f(i, j) \leftarrow x_{i,j}$
            }
        }
    }
}

\Return $\ket{\varphi_{T,x}}$
\end{algorithm}

It remains to show that $\ket{\psi_{T,x}}$ and $\ket{\varphi_{T,x}}$ are close with high probability. For $t \le T$, define $\ket{\gamma_{t,x}}$ as the state obtained by applying the classical algorithm for the first $t$ steps and the quantum algorithm for the remaining $T - t$ steps, i.e.
\[
\ket{\gamma_{t,x}} = W_TU_xW_{T-1}U_x\cdots W_{t+1}U_x\ket{\varphi_{t,x}}.
\]
Note that $\ket{\gamma_{0,x}} = \ket{\psi_{T,x}}$ and $\ket{\gamma_{T,x}} = \ket{\varphi_{T,x}}$. From this, we may bound:
\begin{align}
||\ket{\psi_{T,x}} - \ket{\varphi_{T,x}}|| &= ||\ket{\gamma_{0,x}} - \ket{\gamma_{T,x}}||\nonumber\\
&\le \sum_{t=1}^T ||\ket{\gamma_{t-1,x}} - \ket{\gamma_{t,x}}||\nonumber\\
&= \sum_{t=1}^T ||W_tU_x\ket{\varphi_{t-1,x}} - \ket{\varphi_{t,x}}||\label{eq:sparse_aa_hybrid},
\end{align}
where the second line holds by the triangle inequality, and the last line holds because the $W_t$'s and $U_x$ are unitary transformations. \Cref{lem:sparse_aaronson_ambainis_single_query_2} implies that all of the terms in this sum are bounded with high probability. In particular, denoting by $f^t$ and $S^t$ the values of $f$ and $S$ immediately before iteration $t$ of the outer loop,
we conclude:
\begin{align*}
\Pr_{x \sim \mathcal{D}_{z,N}}\left[||\ket{\psi_{T,x}} - \ket{\varphi_{T,x}}|| \ge \alpha T\right] &\le \Pr_{x \sim \mathcal{D}_{z,N}}\left[\sum_{t=1}^T ||W_tU_x\ket{\varphi_{t-1,x}} - \ket{\varphi_{t,x}}|| \ge \alpha T\right]\\
&\le \sum_{t=1}^T \Pr_{x \sim \mathcal{D}_{z,N}}\left[||W_tU_x\ket{\varphi_{t-1,x}} - \ket{\varphi_{t,x}}|| \ge \alpha\right]\\
&= \sum_{t=1}^T \sum_{f,S}\Pr_{x \sim \mathcal{D}_{z,N}}\left[f^t = f, S^t = S\right] \cdot \\
&\qquad\Pr_{x \sim \mathcal{D}_{z,N}}\left[||W_tU_x\ket{\varphi_{t-1,x}} - \ket{\varphi_{t,x}}|| \ge \alpha \mid f^t = f, S^t = S\right]\\
&= \sum_{t=1}^T \sum_{f,S}\Pr_{x \sim \mathcal{D}_{z,N}}\left[f^t = f, S^t = S\right] \cdot \\
&\qquad\Pr_{x \sim \mathcal{D}_{z,N}}\left[||W_tU_x\ket{\varphi_{t-1,x}} - \ket{\varphi_{t,x}}|| \ge \alpha \mid x_{i,j} = f(i,j)\ \forall (i, j) \in S\right]\\
&\le T \cdot e^{-\frac{\alpha^4 k}{32}},
\end{align*}
where the first line applies \eqref{eq:sparse_aa_hybrid}, the second line holds by a union bound, the third line is true because $f^t = f$ and $S^t = S$ if and only if $x_{i,j} = f(i,j)$ for all $(i, j) \in S$ (for all $f$, $S$ such that $\Pr_{x \sim \mathcal{D}_{z,N}}\left[f^t = f, S^t = S\right] \neq 0$), and the last line holds by \Cref{lem:sparse_aaronson_ambainis_single_query_2}. This last step crucially uses the observation that $\ket{\varphi_{t-1,x}}$ is uniquely determined by $f^t$ and $S^t$, and so it is conditionally independent of $x$ given the queries recorded in $f$ and $S$ from the previous steps of the algorithm.
\end{proof}

The next theorem is essentially just a restatement of \Cref{lem:sparse_aaronson_ambainis_states} with cleaner parameters. It can be understood as a version of the Aaronson-Ambainis conjecture \cite[Conjecture 1.5]{AA14} for sparse oracles.

\begin{theorem}
\label{thm:sparse_aa_parameterized}
Consider a quantum algorithm $Q$ that makes $T$ queries to $x \in \{0,1\}^{MN}$. Then for any $\eps \ge 4T\sqrt{\frac{16}{N}}$ and $\delta > 0$, there exists a classical algorithm that makes $O\left(\frac{T^5}{\eps^4}\log \frac{T}{\delta}\right)$ queries to $x$, and that outputs an estimate $p$ such that for any $z \in \{0,1\}^M$:
\[
\Pr_{x \sim \mathcal{D}_{z,N}} \left[|\Pr[Q(x) = 1] - p| \ge \eps \right] \le \delta.
\]
\end{theorem}

\begin{proof}
Let $Q$ be the quantum algorithm corresponding to $f$, and let $\ket{\psi}$ be the output state of $Q$ on input $x$ immediately before measurement. For some $K$ to be chosen later, consider the classical algorithm corresponding to $Q$ from \Cref{lem:sparse_aaronson_ambainis_states} that makes at most $2KT$ queries and produces a classical description of a quantum state $\ket{\varphi}$ on input $x$. Let $p$ be the probability that the first bit of $\ket{\varphi}$ is measured to be $1$ in the computational basis.

\cite[Lemma 3.6]{BV97} tells us that on any input $x$:
\[
|\Pr[Q(x) = 1] - p| \le 4|| \ket{\psi} - \ket{\varphi}||.
\]

Choose $\alpha = \frac{\eps}{4T}$, which, by the assumption of the theorem, must also satisfy $\alpha \ge \sqrt{\frac{16}{N}}$. By appealing to \Cref{lem:sparse_aaronson_ambainis_states}, we conclude that
\begin{align*}
\Pr_{x \sim \mathcal{D}_{z,N}}\left[|\Pr[Q(x) = 1] - p| \ge \eps\right] &\le \Pr_{x \sim \mathcal{D}_{z,N}}\left[|| \ket{\psi} - \ket{\varphi}|| \ge \frac{\eps}{4} \right]\\
&= \Pr_{x \sim \mathcal{D}_{z,N}}\left[|| \ket{\psi} - \ket{\varphi}|| \ge \alpha T \right]\\
&\le T \cdot e^{-\frac{\alpha^4 K}{32}}.
\end{align*}
Thus, we just need to choose $K$ such that $T \cdot e^{-\frac{\alpha^4 K}{32}} \le \delta$, or equivalently:
\[
\frac{\alpha^4 K}{32} \ge \log T + \log \frac{1}{\delta}.
\]
Choosing $K = O\left(\frac{T^4}{\eps^4} \log \frac{T}{\delta}\right)$ completes the theorem, as the classical algorithm makes at most $2KT$ queries.
\end{proof}

As a straightforward corollary, we obtain the following functional version of \Cref{thm:sparse_aa_parameterized}.

\begin{corollary}
\label{cor:sparse_aa_functional}
Let $f: \{0,1\}^{MN} \to \{0,1,\bot\}$ be a function with $\Q(f) \le T$ for some $T \le \sqrt{\frac{N}{9216}}$. Then for any $\delta > 0$, there exists a function $g: \{0,1\}^{MN} \to \{0,1\}$ with $\D(g) \le O\left(T^5 \log \frac{T}{\delta}\right)$ such that for any $z \in \{0,1\}^M$:
\[
\Pr_{x \sim \mathcal{D}_{z,N}}\left[f(x) \in \{0,1\} \text{\normalfont{ and }} f(x) \neq g(x)\right] \le \delta.
\]
\end{corollary}

\begin{proof}
Let $Q$ be the quantum query algorithm corresponding to $f$. Choose $\eps = \frac{1}{6}$, and consider running the classical algorithm from \Cref{thm:sparse_aa_parameterized} that produces an estimate $p$ of $Q$'s acceptance probability. (The condition of \Cref{thm:sparse_aa_parameterized} is satisfied because $4T\sqrt{\frac{16}{N}} \le \frac{1}{6}$).

Define $g$ by:
\[
g(x) = \begin{cases}
1 & p \ge \frac{1}{2},\\
0 & p < \frac{1}{2}.
\end{cases}
\]
We want to show that $g$ usually agrees with $f$ on inputs drawn from $\mathcal{D}_{z,N}$. Because $Q$ computes $f$ with error at most $\frac{1}{3}$, we have:
\begin{align*}
\Pr_{x \sim \mathcal{D}_{z,N}}\left[f(x) \in \{0,1\} \text{ and } f(x) \neq g(x)\right] &\le \Pr_{x \sim \mathcal{D}_{z,N}}\left[|\Pr[Q(x) = 1] - p| \ge \frac{1}{6} \right]\\
&\le \delta,
\end{align*}
by \Cref{thm:sparse_aa_parameterized}. Additionally, $\D(g) \le O\left(T^5 \log \frac{T}{\delta}\right)$ because $g$ depends only on $p$.
\end{proof}

The next theorem essentially shows that no $\mathsf{PH}^\mathsf{PromiseBQP}$ oracle machine can solve the $\Forrelation \circ \OR$ problem (i.e. given an input divided into rows, decide if the $\OR$s of the rows are Forrelated or uniformly random).

\begin{theorem}
\label{thm:ph^bqp_forrelation_of_ors}
Let $M, N$ satisfy $\quasipoly(M) = \quasipoly(N)$ (i.e. $M \le \quasipoly(N)$ and $N \le \quasipoly(M)$). Let $f: \{0,1\}^{MN} \to \{0,1,\bot\}$ be computable by a depth-$2$ circuit of size $\quasipoly(N)$ in which the top gate is a function in $\mathsf{AC^0}[\quasipoly(N), O(1)]$, and all of the bottom gates are functions with bounded-error quantum query complexity at most $\polylog(N)$.

Let $b \sim \{0,1\}$ be a uniformly random bit. Suppose $z \in \{0,1\}^{M}$ is sampled such that:
\begin{itemize}
\item If $b = 0$, then $z$ is uniformly random.
\item If $b = 1$, then $z$ is drawn from the Forrelation distribution $\mathcal{F}_M$.
\end{itemize}
Then:
\[
\Pr_{b,z,x \sim \mathcal{D}_{z,N}} \left[f(x) = b \right] \le \frac{1}{2} + \frac{\polylog(N)}{\sqrt{M}}.
\]
\end{theorem}

\begin{proof}
Suppose the bottom-level gates all have quantum query complexity at most $T$, and that there are at most $s$ such gates. Let $\delta = \frac{1}{s\sqrt{M}}$. Consider the function $g: \{0,1\}^{MN} \to \{0,1\}$ obtained by replacing all of the bottom-level gates of $f$ with the corresponding decision trees from \Cref{cor:sparse_aa_functional} that have depth $d \le O\left(T^5\log \frac{T}{\delta}\right))$. (The condition of \Cref{cor:sparse_aa_functional} is satisfied for sufficiently large $N$, as $T \le \polylog(N) \ll \sqrt{N}$.)

Note that $g \in \mathsf{AC^0}[\quasipoly(N), O(1)]$: a depth-$d$ decision tree can be computed by a width-$d$ DNF formula, and since $d \le \polylog(N) \cdot \log\left(\quasipoly(N) \cdot \sqrt{M}\right) \le \polylog(N)$ and $s \le \quasipoly(N)$, the total number of gates needed to evaluate all $s$ decision trees is at most $\quasipoly(N)$.

By a union bound over all of the bottom-level gates, observe that
\begin{align}
\Pr_{b,z,x \sim \mathcal{D}_{z,N}} \left[f(x) = b \right]
&\le \Pr_{b,z,x \sim \mathcal{D}_{z,N}} \left[g(x) = b \right] + \Pr_{b,z,x \sim \mathcal{D}_{z,N}}\left[f(x) \in \{0,1\} \text{ and } f(x) \neq g(x) \right]\nonumber\\
&\le \Pr_{b,z,x \sim \mathcal{D}_{z,N}} \left[g(x) = b \right] + \frac{1}{\sqrt{M}},\label{eq:forrelation_of_ors_classical_sim}
\end{align}
from the assumption of \Cref{cor:sparse_aa_functional}, just because $g$ can disagree with $f$ only if at least one of the decision trees disagrees with its corresponding quantum query algorithm.

Consider a Boolean function $C(z, i_1,\ldots,i_M)$ that takes inputs $z \in \{0,1\}^M$ and $i_1,\ldots,i_M \in [N]$. Let $\tilde{x} \in \{0,1\}^{MN}$ be the string in which for each row $j \in [M]$:
\begin{itemize}
\item If $z_j = 0$, then the $j$th row of $\tilde{x}$ is all zeros.
\item If $z_j = 1$, then the $j$th row of $\tilde{x}$ contains a single $1$ in the $i_j$th position.
\end{itemize}
Let $C$ compute $g(\tilde{x})$. Clearly, $C \in \mathsf{AC^0}[\quasipoly(N), O(1)]$. Moreover, if $i_1,\ldots,i_M$ are chosen randomly, then $C$ simulates the behavior of $g$:
\begin{equation}
\label{eq:forrelation_of_ors_raz_tal}
\Pr_{b,z,x \sim \mathcal{D}_{z,N}} \left[g(x) = b \right] = \Pr_{b,z,i_1,\ldots,i_M}\left[C(z, i_1,\ldots,i_M) = b\right].
\end{equation}
Putting these together, we find that:
\begin{align*}
\Pr_{b,z,x \sim \mathcal{D}_{z,N}} \left[f(x) = b \right] &\le \Pr_{b,z,x \sim \mathcal{D}_{z,N}} \left[g(x) = b \right] + \frac{1}{\sqrt{M}}\\
&= \Pr_{b,z,i_1,\ldots,i_M}\left[C(z, i_1,\ldots,i_M) = b\right] + \frac{1}{\sqrt{M}}\\
&\le \frac{1}{2} + \frac{\polylog(M)}{\sqrt{M}}\\
&\le \frac{1}{2} + \frac{\polylog(N)}{\sqrt{M}},
\end{align*}
where the first two lines apply \eqref{eq:forrelation_of_ors_classical_sim} and \eqref{eq:forrelation_of_ors_raz_tal}, the third line holds by \Cref{thm:raz-tal}, and the last line uses the fact that $M \le \quasipoly(N)$.
\end{proof}

To complete this section, we require the following proposition, which is the same as \Cref{prop:fss_bqp^ph} but with the role of $\mathsf{BQP}$ and $\mathsf{\Sigma}_k^\mathsf{P}$ reversed, and with the extra subtlety that we must also consider queries to promise problems.

\begin{proposition}
\label{prop:fss_ph^promisebqp}
Let $M$ be a ${\mathsf{\Sigma}_k^\mathsf{P}}^\mathsf{PromiseBQP}$ oracle machine (i.e. a pair $\langle A, B \rangle$ of a $\mathsf{\Sigma}_k^\mathsf{P}$ oracle machine $A$ and a $\mathsf{PromiseBQP}$ oracle machine $B$). Let $p(n)$ be a polynomial upper bound on the runtime of $A$ and $B$ on inputs of length $n$. Then for any $x \in \{0,1\}^n$, there is a depth-$2$ circuit $C$ of size at most $2^{\poly(n)}$ in which the top gate is computed by a function in $\mathsf{AC^0}\left[2^{\poly(n)}, k+1\right]$, and all of the bottom gates are functions with bounded-error quantum query complexity at most $p(p(n))$, such that for any oracle $\mathcal{O}: \{0,1\}^* \to \{0,1\}$ we have:
\[
M^\mathcal{O}(x) = C\left( \mathcal{O}_{[p(p(n))]} \right),
\]
where $\mathcal{O}_{[p(p(n))]}$ denotes the concatenation of the bits of $\mathcal{O}$ on all strings of length at most $p(p(n))$.
\end{proposition}

\begin{proof}
Let $N = \sum_{m=0}^{p(n)} 2^{m}$. By \Cref{lem:furst-saxe-sipser}, there exists a function $f: \{0,1\}^N \to \{0,1\}$ in $\mathsf{AC^0}\left[2^{\poly(n)}, k+1\right]$ such that, for any language $L$, $A^{L}(x) = f\left(L_{[p(n)]}\right)$. We take this $f$ to be the top gate of our circuit, and will replace the inputs of this gate by functions of low quantum query complexity. Recall from \Cref{sec:circuit_complexity} that for $b \in \{0,1\}$, a gate labeled by $f$ evaluates to $b$ on input $P \in \{0,1,\bot\}^N$ if, for every string $Q \in \{0,1\}^N$ that extends $P$, we have $f(Q) = b$. Additionally, recall from \Cref{sec:complexity_classes} that we define queries to a promise problem $\Pi$ such that:
\[
A^\Pi(x) \coloneqq \begin{cases}
0 & A^L(x) = 0 \text{ for every language } L \text{ that extends } \Pi,\\
1 & A^L(x) = 1 \text{ for every language } L \text{ that extends } \Pi,\\
\bot & \text{otherwise}.
\end{cases}
\]
It follows that, for any promise problem $\Pi$, $A^{\Pi}(x) = f\left(\Pi_{[p(n)]}\right)$ (or, in plain words, the extension of $f$ to allow inputs in $\{0,1,\bot\}$ is consistent with the extension of $A$ to allow queries to a promise problem).

Let $\Pi$ be the promise problem decided by $B^\mathcal{O}$. Since $A^\Pi(x)$ runs in time at most $p(n)$, it can only query the evaluation of $B^\mathcal{O}$ on inputs up to length at most $p(n)$. For each $y \in \{0,1\}^{m}$ with $m \le p(n)$, there exists a partial function $g_y$ with $\Q(g_y) \le p(m) \le p(p(n))$ such that, for every oracle $\mathcal{O}$, $B^\mathcal{O}(y)$ is computed by $g_y\left(\mathcal{O}_{p(m)} \right)$.

Consider the circuit $C$ obtained by feeding these functions $\{g_y : y \in \{0,1\}^{m}, m \le p(n)\}$ into $f$. Then $M^\mathcal{O}(x) = A^{\Pi}(x) = C\left(\mathcal{O}_{[p(p(n))]}\right)$. Furthermore, $C$ clearly satisfies the desired size, depth, and structure requirements.
\end{proof}

By straightforward techniques, \Cref{thm:ph^bqp_forrelation_of_ors} can be extended to a proof of the following oracle result.

\begin{corollary}
\label{cor:bqp^np_not_in_ph^bqp}
There exists an oracle relative to which $\mathsf{BQP}^\mathsf{NP} \not\subset \mathsf{PH}^\mathsf{PromiseBQP}$.
\end{corollary}

\begin{proof}
We construct an oracle $A$ as follows. Let $L^A$ be a uniformly random unary language. For each $n \in \Naturals$, we add into $A$ a region consisting of a function $f_n: \{0,1\}^{2n} \to \{0,1\}$. Viewing the truth table of $f_n$ as a $2^{n} \times 2^{n}$ array of bits, choose $f_n$ as follows:
\begin{itemize}
\item If $L^A\left(0^n\right) = 0$, sample a uniformly random string $z$ of length $2^{n}$, and draw $f \sim \mathcal{D}_{z,2^{n}}$.
\item If $L^A\left(0^n\right) = 1$, sample $z$ from the Forrelation distribution $\mathcal{F}_{2^{n}}$, and draw $f \sim \mathcal{D}_{z,2^{n}}$.
\end{itemize}
Let $\mathcal{D}$ be the resulting distribution over oracles $A$.

We first show that $L^A \in \mathsf{BQP}^{\mathsf{NP}^A}$ with probability $1$ over $A \sim \mathcal{D}$. To do so, we define a language $P^A$ as follows: for a string $x \in \{0,1\}^*$, $P^A(x) = 1$ if $|x| = n$ and the $x$th row of $f_n$ contains a $1$; otherwise $P^A(x) = 0$. Clearly, $P^A \in \mathsf{NP}^A$: to determine if $x \in P^A$, nondeterministically guess a string $y \in \{0,1\}^{n}$ and check if $f_n(x, y) = 1$. Thus, it suffices to show that $L^A \in \mathsf{BQP}^{P^A}$, which we do below. (Note that this proof shares large parts with the proof of \Cref{claim:np_in_bqp}, only modifying a few parameters.)

\begin{claim}
\label{claim:l_in_bqp^p^a}
$L^A \in \mathsf{BQP}^{P^A}$ with probability $1$ over $A$.
\end{claim}

\begin{proof}[Proof of Claim]
A quantum algorithm can decide whether $0^{n} \in L^A$ in $\poly(n)$ time by using the Forrelation distinguishing algorithm $\mathcal{A}$ from \Cref{thm:raz-tal} on the region of $P^A$ corresponding to inputs of length $n$.

In more detail, let $z: \{0,1\}^{n} \to \{0,1\}$ denote the restriction of $P^A$ to inputs of length $n$. By \Cref{thm:raz-tal} we know that:
\[
\Pr_{A \sim \mathcal{D}} \left[\mathcal{A}(z) \neq L^A\left(0^{n}\right)\right] \le 2^{-2n},
\]
where the probability in the above expression is also taken over the randomness of $\mathcal{A}$. By Markov's inequality, we may conclude:
\[
\Pr_{A \sim \mathcal{D}} \left[\Pr\left[\mathcal{A}(z) \neq L^A\left(0^{n}\right)\right] \ge 1/3\right] \le 3 \cdot 2^{-2n}.
\]
Hence, the $\mathsf{BQP}$ promise problem defined by $\mathcal{A}$ agrees with $L^A$ on $0^{n}$, except with probability at most $3 \cdot 2^{-2n}$.

We now appeal to the Borel-Cantelli Lemma to argue that, with probability $1$ over $A$, $\mathcal{A}$ correctly decides $L^A\left(0^{n}\right)$ for all but finitely many $n \in \Naturals$. We have:
\[
\sum_{n \in \Naturals} \Pr_{A\sim D}\left[\mathcal{A}^{P^A} \text{ does not decide } L^A\left(0^{n}\right)\right] \le \sum_{n=1}^{\infty} 3 \cdot 2^{-2n} < \infty.
\]

Therefore, the probability that $\mathcal{A}$ fails on infinitely many inputs $0^n$ is $0$. Hence, $\mathcal{A}$ can be modified into a $\mathsf{BQP}$ algorithm that decides $L^A\left(0^{n}\right)$ for \textit{all} $n \in \Naturals$, with probability $1$ over $A \sim \mathcal{D}$.
\end{proof}

It remains to show that $L^A \not\in \mathsf{PH}^{\mathsf{PromiseBQP}^A}$ with probability $1$ over $A$. As we will show, this follows from \Cref{thm:ph^bqp_forrelation_of_ors} in the much same way that \Cref{cor:bqp_sigma_k} follows from \Cref{thm:bqp_sigma_k_query_version}. Fix a ${\mathsf{\Sigma}_k^\mathsf{P}}^\mathsf{PromiseBQP}$ oracle machine $M$. By the union bound, it suffices to show that
\[
\Pr_{A \sim \mathcal{D}}\left[M^A \text{ decides } L^A\right]=0.
\]

Let $n_1<n_2<\cdots$ be an infinite sequence of input lengths, spaced far enough apart (e.g. $n_{i+1}=2^{n_i}$) such that $M\left(0^{n_i}\right)$ can query the oracle on strings of length $n_{i+1}$ or greater for at most finitely many values of $i$. Next, let
\[
p(M,i)\coloneqq\Pr_{A \sim \mathcal{D}}\left[M^A \text{ correctly decides } 0^{n_i}|M^A \text{ correctly decided }0^{n_1},\dots,0^{n_{i-1}}\right]
\]
Then we have that
\[
\Pr_{A \sim \mathcal{D}}\left[M^A\text{ decides }L^A\right]\leq\prod_{i=1}^\infty p(M,i).
\]

\noindent Thus it suffices to show that, for every fixed $M$, we have $p(M,i)\leq 0.7$ for all but finitely many $i$. \Cref{prop:fss_ph^promisebqp} shows that $M$'s behavior on $0^{n_i}$ can be computed by a depth-$2$ circuit of size $2^{\poly(n_i)}$ in which the top gate is a function in $\mathsf{AC^0}\left[2^{\poly(n_i)}, k+1\right]$, and all of the bottom gates are functions with bounded-error quantum query complexity at most $\poly(n_i)$. \Cref{thm:ph^bqp_forrelation_of_ors} with $M = N = 2^{n_i}$ shows that such a circuit correctly evaluates the $\Forrelation \circ \OR$ function with probability greater than (say) $0.7$ for at most finitely many $i$. This even holds conditioned on $M^A$ correctly deciding $0^{n_1},\ldots,0^{n_{i-1}}$, because the size-$2^{2n_i}$ $\Forrelation \circ \OR$ instance is chosen independently of the smaller instances, and because $M\left(0^{n_i}\right)$ can query the oracle on strings of length $n_{i+1}$ or greater for at most finitely many values of $i$.
\end{proof}

\section{Limitations of the \texorpdfstring{$\mathsf{QMA}$}{QMA}\ Hierarchy (And Beyond)}
In this section, we use random restriction arguments to prove that $\mathsf{PP} \not\subset \mathsf{QMAH}$ relative to a random oracle.

\subsection{The Basic Random Restriction Argument}
The most basic form of our random restriction argument, though not necessarily its most easily applicable, is given below. The theorem can be understood as stating that if we choose a random subset $S$ of the bits of some input $x$, then $S$ usually contains a small set $K$ such that the quantum algorithm's acceptance probability cannot change much when any bits of $S \setminus K$ are flipped. In particular, $K$ serves as a sort of ``certificate'' of the quantum algorithm's behavior when the bits of $S \setminus K$ are unrestricted.

\begin{theorem}[Random restriction for $\mathsf{BQP}$]
\label{thm:bqp_random_restriction}
Consider a quantum algorithm $Q$ that makes $T$ queries to $x \in \{0,1\}^N$. Choose $k \in \Naturals$. If $S \subseteq [N]$ is sampled such that each $i \in [N]$ is in $S$ with probability $p$, then with probability at least $1 - 2e^{-k/6}$, there exists a set $K \subseteq S$ of size at most $k$ such that for every $y \in \{0,1\}^N$ with $\{i \in [N] : x_i \neq y_i\} \subseteq S \setminus K$, we have:
\[ \left|\Pr\left[Q(x) = 1 \right] - \Pr\left[Q(y) = 1 \right] \right| \le 16Tp\sqrt{N/k}\]
\end{theorem}

\begin{proof}
We proceed in cases. Suppose $k > 2pN$. Then, we may simply take the set $K = S$, which satisfies the theorem whenever $|S| \le k$. By a Chernoff bound (\Cref{fact:chernoff}) with $\delta = \frac{k}{pN} - 1$, the probability that this condition is violated is upper bounded by:
\[
\Pr\left[|S| \ge (1 + \delta)pN\right] \le e^{-\frac{\delta^2 pN}{2 + \delta}} \le e^{-\frac{(1 + \delta) pN}{6}} = e^{-k/6},
\]
where we use the inequality $\frac{\delta^2}{2+\delta} \ge \frac{1 + \delta}{6}$ which holds for all $\delta \ge 1$.

In the complementary case, suppose $k \le 2pN$. Recall the definition of the \textit{query magnitudes} $q_i$ from \Cref{lem:bbbv_query_magnitude}, which are defined in terms of the behavior of $Q$ on $x$. Note that $\sum_{i=1}^N q_i = T$. Let $\tau = \frac{2Tp}{k}$. Since all $q_i$s are nonnegative, $|\{i \in [N] : q_i > \tau\}| \le \frac{T}{\tau}$. Choose $K = \{i \in S : q_i > \tau\}$. By a Chernoff bound (\Cref{fact:chernoff}),
\[ \Pr\left[ |K| \ge k \right] \le e^{-k/6}, \]
using the fact that $\E\left[ \left|K\right| \right] = p \cdot |\{i \in [N] : q_i > \tau\}| \le \frac{p T}{\tau} = \frac{k}{2}$. Additionally,
\[ \Pr\left[ |S| \ge 2p N \right] \le e^{-p N/3} \le e^{-k/6}, \]
by another Chernoff bound. Suppose $|K| \le k$ and $|S| \le 2p N$, which happens with probability at least $1 - 2e^{-k/6}$. Then we have:

\begin{align*}
\left|\Pr\left[Q(x) = 1 \right] - \Pr\left[Q(y) = 1 \right] \right| & \le 8\sqrt{T} \cdot \sqrt{\sum_{i : x_i \neq y_i} q_i}\\
&\le 8\sqrt{T} \cdot \sqrt{\sum_{i \in S \setminus K} q_i}\\
&\le 8\sqrt{T} \cdot \sqrt{|S| \tau}\\
&\le 16Tp\sqrt{N/k}.
\end{align*}
Above, the first line applies \Cref{lem:bbbv_query_magnitude}; the second line holds by the assumption that $\{i \in [N] : x_i \neq y_i\} \subseteq S \setminus K$; the third line applies the definition of $K$ to conclude that $q_i \le \tau$ for all $i \in S \setminus K$; and the last line substitutes $|S| \le 2pN$ and $\tau = \frac{2Tp}{k}$.
\end{proof}

%We'll typically choose $\tau = p = \frac{1}{2048\sqrt{TN}}$ I think, which gives a final bound of $1/4$.

\subsection{Measuring Closeness of Functions}
\label{sec:closeness}
In order to better make sense of \Cref{thm:bqp_random_restriction}, we introduce some language that allows us to quantify how ``close'' a pair of partial functions are.

\begin{definition}
\label{def:disagreement}
Let $f, g: \{0,1\}^N \to \{0,1,\bot\}$ be partial functions. We say that \emph{$g$ disagrees with $f$ on $x$} if $x \in \Dom(f)$ and $g(x) \neq f(x)$. The \emph{disagreement of $g$ with respect to $f$}, denoted $\disagr_f(g)$, is the fraction of inputs on which $f$ and $g$ disagree:
\[\disagr_f(g) \coloneqq \Pr_{x \sim \{0,1\}^N} \left[ g \text{ \rm disagrees with } f \text{ \rm on } x \right].\]
If $\mathcal{C}$ is a class of partial functions, the \emph{disagreement of $\mathcal{C}$ with respect to $f$}, denoted $\disagr_f(\mathcal{C})$, is the minimum disagreement of any function in $\mathcal{C}$ with $f$:
\[
\disagr_f(\mathcal{C}) \coloneqq \min_{g \in \mathcal{C}} \disagr_f(g).
\]
\end{definition}

Note that the above definition is not symmetric in $f$ and $g$. We typically think of $f$ as some ``target'' function, and $g$ as some machine that tries to compute $f$ on the inputs where $f$ is defined. The goal is for $g$ to be consistent with $f$ with good probability; thus we only penalize $g$ if it reports an incorrect answer when $f$ takes a value in $\{0,1\}$.

This next few propositions show that disagreement behaves intuitively in various ways. First, we show that disagreement satisfies a sort of ``triangle inequality''.

\begin{proposition}
\label{prop:disagr_triangle_inequality}
Let $f, g, h: \{0,1\}^N \to \{0,1,\bot\}$. Then $\disagr_f(h) \le \disagr_f(g) + \disagr_g(h)$. 
\end{proposition}
\begin{proof}
This follows from \Cref{def:disagreement} and a union bound:
\begin{align*}
\disagr_f(h) &= \Pr_{x \sim \{0,1\}^N} \left[ h \text{ \rm disagrees with } f \text{ \rm on } x \right]\\
&\le \Pr_{x \sim \{0,1\}^N} \left[ g \text{ \rm disagrees with } f \text{ \rm on } x \text{ \rm OR } h \text{ \rm disagrees with } g \text{ \rm on } x \right]\\
&\le \disagr_f(g) + \disagr_g(h).\qedhere
\end{align*}
\end{proof}

The next two propositions show that disagreement behaves intuitively with respect to random restrictions. First, we show that disagreement is preserved, in expectation, under random restrictions.
\begin{proposition}
\label{prop:disagr_random_equal}
Let $f, g: \{0,1\}^N \to \{0,1,\bot\}$. Consider a random restriction $\rho$ with $\Pr[*] = p$. Then $\E_\rho\left[\disagr_{f_\rho}\left(g_\rho\right)\right] = \disagr_f(g)$.
\end{proposition}
\begin{proof}
Let $S = \{i \in [N]: \rho(i) = *\}$, and let $y \in \{0,1\}^{[N] \setminus S}$ be the assignment of non-$*$ variables under $\rho$. Then:
\begin{align*}
\E_\rho\left[\disagr_{f_\rho}\left(g_\rho\right)\right] &= \E_\rho\left[\Pr_{z \in \{0,1\}^S} \left[  g_\rho \text{ \rm disagrees with } f_\rho \text{ \rm on } z \right] \right]\\
&= \E_{y \in \{0,1\}^{[N] \setminus S}}\left[\Pr_{z \in \{0,1\}^S} \left[  g \text{ \rm disagrees with } f \text{ \rm on } (y, z) \right] \right]\\
&= \Pr_{x \in \{0,1\}^{N}}\left[g \text{ \rm disagrees with } f \text{ \rm on } x \right]\\
&= \disagr_f(g).\qedhere
\end{align*}
\end{proof}

Finally, we show that if we perform a sequence of random restrictions, each of which incurs some cost in disagreement, then the disagreement accumulates additively.
\begin{proposition}
\label{prop:disagr_composing_restrictions}
Let $f: \{0,1\}^N \to \{0,1,\bot\}$, and let $\rho, \sigma$ be random restrictions with $\Pr[*] = p,q$ respectively. Suppose there exist classes of functions $\mathcal{C},\mathcal{D}$ such that:
\begin{enumerate}[(a)]
\item $\E_{\rho} \left[\disagr_{f_{\rho}}\left(\mathcal{C}\right)\right] \le \eps$.
\item For all $g \in \mathcal{C}$, $\E_\sigma \left[\disagr_{g_{\sigma}}\left(\mathcal{D}\right)\right] \le \delta$.
\end{enumerate}
Then $\E_{\rho \sigma} \left[\disagr_{f_{\rho \sigma}}\left(\mathcal{D}\right)\right] \le \eps + \delta$.
\end{proposition}
\begin{proof}
Let $g \in \mathcal{C}$ be the function (depending on $\rho$) that minimizes $\disagr_{f_\rho}(g)$, and let $h \in \mathcal{D}$ be the function (depending on $\rho$ and $\sigma$) that minimizes $\disagr_{g_\sigma}(h)$. Then we have:
\begin{align*}
\E_{\rho,\sigma}\left[\disagr_{f_{\rho \sigma}}(h)\right] &\le \E_{\rho,\sigma}\left[\disagr_{f_{\rho \sigma}}\left(g_\sigma\right) + \disagr_{g_{\sigma}}(h)\right]\\
&= \E_{\rho}\left[\disagr_{f_\rho}(g)\right] + \E_{\rho,\sigma}\left[\disagr_{g_\sigma}(h)\right]\\
&\le \eps + \delta,
\end{align*}
where the first line holds by \Cref{prop:disagr_triangle_inequality}; the second line applies \Cref{prop:disagr_random_equal} and linearity of expectation; and the last line holds because of assumptions (a) and (b).
\end{proof}

\subsection{Random Restriction for \texorpdfstring{$\mathsf{QMA}$}{QMA} Queries}
With the tools introduced in the previous section, we can state a more intuitive and useful form of our random restriction argument for $\mathsf{QMA}$ query algorithms. It states that a random restriction of a $\mathsf{QMA}$ query algorithm is close in expectation to a small-width DNF formula.

\begin{theorem}
\label{thm:qma_random_restriction}
Consider a partial function $f: \{0,1\}^N \to \{0,1,\bot\}$ with $\QMA(f) \le T$. Set $p = \frac{\sqrt{k}}{64T\sqrt{N}}$ for some $k \in \Naturals$. Let $\rho$ be a random restriction with $\Pr[*] = p$. Let $\mathcal{DNF}_k$ denote the set of width-$k$ DNFs. Then $\E_\rho \left[\disagr_{f_\rho}\left(\mathcal{DNF}_k\right)\right] \le 2e^{-k/6}$.
\end{theorem}

\begin{proof}
It will be convenient to view the choice of $\rho$ as follows: we choose a string $z \in \{0,1\}^N$ uniformly at random, and then we choose a set $S \subseteq [N]$ wherein each $i \in [N]$ is included in $S$ independently with probability $p$. Then, we take $\rho$ to be:
\[\rho(i) = \begin{cases} * & i \in S\\ z_i & i \not\in S. \end{cases}\]
Thus, by definition, it holds that $f_\rho(z|_S) = f(z)$.

Choose $g \in \mathcal{DNF}_k$ as follows, depending on the choice of $\rho$. For each $x \in \Dom(f_\rho)$, choose a certificate $K_x$ for $f_\rho$ on $x$ of minimal size. Define $g$ by:
\[g(y) \coloneqq \bigvee_{\substack{x \in f_\rho^{-1}(1) \\ \C^x(f) \le k}} \bigwedge_{i \in K_x} y_i = x_i.\]
This is to say that we take $g$ to be the $\OR$ of all of the chosen $1$-certificates of $f_\rho$ that have size at most $k$. Clearly, $g$ is computable by a width-$k$ DNF, so it remains to show that $g$ has small disagreement with respect to $f$ in expectation.

Call the pair $(z, S)$ ``good'' if either $f(z) \neq 1$ or $\C^{z|_S}(f_\rho) \le k$. Observe that $g$ agrees with $f_\rho$ on the input $z|_S$ whenever $(z, S)$ is good:
\begin{itemize}
\item If $f(z) = 0$, then $z|_S$ cannot contain a $1$-certificate for $f_\rho$, and so $g(z|_S) = f_\rho(z|_S) = 0$.
\item If $f(z) = \bot$ then $g$ always agrees with $f_\rho$ on $z|_S$.
\item Lastly, if $f(z) = 1$ and $\C^{z|_S}(f_\rho) \le k$, then $z|_S$ certainly contains the certificate $K_{z|_S}$ and hence $g(z|_S) = f_\rho(z|_S) = 1$.
\end{itemize}

Thus we can see that the expected disagreement of $g$ with respect to $f$ satisfies:
\[\E_\rho \left[\disagr_{f_\rho}(g) \right] = \Pr_{z,S}\left[g \text{ disagrees with } f_\rho \text{ on } z|_S\right] \le \Pr_{z,S}[(z, S) \text{ is not good}].\]

It remains to prove that most $(z, S)$ are good. Let $V(\ket{\psi},z)$ be the $\mathsf{QMA}$ verifier corresponding to $f$ on input $z$, where $\ket{\psi}$ is the witness. Fix $z$, and let $\ket{\psi_z}$ be the witness that maximizes $\Pr[V(\ket{\psi_z},z) = 1]$. By \Cref{thm:bqp_random_restriction}, for any $z \in \{0,1\}^N$, with probability at least $1 - 2e^{-k/6}$ over $S$, there exists a set $K \subseteq S$ of size at most $k$ such that for every $y \in \{0,1\}^N$ with $\{i \in [N] : z_i \neq y_i\} \in S \setminus K$ we have:
\[|\Pr[V(\ket{\psi_z},z) = 1] - \Pr[V(\ket{\psi_z},y) = 1]| \le \frac{1}{4}.\]
In particular, if $f(z) = 1$ then $\Pr[V(\ket{\psi_z},y) = 1] \ge \frac{2}{3} - \frac{1}{4} > \frac{1}{3}$, and so $f(y) \neq 0$. This is to say that $K$ is a certificate for $f_\rho$ on $z|_S$, and therefore $\C^{z|_S}(f_\rho) \le k$ and $(z, S)$ is good.
\end{proof}

\subsection{Application to \texorpdfstring{$\mathsf{QMAH}$}{QMAH}}

In order to generalize \Cref{thm:qma_random_restriction} to $\mathsf{QMAH}$ machines, we require the following form of H\r{a}stad's switching lemma for DNF formulas \cite{Has87}. The statement given below, and arguably its simplest proof, are given in an exposition by Thapen \cite{Tha09}. Technically, this is just a weaker statement of \Cref{thm:rossman_ac0_restriction}, though we prefer the version given here because it makes the constant factor explicit.

\begin{lemma}[Switching Lemma]
\label{lem:dnf_random_restriction}
Let $f$ be a width-$k$ DNF. If $\rho$ is a random restriction with $\Pr[*] = q < \frac{1}{9}$, then for any $t > 0$, $\Pr\left[\D(f_\rho) > t \right] \le (9qk)^t$.
\end{lemma}

\begin{corollary}
\label{cor:dnf_random_restriction}
Let $f$ be a width-$k$ DNF. Denote by $\mathcal{D}_t$ the set of functions that have deterministic query complexity at most $t$. If $\rho$ is a random restriction with $\Pr[*] = q < \frac{1}{9}$, then for any $t > 0$, we have $\E_\rho[\disagr_{f_\rho}\left(\mathcal{D}_t\right)] \le (9qk)^t$.
\end{corollary}

\begin{proof}
Take $g = f_\rho$ if $\D(f_\rho) \le t$, and otherwise let $g$ be the all-zeros function. Then $g \in \mathcal{D}_t$, and by \Cref{lem:dnf_random_restriction}, $\E_\rho[\disagr_{f_\rho}\left(g\right)] \le (9qk)^t$.
\end{proof}

With all of these tools in hand, we prove in our next theorem that under an appropriately chosen random restriction, a circuit composed of $\mathsf{QMA}$ query gates simplifies to a function that is close (in expectation) to a function with low deterministic query complexity. The proof amounts to a recursive application of \Cref{thm:qma_random_restriction} combined with \Cref{cor:dnf_random_restriction}.

\begin{theorem}
\label{thm:qmah_random_restriction}
Let $f: \{0,1\}^N \to \{0,1,\bot\}$ be computable by a size-$s$ depth-$d$ circuit where each gate is a (possibly partial) function with $\mathsf{QMA}$ query complexity at most $R$. Fix $k \in \Naturals$, and consider a random restriction $\rho$ with
\[
\Pr[*] = \left(1024R^2 k N\right)^{2^{-d} - 1} \cdot \left(\frac{e^{-1/6}}{18k}\right)^d,
\]
Denote by $\mathcal{D}_k$ the set of functions that have deterministic query complexity at most $k$. Then $\E_\rho\left[\disagr_{f_\rho}\left(\mathcal{D}_k\right)\right] \le 4se^{-k/6}$.
\end{theorem}

\begin{proof}
For convenience, define $\alpha = \frac{1}{64R\sqrt{k}}$. Let $N_0 = N$, and for $i \in [d]$ define $p_i$ and $N_i$ recursively by:
\begin{align*}
p_i &= \frac{\alpha}{\sqrt{N_{i-1}}}\\
N_i &= 2p_iN_{i-1}.
\end{align*}
This recursive definition implies that:
\begin{equation}
\label{eq:product_p_is}
\prod_{i=1}^d p_i = \frac{N_d}{2^dN}.
\end{equation}
Additionally, a simple inductive calculation shows that $N_i$ takes the closed form:
\begin{equation}
\label{eq:N_i_closed_form}
N_i = (2\alpha)^{2(1 - 2^{-i})} N^{2^{-i}}.
\end{equation}

We view $\rho$ as a sequence of restrictions $\rho_1,\ldots,\rho_d$ in which $\rho_i$ has $\Pr[*] = p_iq$, where $q = \frac{e^{-1/6}}{9k}$ (one can easily verify from \eqref{eq:product_p_is} and \eqref{eq:N_i_closed_form} that $\prod_{i=1}^d p_iq$ equals the probability given in the statement in the theorem). We view each $\rho_i$ itself as the composition of two random restrictions, one with $\Pr[*] = p_i$ and one with $\Pr[*] = q$.

We proceed in cases. Suppose $k > N_d$. Let $S = \{i \in [N]: \rho(i) = *\}$ denote the set of unrestricted variables. Notice that $\disagr_{f_\rho}\left(\mathcal{D}_k\right) = 0$ whenever $|S| \le k$, just because $\mathcal{D}_k$ contains all total functions on at most $k$ bits. Let $\delta = \frac{k}{N\cdot\Pr[*]} - 1$. By \eqref{eq:product_p_is},
\[
k > N_d = \frac{2^d N\cdot\Pr[*]}{q^d} \ge 2^d N\cdot\Pr[*] \ge 2N\cdot\Pr[*],
\]
which implies that $\delta \ge 1$. By a Chernoff bound (\Cref{fact:chernoff}), this implies that:
\[
\E_\rho\left[\disagr_{f_\rho}\left(\mathcal{D}_k\right)\right] \le \Pr\left[|S| \ge (1 + \delta)N\cdot\Pr[*]\right] \le e^{-\frac{\delta^2 N\cdot\Pr[*]}{2 + \delta}} \le e^{-\frac{(1 + \delta) N\cdot\Pr[*]}{6}} = e^{-k/6},
\]
where we use the inequality $\frac{\delta^2}{2+\delta} \ge \frac{1 + \delta}{6}$ which holds for all $\delta \ge 1$. Thus, the theorem is proved in this $k > N_d$ case.

In the complementary case, suppose $k \le N_d$. Let $\mathcal{C}_i$ be the class of functions such that for all $g \in \mathcal{C}_i$:
\begin{enumerate}[(a)]
\item $g$ is computable by a circuit with the same structure as $f$, except that the gates at distance\footnote{Here, distance is defined as the length of the \textit{longest} path from that gate to any of the inputs.} at most $i$ from the input are eliminated and the gates at distance $i + 1$ make at most $Rk$ queries (or, if $i = d$, $\D(g) \le k$).
\item $g$ depends on at most $N_i$ inputs.
\end{enumerate}
By convention, let $\mathcal{C}_0 = \{f\}$. With this definition, the statement of the theorem follows from \Cref{prop:disagr_composing_restrictions} and an inductive application of the following claim:
\begin{claim}
For all $g \in \mathcal{C}_{i-1}$, $\E_{\rho_i}\left[\disagr_{g_{\rho_i}}\left(\mathcal{C}_i\right)\right] \le 4e^{-k/6} \cdot s_i$, where $s_i$ is the number of gates at distance exactly $i$ from the inputs in the circuit that computes $f$.
\end{claim}
\begin{proof}[Proof of Claim]
Consider applying $\rho_i$ to $g$. After the random restriction with $\Pr[*] = p_i$, by a Chernoff bound (\Cref{fact:chernoff}) and because $g$ depends on at most $N_{i-1}$ variables, the resulting function depends on at most $N_i = 2p_iN_{i-1}$ variables, except with probability at most $e^{-N_i / 6} \le e^{-N_d / 6} \le e^{-k/6}$ over this first restriction.\footnote{Actually, a careful inspection of the steps leading up to this proof reveals that this Chernoff bound is unnecessary: we already account for this ``bad'' event (the number of unrestricted variables being larger than $2p_i N_{i-1}$) in \Cref{thm:bqp_random_restriction}. We only write it this way to make each step of the proof is as self-contained as possible.} Additionally, by \Cref{thm:qma_random_restriction} with $T = Rk$ and $p = p_i$, because $g$ is a function of at most $N_{i-1}$ variables, there exist width-$k$ DNFs that each have expected disagreement at most $2e^{-k/6}$ with respect to the corresponding bottom-layer $\mathsf{QMA}$ gates of the circuit that computes $g_{\rho_i}$.

After the next random restriction with $\Pr[*] = q$, by \Cref{cor:dnf_random_restriction}, these width-$k$ DNFs each have expected disagreement at most $e^{-k/6}$ from functions of deterministic query complexity at most $k$. Hence, by \Cref{prop:disagr_composing_restrictions}, viewing $\rho_i$ as the composition of these two random restrictions, each bottom-layer $\mathsf{QMA}$ gate in the circuit that computes $g_{\rho_i}$ has expected disagreement at most $3e^{-k/6}$ from a function of deterministic query complexity at most $k$.

Let $h$ be a function depending on $\rho_i$, chosen as follows. Take $h$ to be the all zeros function if $g_{\rho_i}$ depends on more than $N_i$ variables; otherwise let $h$ be the function obtained from $g$ by replacing the bottom-level gates of the circuit that computes $g$ with the corresponding functions of deterministic query complexity $k$. We verify that $h \in \mathcal{C}_i$:
\begin{enumerate}[(a)]
\item We can absorb the functions of query complexity $k$ into the next layer of $\mathsf{QMA}$ gates, increasing the query complexity of each gate by a multiplicative factor of $k$. Alternatively, if $i = d$, then $\D(h) \le k$ just because $g$ consists of a single gate.
\item Either $g_{\rho_i}$ depends on at most $N_i$ variables, or else $h$ is trivial; in either case $h$ depends on at most $N_i$ inputs.
\end{enumerate}
We now demonstrate that $\E_{\rho_i}\left[\disagr_{g_{\rho_i}}(h)\right] \le 4e^{-k/6} \cdot s_i$. Notice that $h$ never disagrees with $g_{\rho_i}$ on input $x$, unless either (1) $g_{\rho_i}$ depends on more than $N_i$ variables, or (2) one of the functions of deterministic query complexity $k$ disagrees with its corresponding $\mathsf{QMA}$ gate on $x$. Hence, by a union bound, $\E_{\rho_i}\left[\disagr_{g_{\rho_i}}(h)\right] \le e^{-k/6} + 3e^{-k/6} \cdot s_i \le 4e^{-k/6} \cdot s_i$.
\end{proof}
This completes the theorem in the $k \le N_d$ case.
\end{proof}

As a corollary, we obtain the following result, which shows that small circuits composed of functions with low $\mathsf{QMA}$ query complexity cannot compute the $\PARITY$ function.

\begin{corollary}
\label{cor:qmah_parity_disagreement}
Let $f: \{0,1\}^N \to \{0,1,\bot\}$ be computed by a circuit of size $s = \quasipoly(N)$ and depth $d = O(1)$, where each gate has $\mathsf{QMA}$ query complexity at most $R \le \polylog(N)$. Let $\PARITY_N$ be the parity function on $N$ bits. Then for any $\eps \ge \frac{1}{\quasipoly(N)}$ and sufficiently large $N$, $\disagr_{\PARITY_N}(f) \ge \frac{1}{2} - \eps$.
\end{corollary}
\begin{proof}
Choose $k = \left\lceil 6\ln\left(5s/\eps\right)\right\rceil \le \polylog(N)$. Let $\rho$ be a random restriction where $p = \Pr[*]$ is the probability given in the statement of \Cref{thm:qmah_random_restriction}. A simple calculation shows that $pN \ge \frac{N^{\Omega(1)}}{\polylog(N)}$; hence $k \le \frac{pN}{2}$ for sufficiently large $N$. Let $g \in \mathcal{D}_k$ be the function (depending on $\rho$) that minimizes $\disagr_{f_\rho}(g)$. With this, we have the following chain of inequalities:
\begin{align*}
\disagr_{\PARITY_N}(f)
&= \E_\rho\left[\disagr_{\PARITY_{N|\rho}}\left(f_\rho\right)\right]\\
&\ge \E_\rho\left[\disagr_{\PARITY_{N|\rho}}\left(g\right) - \disagr_{f_\rho}\left(g\right)\right]\\
&\ge \E_\rho\left[\disagr_{\PARITY_{N|\rho}}\left(g\right)\right] - 4se^{-k/6}\\
&\ge \frac{1}{2}\Pr_\rho\left[|\{i \in [N]: \rho(i) = *\}| > k\right] - 4se^{-k/6}\\
&\ge \frac{1}{2}\left(1 - e^{-k/4}\right) - 4se^{-k/6}\\
&\ge \frac{1}{2} - 5se^{-k/6}\\
&\ge \frac{1}{2} - \eps.
\end{align*}
Above, the first line holds by \Cref{prop:disagr_random_equal}; the second line holds by \Cref{prop:disagr_triangle_inequality}; the third line applies linearity of expectation along with the bound from \Cref{thm:qmah_random_restriction}; the fourth line uses the fact that any function of deterministic query complexity $k$ disagrees with the $(k+1)$-bit parity function on exactly half of all inputs; the fifth line uses a Chernoff bound (\Cref{fact:chernoff}) and $k \le \frac{pN}{2}$; the sixth line substitutes $\frac{e^{-k/4}}{2} \le e^{-k/4} \le e^{-k/6} \le se^{-k/6}$; and the last line substitutes the definition of $k$.
\end{proof}

To complete the oracle result of this section, we require the following analogue of Furst-Saxe-Sipser \cite{FSS84} (\Cref{lem:furst-saxe-sipser}) for $\mathsf{QMAH}$.

\begin{proposition}
\label{prop:qmah_fss}
For some constant $d$, let $M$ be a $\mathsf{PromiseQMAH}_d$ oracle machine (i.e. a tuple of $\mathsf{PromiseQMA}$ oracle machines $\langle M_1, \ldots, M_d \rangle$), and let $p(n)$ be a polynomial upper bound on the runtime of each $M_i$ on inputs of length $n$. Define $p^d(n) \coloneqq \underbrace{p(p( \cdots p}_{d \text{\rm \ times}}(n)))$. Then for any $x \in \{0,1\}^n$, there is a circuit $C$ of size at most $2^{\poly(n)}$ and depth $d$ in which each gate has $\mathsf{QMA}$ query complexity at most $p^d(n)$, such that for any oracle $\mathcal{O}: \{0,1\}^* \to \{0,1\}$, we have:
\[
M^\mathcal{O}(x) = C\left(\mathcal{O}_{[p^d(n)]}\right),
\]
where $\mathcal{O}_{[p^d(n)]}$ denotes the concatenation of the bits of $\mathcal{O}$ on all strings of length at most $p^d(n)$.
\end{proposition}
\begin{proof}
We prove by induction on $d$. In the base case $d = 1$, we simply have a $\mathsf{PromiseQMA}$ machine. Thus, $M_1^\mathcal{O}(x)$ is a partial function of $\mathsf{QMA}$ query complexity at most $p(n)$ in the bits of $\mathcal{O}$, and since $M_1$ runs in time at most $p(n)$, it can only query bits of $\mathcal{O}$ up to length at most $p(n)$. We may view this $\mathsf{QMA}$ query function as a circuit of the desired form consisting of only a single gate.

For the inductive step, let $d > 1$. We can view $M$ as a $\mathsf{PromiseQMA}^\mathsf{PromiseQMAH_{d-1}}$ machine, where $M_d$ is the base $\mathsf{PromiseQMA}$ machine and $M' \coloneqq \langle M_1,\ldots,M_{d-1} \rangle$ is the $\mathsf{PromiseQMAH}_{d-1}$ machine. $M_d^{M'^{\mathcal{O}}}(x)$ is a partial function of $\mathsf{QMA}$ query complexity at most $p(n)$ in the bits of $M'^{\mathcal{O}}$.\footnote{Here, we slightly abuse notation to let $M'^{\mathcal{O}}$ denote the promise problem decided by $M'$ with oracle $\mathcal{O}$. Also observe that, as in the proof of \Cref{prop:fss_ph^promisebqp}, the notion of promise problem queries defined in \Cref{sec:complexity_classes} is consistent with the way we extend the domain of circuit gates to $\{0,1,\bot\}$ in \Cref{sec:circuit_complexity}.} We take this partial function to be the top gate of our circuit, and use the inductive hypothesis to replace the inputs to this gate, the bits of $M'^{\mathcal{O}}$, with depth-$(d-1)$ circuits.

Since $M_1$ runs in time at most $p(n)$, it can only query bits of $M'^{\mathcal{O}}$ up to length at most $p(n)$. By the inductive hypothesis, for each $y \in \{0,1\}^{m}$ with $m \le p(n)$, $M'^{\mathcal{O}}(y)$ is computed by a circuit of size $2^{\poly(m)} \le 2^{\poly(n)}$ and depth $d$, with $\mathsf{QMA}$ query complexity at most $p^{d-1}(m) \le p^d(n)$ at each gate, where the inputs to this circuit are the bits of $\mathcal{O}$ on inputs of length at most $p^{d-1}(m) \le p^d(n)$. So, the resulting circuit obtained by composing the top gate with these circuits clearly has depth $d$, query complexity at most $p^d(n)$ at each gate, and depends only on $\mathcal{O}_{[p^d(n)]}$. The total size of this circuit is upper bounded by:
\[
1 + \sum_{m=0}^{p(n)} 2^{m} \cdot 2^{\poly(m)} \le 2^{\poly(n)},
\]
which proves the proposition.
\end{proof}

Via standard complexity-theoretic techniques, this implies the following:
\begin{corollary}
\label{cor:pp_not_in_qmah}
$\mathsf{PP}^\mathcal{O}\not \subset \mathsf{QMAH}^\mathcal{O}$ with probability $1$ over a random oracle $\mathcal{O}$.
\end{corollary}

\begin{proof}[Proof]
Note that $\mathsf{PP}^\mathcal{O} \subseteq \mathsf{QMAH}^\mathcal{O}$ if and only if $\mathsf{P}^{\mathsf{\# P}^\mathcal{O}} \subseteq \mathsf{QMAH}^\mathcal{O}$, just because $\mathsf{QMAH}^\mathcal{O}$ is closed under polynomial-time reductions. Hence, it suffices to show that $\mathsf{P}^{\mathsf{\# P}^\mathcal{O}} \not\subset \mathsf{QMAH}^\mathcal{O}$.

Let $L^\mathcal{O}$ be the language consisting of strings $0^n$ such that, if we treat $n$ as an index into a portion of $\mathcal{O}$ of size $2^n$, then the parity of that length-$2^n$ string is $1$. Then $L^\mathcal{O} \in \mathsf{P}^{\mathsf{\# P}^\mathcal{O}}$ (indeed, $L^\mathcal{O} \in \mathsf{\oplus P}^\mathcal{O}$).

It remains to show that $L^\mathcal{O} \not\in \mathsf{QMAH}^\mathcal{O}$. (The remainder of this proof is largely the same as our other oracle separations that follow from circuit lower bounds.) Fix a $\mathsf{QMAH}$ oracle machine $M$. By the union bound, it suffices to show that
\[
\Pr_{\mathcal{O}}\left[M^\mathcal{O} \text{ decides } L^\mathcal{O}\right]=0.
\]

Let $n_1<n_2<\cdots$ be an infinite sequence of input lengths, spaced far enough apart (e.g. $n_{i+1}=2^{n_i}$) such that $M\left(0^{n_i}\right)$ can query the oracle on strings of length $n_{i+1}$ or greater for at most finitely many values of $i$. Next, let
\[
p(M,i)\coloneqq\Pr_{\mathcal{O}}\left[M^\mathcal{O} \text{ correctly decides } 0^{n_i}|M^\mathcal{O} \text{ correctly decided }0^{n_1},\dots,0^{n_{i-1}}\right]
\]
Then we have that
\[
\Pr_{\mathcal{O}}\left[M^\mathcal{O}\text{ decides }L^\mathcal{O}\right]\leq\prod_{i=1}^\infty p(M,i).
\]

\noindent Thus it suffices to show that, for every fixed $M$, we have $p(M,i)\leq 0.7$ for all but finitely many $i$. \Cref{prop:qmah_fss} shows that $M$'s behavior on $0^{n_i}$ can be computed by a circuit of size at most $2^{\poly(n_i)}$ and depth $O(1)$ in which each gate has $\mathsf{QMA}$ query complexity at most $\poly(n_i)$. \Cref{cor:qmah_parity_disagreement} with $N = 2^{n_i}$, $R = \poly(n_i)$, and $\eps = 0.2$ shows that such a circuit correctly evaluates the $\PARITY_N$ function with probability greater than $0.7$ for at most finitely many $i$. This even holds conditioned on $M^\mathcal{O}$ correctly deciding $0^{n_1},\ldots,0^{n_{i-1}}$, because the size-$2^{n_i}$ $\PARITY$ instance is chosen independently of the smaller instances, and because $M\left(0^{n_i}\right)$ can query the oracle on strings of length $n_{i+1}$ or greater for at most finitely many values of $i$.
\end{proof}

\subsection{Beyond \texorpdfstring{$\mathsf{QMAH}$}{QMAH}}
Our proof that $\mathsf{PP} \not\subset \mathsf{QMAH}$ relative to a random oracle also extends to complexity classes that are potentially much stronger than $\mathsf{QMAH}$. This is because our definition of $\mathsf{QMA}$ query complexity (\Cref{def:qma_query_complexity}) only depends on the number of queries made by the verifier, and not on the length of the witness state. Hence, $\mathsf{QMA}$ query complexity actually upper bounds the relativized power of almost any complexity class that involves interactive proofs with a polynomial-time quantum verifier, including $\mathsf{QMA}(2)$ \cite{KMY03}, $\mathsf{QSZK}$ \cite{Wat02}, and $\mathsf{QMIP}$ \cite{KM02}. To illustrate, we argue briefly that $\mathsf{PP} \not\subset \mathsf{QMIP}^{\mathsf{PromiseQMIP}^{\mathsf{PromiseQMIP}^{\cdots}}}$ relative to a random oracle.

Recall that $\mathsf{PromiseQMIP}$ is the set of promise problems $\Pi$ for which there exists an efficient \textit{quantum multiprover interactive proof system}: a communication protocol in which one or more provers communicate with a verifier, trying to convince the verifier that $\Pi(x) = 1$. The verifier is a polynomial time machine that can send and receive quantum messages. The provers are computationally unbounded, and may share an entangled state at the start of the protocol. Otherwise, the provers are not allowed to communicate with each other during the protocol. Then, $\Pi(x) = 1$ if there exists a prover strategy that causes the verifier to accept with probability at least $\frac{2}{3}$, while $\Pi(x) = 0$ if, for every prover strategy, the verifier accepts with probability at most $\frac{1}{3}$.

The key observation is that a $\poly(n)$-time $\mathsf{QMIP}$ oracle protocol can be simulated by a $\poly(n)$-query $\mathsf{QMA}$ protocol in which the verifier receives an arbitrarily long witness, and the verifier is computationally unbounded. In this $\mathsf{QMA}$ protocol, the witness is interpreted as a string that is purported to encode the answers to all oracle queries on at most $\poly(n)$ bits. The verifier then simulates the $\mathsf{QMIP}$ protocol, choosing the prover strategy that causes the $\mathsf{QMIP}$ verifier to accept with the greatest possible probability when the oracle is consistent with the given witness. Finding this optimal strategy is merely a computational problem, and so the $\mathsf{QMA}$ verifier remains query-efficient.

Thus, we can extend \Cref{prop:qmah_fss} from $\mathsf{QMAH}$ oracle machines to $\mathsf{QMIP}^{\mathsf{PromiseQMIP}^{\mathsf{PromiseQMIP}^{\cdots}}}$ oracle machines. It follows, using the same proof as \Cref{cor:pp_not_in_qmah}, that $\mathsf{PP} \not\subset \mathsf{QMIP}^{\mathsf{PromiseQMIP}^{\mathsf{PromiseQMIP}^{\cdots}}}$ relative to a random oracle---despite the fact that $\mathsf{QMIP} = \mathsf{MIP}^{*}=\mathsf{RE}$ in the unrelativized world \cite{RUV13,JNVWY20}!

\section{Open Problems\label{OPEN}}

\subsection{Oracles where \texorpdfstring{$\mathsf{BQP} = \mathsf{EXP}$}{BQP = EXP}}

We construct oracles relative to which $\mathsf{BQP} = \mathsf{P}^{\mathsf{\# P}}$ and yet either $\mathsf{PH}$ is infinite (\Cref{thm:bqp=pp_ph_infinite}), or $\mathsf{P} = \mathsf{NP}$ (\Cref{thm:p=np_bqp=pp}). Can these be strengthened to oracles where we also have $\mathsf{BQP} = \mathsf{EXP}$? The main challenge in generalizing our proofs is that $\mathsf{P^{\# P}}$ machines, unlike $\mathsf{EXP}$ machines, have a polynomial upper bound on the length of the queries they can make. This property allowed us to encode the behavior of a $\mathsf{P^{\# P}}$ machine $M$ into a part of the oracle that $M$ cannot query, but that a $\mathsf{BQP}$ machine with a larger polynomial running time \textit{can} query. Alas, such a simple trick will not work when $M$ is an $\mathsf{EXP}$ machine. Nevertheless, there exist alternative tools that can collapse $\mathsf{EXP}$ to such weaker complexity classes. For instance, Heller \cite{Hel86} gives an oracle relative to which $\mathsf{BPP} = \mathsf{EXP}$. Beigel and Maciel \cite{BM99} even construct an oracle relative to which $\mathsf{P} = \mathsf{NP}$ and $\mathsf{\oplus P} = \mathsf{EXP}$, using $\mathsf{AC^0}$ circuit lower bounds for the $\PARITY$ problem that are analogous to the lower bounds we use for $\Forrelation$.

\subsection{Finer Control over \texorpdfstring{$\mathsf{BQP}$}{BQP} and \texorpdfstring{$\mathsf{PH}$}{PH}}
Recall \Cref{conj:sigma_k_in_bqp_sigma_k+1_not}, which states that for every $k$, there exists an oracle relative to which $\mathsf{\Sigma}_{k}^{\mathsf{P}}\subseteq\mathsf{BQP}$\ but $\mathsf{\Sigma}_{k+1}^{\mathsf{P}}\not \subset \mathsf{BQP}$. We conjecture more strongly that a small modification of the oracle $\mathcal{O}$ constructed in \Cref{thm:bqp=pp_ph_infinite} achieves this. Recall that $\mathcal{O}$ consists of a random oracle $A$, and an oracle $B$ that recursively hides the answers to all possible $\mathsf{P}^{\mathsf{\# P}^\mathcal{O}}$ queries in instances of the $\Forrelation$ problem. The idea is simply to modify the definition of $B$ so that it instead encodes the outputs of $\mathsf{\Sigma}_{k}^{\mathsf{P}}$ machines instead of $\mathsf{P}^{\mathsf{\# P}}$ machines.

Our intuition is that, because the $\Forrelation$ instances look random to $\mathsf{\Sigma}_{k}^{\mathsf{P}}$ machines, a $\mathsf{\Sigma}_{k}^{\mathsf{P}}$ machine should not be able to recursively reason about $B$. Thus, a $\mathsf{BQP}$ machine that queries $\mathcal{O} = (A, B)$ should be effectively no more powerful than a $\mathsf{BQP}^{\mathsf{\Sigma}_{k}^{\mathsf{P}}}$ machine that queries only $A$. If this intuition can be made precise, then one could possibly appeal to our proof that $\mathsf{\Sigma}_{k+1}^{\mathsf{P}}\not \subset \mathsf{BQP}^{\mathsf{\Sigma}_{k}^\mathsf{P}}$ relative to a random oracle. Of course, we could not get this proof strategy to work---otherwise, we would not have needed the machinery surrounding sensitivity concentration of $\mathsf{AC^0}$ circuits in order to get an oracle where $\mathsf{NP}^\mathsf{BQP} \not\subset \mathsf{BQP}^\mathsf{NP}$!

We now sketch what we consider a viable alternative approach towards showing that our conjectured oracle separation holds. Instead of the ``top-down'' view taken above, where one tries to argue that a $\mathsf{\Sigma}_{k}^{\mathsf{P}}$ machine gains no benefit from making recursive queries to $B$, one might instead attempt a ``bottom-up'' approach, where one uses the structure of the target $\mathsf{\Sigma}_{k+1}^\mathsf{P}$ problem (the $\Sipser_{k+2}$ function) to argue that each bit of $B$ has only minimal correlation with the answer, starting with the parts of $B$ that are constructed first. Very roughly speaking, our idea would be to combine the random projection technique of \cite{HRST17} with some generalization of the $\mathsf{AC^0}$ sensitivity concentration bounds that we prove in \Cref{sec:implications_of_ac0_concentration}.

In slightly more detail, we would first hit $A$ with a random projection, one that with high probability turns the $\Sipser_{k+2}$ function into an $\AND$ of large fan-in, while turning any $\mathsf{\Sigma}_{k}^{\mathsf{P}^A}$ machine into a low-depth decision tree. Then, we would want to argue that if we fix the unrestricted variables of $A$ to all $1$s, and choose $\Forrelation$ instances in $B$ consistent with this, then each bit of $B$ is unlikely to flip if we instead randomly change a few bits of $A$ to $0$s, and resample the $\Forrelation$ instances of $B$ corresponding to $\mathsf{\Sigma}_{k}^{\mathsf{P}}$ machines that return different answers. If this could be shown, then as in \Cref{thm:bqp_sigma_k_query_version}, an appeal to \Cref{lem:average_case_bbbv} (which is a modification of the BBBV Theorem \cite{BBBV97}) ought to be sufficient to argue that a $\mathsf{BQP}^\mathcal{O}$ machine could not compute the $\Sipser_{k+2}$ function.

For the bits of $B$ corresponding to the bottom-level $\mathsf{\Sigma}_{k}^{\mathsf{P}}$ machines that only query $A$ directly, this is easy to show, as a low-depth decision tree is unlikely to query any $0$s under a distribution of mostly $1$s. However, for the higher levels of $B$ corresponding to $\mathsf{\Sigma}_{k}^{\mathsf{P}}$ machines that can query the earlier bits of $B$, this becomes more challenging: we have to argue that a $\mathsf{\Sigma}_{k}^{\mathsf{P}}$ machine that queries a long list of $\Forrelation$ instances is unlikely to return a different answer when we randomly flip a few of the instances between the uniform and Forrelated distributions. This might require a generalization of \Cref{lem:ac0_single_forrelated_block_indistinguishable} in which (1) the string $x$ is not just uniformly random, but is an arbitrary sequence of Forrelated and uniformly random rows, and (2) instead of flipping a single random row of $x$ from uniformly random to Forrelated, we flip an arbitrary subset of the rows between random and Forrelated, subject only to the constraint that the probability of any individual row being chosen is small.

If this problem is too difficult, it remains interesting, in our view, to give an oracle where $\mathsf{NP} \subseteq \mathsf{BQP}$ but $\mathsf{PH} \not\subset \mathsf{BQP}$. This would merely require proving our proposed generalization of \Cref{lem:ac0_single_forrelated_block_indistinguishable} for low-width DNF formulas, as opposed to arbitrary $\mathsf{AC^0}$ circuits of quasipolynomial size.

\subsection{Stronger Random Restriction Lemmas}
Can one prove a sharper version of our random restriction lemma for $\mathsf{QMA}$ query algorithms (\Cref{thm:qma_random_restriction})? Unlike the switching lemma for DNF formulas (\Cref{lem:dnf_random_restriction}), our result has a quantitative dependence on the number of inputs $N$. Thus, whereas a $\polylog(N)$-width DNF simplifies (to a low-depth decision tree, with high probability) under a random restriction with $\Pr[*] = \frac{1}{\polylog(N)}$, we can only show that a $\polylog(N)$-query $\mathsf{QMA}$ algorithm simplifies under a random restriction with $\Pr[*] = \frac{1}{\sqrt{N}\polylog(N)}$, which leaves much fewer unrestricted variables. We see no reason why such a dependence on $N$ should be necessary, and we conjecture that a $\polylog(N)$-query $\mathsf{QMA}$ algorithm should simplify greatly under a random restriction with $\Pr[*] = \frac{1}{\polylog(N)}$. It would be interesting to see whether one could prove this even without a bound on the $\mathsf{QMA}$ witness length, as we do in our proofs.

It is also worth exploring whether our random restriction lemma could be generalized to other classes of functions. Our argument works for functions of low quantum query complexity, so it is natural to ask: is there a comparable random restriction lemma for bounded low-degree polynomials, and thus functions of low \textit{approximate degree}? Kabanets, Kane, and Lu \cite{KKL17} exhibit a random restriction lemma for \textit{polynomial threshold functions}, an even stronger class of functions, though their bounds become very weak when the degree is much larger than $\sqrt{\log N}$. We conjecture that an analogue of \Cref{thm:qma_random_restriction} should hold if we replace low $\mathsf{QMA}$ query complexity by low approximate degree, perhaps even with better quantitative parameters.\footnote{One could conceivably even show this by simply proving that \textit{every} partial function with low approximate degree also has low $\mathsf{QMA}$ query complexity, made easier by the fact that our definition of $\mathsf{QMA}$ query complexity allows for unbounded witness length. This is an easier version of the problem of showing whether approximate degree and quantum query complexity are polynomially related for all partial functions, which remains an open problem.}

\subsection{Collapsing \texorpdfstring{$\mathsf{QMAH}$}{QMAH} to \texorpdfstring{$\mathsf{P}$}{P}}
In \Cref{cor:pp_not_in_qmah}, we gave an oracle relative to which $\mathsf{PP} \not\subset \mathsf{QMAH}$ (indeed, we showed that this holds even for a random oracle). Can one generalize this to an oracle relative to which $\mathsf{P}=\mathsf{QMA}=\mathsf{QMAH}\neq\mathsf{PP}$? A priori, it might seem that one could use techniques similar to the ones we used in \Cref{thm:p=np_bqp=pp} to set $\mathsf{P} = \mathsf{NP}$ while still keeping $\mathsf{P} \neq \mathsf{P^{\# P}}$. That is, the idea would be to start with a random oracle $A$, then inductively construct an oracle $B$, recursively encoding into $B$ answers to all $\mathsf{QMA}$ machines that query earlier parts of $A$ and $B$. One would then hope to prove an analogue of \Cref{lem:recursive_np_ac0}, showing that the bits of $B$ can be computed by small low-depth circuits where the gates are functions of low $\mathsf{QMA}$ query complexity, and the inputs are in $A$. Finally, one could appeal to \Cref{cor:qmah_parity_disagreement} to argue that such a circuit cannot compute $\PARITY$.

The main issue is that $\mathsf{QMA}$ is a semantic complexity class, in contrast to $\mathsf{NP}$, which is a syntactic complexity class. This is to say that every $\mathsf{NP}$ machine defines a language, whereas a $\mathsf{QMA}$ machine only defines a promise problem. Hence, it is not clear how $B$ should answer on machines that fail to satisfy the $\mathsf{QMA}$ promise without ``leaking'' information that would otherwise be difficult to compute. Even if we, say, assign those bits of $B$ randomly, we can no longer argue that those bits are computable by a $\mathsf{QMA}$ query algorithm, which would break our idea for generalizing \Cref{lem:recursive_np_ac0}.

To illustrate the difficulty in constructing such an oracle, we describe an example of an oracle $\mathcal{O} = (A, B)$ that \textit{fails} to put $\mathsf{PP}$ outside $\mathsf{QMAH}$. We start by taking a random oracle $A$. Then, we inductively construct $B$, where each bit of $B$ encodes the behavior of a $\mathsf{QMA}^\mathcal{O}$ verifier $\langle M, x \rangle$, where $M$ can query the previously constructed parts of the oracle, as follows. We let $p \coloneqq \max_{\ket{\psi}} \Pr[M(x, \ket{\psi})] = 1$, and then we randomly choose the encoded bit to be $1$ with probability $p$ and $0$ with probability $1-p$. This is to say that we set the bit to $1$ with probability equaling the acceptance probability of the $\mathsf{QMA}$ verifier, maximized over all possible witness states $\ket{\psi}$.

Unfortunately, while one can easily show that $\mathsf{BPP}^\mathcal{O} = \mathsf{QMA}^\mathcal{O}$, $\mathcal{O}$ also allows an algorithm to ``pull the randomness out'' of a quantum algorithm, which makes $\mathcal{O}$ much more powerful than it seems! By padding $\langle M, x \rangle$ with extra bits, one can obtain from the oracle arbitrarily many independent bits sampled with bias $p$. Because $\mathsf{PH}^\mathcal{O} \subseteq \mathsf{BPP}^\mathcal{O}$, a $\mathsf{BPP}^\mathcal{O}$ machine can run Stockmeyer's algorithm \cite{Sto83} on these samples to obtain a multiplicative approximation of any such $p$. In particular, this implies that the quantum approximate counting problem, defined in \Cref{sec:results}, is in $\mathsf{BPP}^\mathcal{O}$. But the quantum approximate counting problem is $\mathsf{PP}^\mathcal{O}$-hard \cite{Kup15}, so we also have $\mathsf{BPP}^\mathcal{O} = \mathsf{PP}^\mathcal{O}$. Hence, any oracle that makes $\mathsf{P} = \mathsf{QMA} \neq \mathsf{PP}$ would have to choose a more careful encoding of the answers to $\mathsf{QMA}$ problems than the one described here.

\section{Acknowledgments}

We thank Lance Fortnow, Greg Kuperberg, Patrick Rall, and Avishay Tal for helpful conversations. We are especially grateful to Avishay Tal for providing us with a proof of \Cref{cor:ac0_block_sensitivity_tail_bound}. We further thank Chinmay Nirkhe for finding an error in the proof of \Cref{lem:sparse_aaronson_ambainis_states}.

\bibliographystyle{alphaurl}
\bibliography{MainBibliographyOld}

\end{document}